\documentclass[10pt]{article}

\usepackage{bm}
\usepackage{tikz}
\usepackage{scalefnt}
\usetikzlibrary{arrows,calc,scopes}
\usetikzlibrary{shapes.geometric}
\usetikzlibrary{decorations.pathreplacing,decorations.markings,decorations.pathmorphing}
\usetikzlibrary{plotmarks}
\usetikzlibrary{matrix}
\usepackage{color}
\usepackage{amsthm,amsfonts,amsmath,amssymb}
\usepackage{pxfonts}
\usepackage{euscript}
\usepackage{hyperref}

\usepackage{xypic}
\usepackage{graphicx}
\usepackage{epstopdf} 
\usepackage{enumerate}
\usepackage{bbold}

\usepackage[symbol]{footmisc}
\usepackage{appendix}

\newcommand\void[1]       {}

\newtheorem{thm}{Theorem}
\newtheorem{prp}[thm]{Proposition}
\newtheorem{crl}[thm]{Corollary}
\newtheorem{lem}[thm]{Lemma}

\newtheorem{prop-ph}[thm]{Proposition$^{\mathrm{ph}}$}
\newtheorem{cor-ph}[thm]{Corollary$^{\mathrm{ph}}$}
\newtheorem{lemma-ph}[thm]{Lemma$^{\mathrm{ph}}$}
\newtheorem{thm-ph}[thm]{Theorem}

\theoremstyle{definition}
\newtheorem{dfn}[thm]{Definition}
\newtheorem{expl}[thm]{Example}
\newtheorem{notation}[thm]{Notation}
\newtheorem{defn-ph}[thm]{Definition$^{\mathrm{ph}}$}

\theoremstyle{remark}
\newtheorem{rmk}[thm]{Remark}

\numberwithin{equation}{section}
\numberwithin{thm}{section}

\newcommand\be            {\begin{equation}}
\newcommand\ee            {\end{equation}}
\newcommand\bea           {\begin{eqnarray}}
\newcommand\eea         {\end{eqnarray}}
\newcommand\bnu          {\begin{enumerate}}
\newcommand\enu          {\end{enumerate}}
\newcommand{\bmm}{\begin{matrix}}
\newcommand{\emm}{\end{matrix}}
\newcommand{\pf}{\begin{proof}}
\newcommand{\epf}{\end{proof}}

\newcommand\id            {\mathrm{id}}

\newcommand\rVec    {\mathrm{Vec}}
\newcommand\rep     {\mathrm{Rep}}

\newcommand\one    {\mathbb{1}}

\newcommand\bC            {\mathbb{C}}
\newcommand\bN            {\mathbb{N}}

\newcommand\bZ            {\mathbb{Z}}
\newcommand\Zb            {\mathbb{Z}}

\newcommand\CC           {\EuScript{C}}
\newcommand\CD           {\EuScript{D}}

\newcommand\CH         {\EuScript{H}}

\newcommand\CM          {\EuScript{M}}
\newcommand\CN         {\EuScript{N}}

\newcommand\CP         {\EuScript{P}}

\newcommand\CR         {\EuScript{R}}

\newcommand\CX         {\EuScript{X}}

\newcommand\fM           {\mathfrak{M}}
\newcommand{\fU}          {\mathfrak{U}}
\newcommand\fZ           {\text{\usefont{U}{euf}{m}{n}Z}}

\newcommand\cO            {\mathcal{O}}
\newcommand\cX            {\mathcal{X}}

\newcommand\bk    {\mathbb{k}}
\newcommand\bfm   {\mathbf{m}}
\newcommand\bfe   {\mathbf{e}}

\DeclareMathAlphabet{\mathcal}{OMS}{cmsy}{m}{n} 

\newcommand{\BLvert}{\Biggl\vert\bmm\vspace{-3pt}\scalefont{0.6}} 
\newcommand{\Brangle}{\emm\Biggr\rangle}

\newcommand{\Psix}[3][1]{
\begin{tikzpicture}[scale=0.8]
\node[name=s, regular polygon, regular polygon sides=6, minimum size=1cm, outer sep=0pt ,draw] at (0,0) {}; 
\foreach \anchor/\x/\y /\xx/\yy /\b in
{corner 1/0.17/0.17*1.732/-0.11/0.18/1, corner 2/-0.17/0.17*1.732/0.07/0.18/2, corner 3/-0.34/0/-0.15/-0.18/3, corner 4/-0.17/-0.17*1.732/-0.22/-0.05/4, corner 5/0.17/-0.17*1.732/0.2/-0.05/5, corner 6/0.34/0/0.15/-0.18/6}
{
 \draw[shift=(s.\anchor)] (0,0) -- (\x,\y) node at(\xx,\yy) {$#2_{\text{\scalebox{0.7}{$\b$}}}$};
 \ifnum #1=1
 \draw[shift=(s.\anchor),<-,>=stealth', line width=0.01pt] (s.\anchor) -- (\x,\y);
 \fi
 }
\foreach \anchor/\xx/\yy /\a in
{side 1/0/-0.18/1, side 2/-0.18/0.05/2, side 3/0.15/0.05/3, side 4/0/-0.18/4, side 5/-0.18/0.05/5, side 6/0.15/0.05/6}
 \draw[shift=(s.\anchor)]  node at(\xx,\yy) {$#3_{\text{\scalebox{0.7}{$\a$}}}$};
\ifnum #1=1{
  \foreach \anchorr/\anchorf in
   {corner 1/corner 2, corner 2/corner 3, corner 3/corner 4, corner 4/corner 5, corner 5/corner 6, corner 6/corner 1}
   \draw[shift=(s.\anchorr), ->, >=stealth', line width=0.01pt]  (s.\anchorr) -- (s.\anchorf);}
 \else {
  \foreach \anchorb/\anchorw in
   {corner 1/corner 2, corner 3/corner 4, corner 5/corner 6} {
   \node[fill=black, circle, minimum size=2.5, inner sep=0, outer sep=0, draw] at(s.\anchorb) {};
   \node[fill=white, circle, minimum size=2.5, inner sep=0, outer sep=0, draw] at(s.\anchorw) {};}
}
\fi
\end{tikzpicture}
}

\newcommand{\Ygraph}[4][1]{
\begin{tikzpicture}[scale=0.6,baseline]
  \draw [->,>=stealth',line width=0.01pt] (30:0.1) -- (0,0) ; 
    \draw (30:1) -- (0,0) ; 
  \draw [->,>=stealth',line width=0.01pt] (150:0.1) -- (0,0); 
    \draw (150:1) -- (0,0); 
  \node at(-0.5,0.01) {$#2$};
  \node at(0.5,0.01) {$#4$};
  \ifnum #1=1 {
   \draw [->,>=stealth',line width=0.01pt] (0,-1/20) -- (0,0); 
   \node at(-0.3,-0.5) {$#3$};
  }
  \else{
    \draw [<-,>=stealth',line width=0.01pt] (0,-1/2) -- (0,0); 
    \node at(-0.4,-0.5) {$#3$};
    }
  \fi
    \draw (0,-1) -- (0,0); 
  \end{tikzpicture}
  }

\makeatletter
\newcommand\xleftrightarrow[2][]{%
  \ext@arrow 9999{\longleftrightarrowfill@}{#1}{#2}}
\newcommand\longleftrightarrowfill@{%
  \arrowfill@\leftarrow\relbar\rightarrow}
\makeatother

\tikzset{middlearrow/.style={
        decoration={markings,
            mark= at position 0.5 with {\arrow{#1}} ,
        },
        postaction={decorate}
    }
}

\newcommand{\Cross}{$\mathbin{\tikz [x=1.4ex,y=1.4ex,line width=.2ex, red] \draw (0,0) -- (1,1) (0,1) -- (1,0);}$}%
\newcommand{\Checkmark}{$\color{teal}\checkmark$}

\newcommand{\Se}[1]{$S_{\text{bdy}}^{\mathbf{e}}(
  \begin{tikzpicture}[x=7.2mm,y=4.16628mm,baseline]

  \tikzset{
    pointb/.style={circle,inner sep=0pt,minimum size=0.5mm,color=black,fill=black,line width=0.5mm,draw
    }
  }

  \tikzset{
    pointw/.style={circle,inner sep=0pt,minimum size=0.5mm,color=white,fill=white,line width=0.5mm,draw
    }
  }

  \tikzset{->-/.style={decoration={markings,mark=at position #1 with {\arrow{stealth}}},postaction={decorate}}}

\fill[top color=white,bottom color=cyan!20,draw=white] (0.9,-0.2) rectangle   (2.1,0.6);

\draw[line width=0.08cm, gray!50] (1,0)--(1.5,0.5)--(2,0);

\draw[purple!50,line width=0.09cm] plot[smooth,tension=0.9] coordinates {(0.95,-0.1)(1.5,0.5) (2.05,-0.1)};

\foreach \i in {2,...,3} {
         \node[pointw] at (1*\i-1,0) {};};

     \node[pointb] at (1.5,0.5) {};
\end{tikzpicture}#1)$}

\newcommand{\Sm}[1]{$S_{\text{bdy}}^{\mathbf{m}}(
\begin{tikzpicture}[x=7.2mm,y=4.16628mm,baseline]

    \tikzset{
      pointb/.style={circle,inner sep=0pt,minimum size=0.5mm,color=black,fill=black,line width=0.5mm,draw
      }
    }
  
    \tikzset{
      pointw/.style={circle,inner sep=0pt,minimum size=0.5mm,color=white,fill=white,line width=0.5mm,draw
      }
    }

    \tikzset{->-/.style={decoration={markings,mark=at position #1 with {\arrow{stealth}}},postaction={decorate}}}

\fill[top color=white,bottom color=cyan!20,draw=white] (1.4,-.2) rectangle   (2.6,0.6);

 \draw[line width=0.08cm, gray!50] (1.5,0.5)--(2,0)--(2.5,0.5);
 \foreach \i in {2,...,3} {
           \node[pointb] at (1*\i-0.5,0.5) {};};
  \node[pointw] at (2,0) {};

  \draw[cyan!50,line width=0.09cm] plot[smooth,tension=0.9] coordinates { (1.6,-0.1) (2,0.5)(2.4,-0.1)};

  \end{tikzpicture}#1)$}

\begin{document}

\begin{center} \LARGE
The boundary phase transitions of the 2+1D $\Zb_N$ topological order via topological Wick rotation
\\
\end{center}

\vskip 2em
\begin{center}
{Yalei Lu$^{a,c,e,}$
\footnote[2]{luleisure@gmail.com}
Holiverse Yang$^{b,c,d}$\footnote[3]{hpyang1996@163.com},
}
\\[1.8em]
$^a$ Key Laboratory for Magnetism and Magnetic Materials of MOE, Lanzhou University, 730000 Lanzhou, China.\\
$^b$ Shenzhen Institute for Quantum Science and Engineering, Southern University of Science and Technology, Shenzhen 518055, China. \\
$^c$ International Quantum Academy, Shenzhen 518048, China.\\
$^d$ Guangdong Provincial Key Laboratory of Quantum Science and Engineering, Southern University of Science and Technology, Shenzhen 518055, China.\\
$^e$ Department of Physics, \\
Southern University of Science and Technology, Shenzhen 518055, China
\\[0.8em]

\end{center}

\vskip 2.5em

\begin{abstract}

In this work, we show 
that a critical point of a 1d self-dual boundary phase transition between two gapped boundaries of the $\mathbb{Z}_N$ topological order can be described by a mathematical structure called an enriched fusion category. 
The critical point of a boundary phase transition can be viewed as a gappable non-chiral gapless boundary of the $\mathbb{Z}_N$ topological order. A mathematical theory of the gapless boundaries of 2d topological orders developed by Kong and Zheng (\href{https://arxiv.org/abs/1905.04924}{arXiv:1905.04924} and \href{https://arxiv.org/abs/1912.01760}{arXiv:1912.01760}) 
tells us that all macroscopic observables on the gapless boundary form an enriched unitary fusion category, which can be obtained by a holographic principle called the ``topological Wick rotation." Using this method, we obtain the enriched fusion category that describes a critical point of the phase transition between the $\bfe$-condensed boundary and the $\bfm$-condensed boundary of the $\bZ_N$ topological order. To verify this idea, we also construct a lattice model to realize the critical point and recover the mathematical data of this enriched fusion category. The construction further shows that the categorical symmetry of the boundary is determined by the topological defects in the bulk, which indicates the holographic principle indirectly. This work shows, as a concrete example, that the mathematical theory of the gapless boundaries of 2+1D topological orders is a powerful tool to study general phase transitions.

\end{abstract}

\tableofcontents

\section{Introduction}
The study of topological orders (TOs) has attracted much attention in condensed matter physics and mathematical physics, see review \cite{Wen2017,Wen2019}. 
The existence of topological order appears to indicate that phases and phase transitions in nature are much richer than Landau's symmetry breaking theory. 
Therefore, the new topological phases and phase transitions require a new mathematical description beyond group theory.
In this work, we provide a concrete example to show that all physical observables of a topological phase transition that occurs on the boundary of 2d\footnote[1]{we use $nd$ to represent $n$ spatial dimension and $(n+1)D$ to represent $n+1$ spacetime dimension.} $\Zb_N$ TO form a new mathematical structure called an \textit{enriched fusion category} \cite{MP2019,KZ2018}. 

A critical point of a 1+1D topological phase transition can be viewed as a gappable CFT-type gapless phase. 
If such a phase transition occurs at the boundary of a non-trivial 2d topological order, then it should be a gappable gapless boundary of the TO.
So, a 1+1D purely boundary phase transition should be nothing but a gappable non-chiral gapless boundary of the 2d anomaly-free topological order.
In other words, this 1d gapless phase cannot be realized by an 1d local Hamiltonian lattice model, which means that this 1d gappable non-chiral gapless phase is \textit{anomalous} \cite{wen2013classifying,kong2014braided,jiwen2021unified}.
Recently, many works reveal a new type of holographic dualities between 2+1D topological orders with a gapped boundary and 1+1D (potentially gapless) quantum liquids in different contexts \cite{AMF2016,KZ2018,AFM2020,KZ2020,JW2020,KZ2021,WJX2021,LDOV2021,KWZ2022,chatterjee2022algebra,chatterjee2022,XZ2022,MMT2022,LJ2022,BCDP2022}. 
The story about purely boundary phase transition belongs to a generalized case: the duality between gapped domain walls of two 2+1D TOs and anomalous gapless boundaries of one of the TOs \cite{KZ2021}. It should be stressed here that the gapless phase in our paper has a non-trival bulk (TO), but the one in the former type of holographic dualities doesn't. This is clarified in Section \ref{sec:boundary}.

A mathematical theory of a gapless boundary of a 2d anomaly-free TO was developed in \cite{KZ2018,KZ2020,KZ2021}.
According to this theory, all long-wave-length-limit physical observables on a gapped/gapless boundary form an enriched fusion category.
Therefore, the mathematical structure of a critical point of a boundary topological phase transition is an enriched fusion category.
\begin{rmk}
    To be more precise, the mathematical structure on a 1d gapped/gapless boundary can be split into two parts: a \textit{local quantum symmetry} $U$ which encodes the information of local observables; and a $\mathrm{Mod}_U$-enriched fusion category  ${}^{\mathrm{Mod}_U}\CM$ which encodes all the topological defects. The enriched fusion category  ${}^{\mathrm{Mod}_U}\CM$ is called the $topological$ $skeleton$ \cite{KZ2022liquids}.
\end{rmk}

The first concrete example can be found in \cite{CJKYZ2020}, whose authors use the enriched fusion category to mathematically describe a critical point of a boundary phase transition of the $\Zb_2$ TO (toric code). 
In this work, we demonstrate this idea in a more general case: 
A self-dual boundary phase transition between two gapped boundaries $\rep(\bZ_N)$ ($\bfm$-condensed) and $\rVec_{\bZ_N}$ ($\bfe$-condensed) of the $\bZ_N$ topological order can be mathematically described by an enriched fusion category. 
Here, the word ``self-dual" means that the critical point has an emergent $\Zb_2$ symmetry from the $\bfe$-$\bfm$ duality.
And we construct a lattice model to realize this gappable gapless boundary phase, including all data from the corresponding enriched fusion category. 
The process of the construction magically shows that the local operators of the boundary are determined by the string operators of the bulk. 
It is natural because the anyon traveling to the boundary can cause the fluctuation of the boundary \cite{CJKYZ2020,JW2020,JW2020metallic,chatterjee2022}.

In Section \ref{sec:preminaliers}, we briefly review some mathematical basics of a gapped 2d anomaly-free TO, including the categorical description of a 2d anomaly-free TO by a unitary modular tensor category (UMTC), and the anyon condensation theory of 2d anomaly-free TO \cite{Kong2014}.
We also give two examples of unitary modular tensor categories: the $\bZ_N$ parafermion UMTC $\mathrm{PF}_N$ and the $\bZ_N$ quantum double UMTC $\fZ_1(\rep(\bZ_N))$, both of which are related to this work.
In Section \ref{sec:gappable_boundary}, we first review the mathematical description of gapless boundaries introduced in \cite{KZ2020,KZ2021}, including a refresher of the rational conformal field theory (RCFT) and the ingredients of an enriched fusion category.
Then, we find a gapped domain wall between the double $\bZ_N$ parafermion TO and the $\bZ_N$ quantum double TO by anyon condensation theory.
Finally, by a trick called the \textit{topological Wick rotation} \cite{KZ2020}, this domain wall can be holographically dual to a gappable non-chiral gapless boundary:
\begin{align}
    (V_{\mathrm{PF}_N}\otimes_{\bC}\overline{V_{\mathrm{PF}_N}},{}^{\fZ_1(\mathrm{PF}_N)}\fZ_1(\mathrm{PF}_N)_B)
\end{align}
of the $\bZ_N$ quantum double TO $\fZ_1(\rep(\bZ_N))$, where 
\begin{itemize}
    \item $V_{\mathrm{PF}}$ is the $\bZ_N$ parafermion vertex operator algebra (VOA) whose module category is $\mathrm{PF}_N$.
    $V_{\mathrm{PF}_N}\otimes_{\bC}\overline{V_{\mathrm{PF}_N}}$ is the full field algebra (FFA) that describes the local quantum symmetry on the 1+1D world sheet of this gapless boundary.
    \item $B=\bigoplus_{u=0}^{[\frac{N}{2}]}\fM_{2u,0}\boxtimes \fM_{2u,0}$ is a condensable algebra such that the condensed phase of the double $\bZ_N$ parafermion TO via $B$ is the $\bZ_N$ quantum double TO.
    \item ${}^{\fZ_1(\mathrm{PF}_N)}\fZ_1(\mathrm{PF}_N)_B$ is the enriched fusion category that describes the topological skeleton of the gappable non-chiral gapless boundary, where 
    \begin{itemize}
        \item $\fZ_1(\mathrm{PF}_N)_B$ is the category of right $B$-modules in $\fZ_1(\mathrm{PF}_N)$. 
        In physics, the boundary topological excitations (or boundary conditions) on this gapless boundary are described by these right $B$-modules.
        \item $\fZ_1(\mathrm{PF}_N)$ contains the internal homs $M_{x,y}$ between boundary topological excitations, or equivalently, the domain walls between two boundary CFT's with boundary conditions $x$ and $y$.
    \end{itemize}
\end{itemize}
We also compute the partition functions of $M_{x,y}$ exactly. It is expected that the enriched fusion category  $(V_{\mathrm{PF}_N}\otimes_{\bC}\overline{V_{\mathrm{PF}_N}},{}^{\fZ_1(\mathrm{PF}_N)}\fZ_1(\mathrm{PF}_N)_B)$ indeed describes a critical point of a self-dual boundary phase transition of $\bZ_N$ TO. 
Therefore, in Section \ref{sec:lattice}, we provide a lattice model realization of a critical point of this boundary phase transition to recover all the ingredients in ${}^{\fZ_1(\mathrm{PF}_N)}\fZ_1(\mathrm{PF}_N)_B$. 
The low-energy effective theory of this critical model is described by a set of modular covariant partition functions. They are all agree with the mathematical result. Finally, some examples for $\bZ_2$, $\bZ_3$, and $\bZ_4$ cases are listed. 
The results in \cite{CJKYZ2020} can be regarded as a special case $\bZ_2$ of our work. 
  
This work shows that the categorical language can effectively describe the new-type phase transitions beyond the Laudau picture, and a holographic principle known as the ``topological Wick rotation" plays a key role in the process. We believe the mathematical theory developed in \cite{KZ2018,KZ2020,KZ2021} is powerful and provides a potential tool to study general phase transitions.

\section{Categorical Preliminaries}\label{sec:preminaliers}
Our goal in this paper is to find a gapless boundary that can describe a critical point of a self-dual topological phase transition on the boundary of the 2d $\bZ_N$ topological order.
Before we discuss an anomalous 1d gapless phase, it is necessary to review some mathematical basics of 2d gapped TOs, including a categorical description of a 2d anomaly-free TO and the anyon condensation theory. 
In particular, the anyon condensation theory, which provides a description of the domain walls between two TOs, plays an important role in the construction of the gapless boundary we need.

In this section, we briefly review the notion of a unitary modular tensor category, which describes the particle-like excitations of a 2d anomaly-free TO and the anyon condensation theory.

\subsection{Unitary modular tensor categories}\label{sec:UMTC}
A gapped quantum liquid without symmetry is called a topological order (TO) \cite{Wen1990,ZW2015,SM2016}.
It is known that a 2d anomaly-free TO can be described by a unitary modular tensor category (UMTC) $\CC$ with a central charge $c$ \cite{Kitaev2006,KZ2022}.
Therefore, we denote a TO by a pair $(\CC,c)$.
For the sake of brevity, we won't give an explicitly mathematical definition of a UMTC here. 
Instead, we introduce the ingredients and some properties of a UMTC $\CC$.
\begin{itemize}
    \item $\CC$ has finitely many simple objects/$anyons$ $x,y,z,\dots$.
    We denote $\mathrm{Irr}(\CC)$ the set of isomorphism classes of simple objects.
    Each composite anyon can be written as a direct sum of simple anyons, i.e. $a=x\oplus z\oplus\cdots$.

    \item For two anyons $x,y$, there is a finite-dimensional Hilbert space $\hom_{\CC}(x,y)$ called the hom space of $x$ and $y$.
    An element in a hom space is a $0D$ defect on time axis, so we call it an \textit{instanton}.
    For each triple $x$, $y$ and $z$, there is a map
    \begin{align*}
        \circ:\hom_{\CC}(y,z)\otimes_{\bC}\hom_{\CC}(x,y)\to\hom_{\CC}(x,z),
    \end{align*}
    called the composition map, which describes the fusion of instantons along the time axis.
    The composition map is associative.

    \item There is a tensor product functor $\otimes:\CC\times\CC\to\CC$, that describes the fusion of two anyons $x,y$ into the anyon $x\otimes y$.
    The fusion of anyons should be associative, so that for all $x,y,z\in\CC$ there exists an isomorphism $\alpha_{x,y,z}:x\otimes (y\otimes z)\to (x\otimes y)\otimes z$ satisfying the necessary coherence conditions.

    \item There is a distinguished simple object $\mathbb{1}$ called the tensor unit (or the trivial anyon).

    \item Each object $x\in\CC$ has a dual $x^*$, together with the creation and annihilation morphisms $v_x:x^*\otimes x\to \one$ and $u_x:\one\to x\otimes x^*$ satisfying some coherence conditions.
    
    \item $\CC$ has a braiding structure, i.e. for all $x,y\in \CC$ there is an isomorphism $c_{x,y}:x\otimes y\to y\otimes x$  satisfying Yang-Baxter equations.
    The braiding is non-degenerate; that is, the $S$-matrix, whose entry is the trace of double braiding of simple objects, is non-degenerate.
    
    \item Each object $x\in\CC$ has a twist (topological spin), which is an isomorphism $\theta_x:x\to x$ representing the self statistics of anyon.
    
    \item $\CC$ has a unitary structure: for each morphism $f:x\to y$, there is an adjoint morphism $f^{\dagger}:y\to x$ such that $(g\otimes h)^{\dagger}=g^{\dagger}\otimes h^{\dagger}$ for any $g:v\to w$ and $h:x\to y$.
    And the coherence data $\alpha,l,r,c$ are both unitary.
\end{itemize} 

The simplest example of UMTCs is the category of finite-dimensional Hilbert space denoted by $\mathbf{H}$.
It has a unique simple object, the one-dimensional Hilbert space $\bC$.
In particular, the pair $(\mathbf{H},0)$ describes the trivial 2d TO.

\begin{rmk}
    For readers who are interested in the bootstrap analysis that an anomaly-free 2d TO can be described by a UMTC, we recommend \cite{KZ2022}.
    For readers who are interested in the rigorous mathematical definition of a UMTC, we recommend \cite{Muger2000,EGNO2016}.
\end{rmk}

\subsubsection{UMTC \texorpdfstring{$\fZ_1(\rep(\bZ_N))$}{Z(rep(Z N))}}\label{sec:ZrepZN}
Consider the category $\rep(\bZ_N)$ of finite-dimensional representations of the cylic group $\bZ_N$. 
The Drinfeld center $\fZ_1(\rep(\bZ_N))$ of the $\rep(\bZ_N)$ is a UMTC \cite{BK2001}, which describes the particle-like topological excitations of the $\bZ_N$ quantum double model \cite{Kitaev2003}. We list some of its ingredients below.

\begin{itemize}
    \item The simple objects are  $\cO_{\alpha,\beta}$ labeled by $\alpha,\beta$ with $0\leq \alpha,\beta<N$.
    In physics, $\bfe:=\cO_{1,0}$ denotes the elementary charge and $\bfm:=\cO_{0,1}$ denotes the elementary flux.
    Hence, we use $\mathbf{e}^{\alpha}\mathbf{m}^{\beta}$ to denote the simple object $\cO_{\alpha,\beta}$ later.

    \item The fusion rule of two simple objects $\cO_{\alpha_1,\beta_1}$ and $\cO_{\alpha_2,\beta_2}$ is 
    \begin{align}
    \cO_{\alpha_1,\beta_1}\otimes \cO_{\alpha_2,\beta_2}\simeq \cO_{\alpha_1+\alpha_2,\beta_1+\beta_2}.
    \end{align}
    The tensor unit is $\one:=\cO_{0,0}$.
    \item The braiding of $\cO_{\alpha_1,\beta_1}$ and $\cO_{\alpha_2,\beta_2}$ is given by:
    \begin{align*}
        c_{\alpha_1,\beta_1;\alpha_2,\beta_2}:\cO_{\alpha_1,\beta_1}\otimes \cO_{\alpha_2,\beta_2}\xrightarrow{\mathrm{e}^{\frac{2\pi\mathrm{i}}{N}\alpha_1\beta_2}}\cO_{\alpha_2,\beta_2}\otimes \cO_{\alpha_1,\beta_1}.
    \end{align*}

    \item The twist for simple object $\cO_{\alpha,\beta}$ is given by:
    $\theta_{\alpha,\beta}=\mathrm{e}^{\frac{2\pi \mathrm{i}}{N}\alpha\beta}$ .
\end{itemize}


The $\bZ_N$ quantum double TO $(\fZ_1(\rep(\bZ_N)),0)$ has many gapped boundaries which are classified by the subgroups of $\bZ_N$ \cite{beigi2011}.
The unitary fusion category (UFC) $\rep(\bZ_N)$ describes the 1d $\bfm$-condensed boundary, which corresponds to the subgroup $\bZ_N$ itself.
And the 1d $\bfe$-condensed boundary is $\rVec_{\bZ_N}$, the category of finite-dimensional $\bZ_N$-graded vector spaces, which corresponds to the trivial subgroup $\{e\}$.

\begin{expl}
    When $N=2$, the UMTC $\fZ_1(\rep(\bZ_2))$ describes the toric code model \cite{Kitaev2003}.
    It has 4 simple objects denoted by $\one$, $\bfe$, $\bfm$ and $\mathbf{f}$ with the fusion rules:
    \begin{align}
        \bfe\otimes\bfe\simeq \one\simeq\bfm\otimes\bfm,\quad\quad \bfe\otimes\bfm\simeq \mathbf{f}.
    \end{align}
    It has two gapped boundaries, $\bfm$-condensed boundary $\rep(\bZ_2)$ and $\bfe$-condensed boundary $\rVec_{\bZ_2}$.
\end{expl}

In this paper, we construct a gapless boundary that describes a critical point of a topological phase transition between $\rep(\bZ_N)$ and $\rVec_{\bZ_N}$.

\begin{rmk}
    As fusion categories, $\rep(\bZ_N)$ and $\rVec_{\bZ_N}$ are equivalent.
    But in general, when $G$ is a finite group, $\rep(G)$ is not equivalent to $\rVec_{G}$.
    For example, consider the permutation group $S_3$, $\rep(S_3)$ has 3 simple objects, but $\rVec_{S_3}$ has 6 simple objects.
    The equivalence holds when $G$ is abelian.
\end{rmk}

\begin{rmk}
    In Landau's paradigm, the topological phase transition between two specific 1d gapped boundaries $\rep(\bZ_N)$ and $\rVec_{\bZ_N}$ in this work can be viewed as a completely $spontaneous$ $symmetry$ $breaking$ process.
\end{rmk}

\subsubsection{UMTC \texorpdfstring{$\mathrm{PF}_N$}{PF N} and \texorpdfstring{$\fZ_1(\mathrm{PF}_N)$}{ZPF N}}
In this subsection, we first introduce the parafermion UMTC $\mathrm{PF}_N$.
It is the category of modules over a unitary rational vertex operator algebra $V_{\mathrm{PF}_N}$ called the parafermion VOA \cite{DW2011,DL2012,DLWY2010,ADJR2018}.
For the details of the parafermion VOA and the mathematical details of this UMTC, one can refer to Appendix.~\ref{sec:math_PFN}.
Here we list some ingredients and properties of this UMTC.
\begin{itemize}
    \item Simple objects are denoted by  $\mathfrak{M}_{\ell,m}$, with $0\leq \ell\leq N$,  $0\leq m < 2N$ and $\ell+m\equiv 0\mod 2$.
    There is an equivalence relation between simple objects: 
    \begin{align*}
        \fM_{\ell,m}&\sim \fM_{\ell',m'},\quad\mathrm{if}\quad \ell=N-\ell',m=N+m'.\\
    \end{align*}
    So, there are $N(N+1)/2$ inequivalent simple objects in $\mathrm{PF}_N$.
    
    \item The fusion rule of two simple objects $\fM_{\ell,m}$ and $\fM_{\ell',m'}$ is given by:
    \begin{align}
         \fM_{\ell,m}\otimes \fM_{\ell',m'}=\bigoplus\limits^{\min(\ell,\ell')}_{r=\max(\ell+\ell'-N,0)}\fM_{\ell+\ell'-2r,m+m'}.
    \end{align}
    \item The double braiding of two simple objects $\fM_{\ell,m}$ and $\fM_{\ell',m'}$ is given by:
    \begin{align}\label{eq:double_braiding_PF}
        c_{\ell' m',\ell m}\circ c_{\ell m,\ell' m'}=\bigoplus_{s=\max(\ell+\ell'-N,0)}^{\min(\ell,\ell')}\mathrm{e}^{2\pi \mathrm{i}(h_{\ell+\ell'-2s}-h_{\ell}-h_{\ell'})}\mathrm{e}^{-\frac{\pi\mathrm{i}}{N}mm'}\id_{\ell+\ell'-2s,m+m'},
    \end{align}
    where $h_t:=\frac{t(t+2)}{4(N+2)}$.
    \item The twist of the simple object $\fM_{\ell,m}$ is:
    \begin{align*}
        \theta_{\ell m}=\mathrm{e}^{2\pi\mathrm{i}(h_{\ell}-\frac{m^2}{4N})}.
    \end{align*}
\end{itemize}

The quantum dimension of $\mathrm{PF}_N$ is $\frac{N(N+2)}{4\sin^2(\frac{\pi}{N+2})}$.

\begin{expl}
    When $N=2$, we have $\mathbf{Is}\simeq \mathrm{PF_2}$, where $\mathbf{Is}$ is the module category over the Ising VOA consisting of the simple objects $\{\one,\sigma,\psi\}$ with the fusion rules:
    \begin{align*}
        \psi\otimes\psi\simeq \one, \quad\quad \sigma\otimes\sigma\simeq \one\oplus\psi.
    \end{align*}
    We have the following correspondence of the simple objects between these two categories:
    \begin{align*}
        \mathbb{1}=\fM_{0,0},\quad \sigma=\fM_{1,1},\quad \psi=\fM_{2,0}. 
    \end{align*}
    It is not hard to check by the above formulas that the fusion rules, braiding, and the twists of simple objects of $\mathrm{PF_2}$ coincide with those of $\mathbf{Is}$.
    
\end{expl}

Now we introduce the Drinfeld center $\fZ_1(\mathrm{PF}_N)$ of the parafermion UMTC which describes the double $\bZ_N$ parafermion TO.
Since $\mathrm{PF}_N$ is non-degenerate, we have $\fZ_1(\mathrm{PF}_N)\simeq \mathrm{PF}_N\boxtimes \overline{\mathrm{PF}_N}$ \cite{Muger2003subfactor}.
Here, the time reversal $\overline{\mathrm{PF}_N}$ has the same objects and fusion rules with $\mathrm{PF}_N$ but has a different braiding defined by $\tilde{c}_{x,y}:=c^{-1}_{y,x}$ for $x,y\in\fZ_1(\mathrm{PF}_N)$.
And $\boxtimes$ is the Deligne tensor product of $\bC$-linear categories whose physical significance is stacking of 2d TOs.
We list some ingredients of $\fZ_1(\mathrm{PF}_N)$ bellow.
\begin{itemize}
    \item Simple objects in $\mathrm{PF}_N\boxtimes \overline{\mathrm{PF}_N}$ are $\fM_{\ell,m}\boxtimes \overline{\fM_{\ell',m'}}$ where $\fM_{\ell,m}$ is a simple object in $\mathrm{PF}_N$ and $\overline{\fM_{\ell',m'}}$ is a simple object in $\overline{\mathrm{PF}_N}$.
    Recall that $\mathrm{PF}_N$ has $N(N+1)/2$ inequivalent simple objects, hence there are $N^2(N+1)^2/4$ inequivalent simple objects in $\fZ_1(\mathrm{PF}_N)$.
    
    \item For two simple objects $\fM_{\ell_1,m_1}\boxtimes \overline{\fM_{\ell_1',m_1'}}$ and $\fM_{\ell_2,m_2}\boxtimes \overline{\fM_{\ell_2',m_2'}}$, the fusion rule is given by 
    \begin{align*}
        (\fM_{\ell_1,m_1}\boxtimes \overline{\fM_{\ell_1',m_1'}}) \otimes_{\fZ_1(\mathrm{PF}_N)} (\fM_{\ell_2,m_2}\boxtimes \overline{\fM_{\ell_2',m_2'}})=(\fM_{\ell_1,m_1}\otimes \fM_{\ell_2,m_2})\boxtimes (\overline{\fM_{\ell_1',m_1'}}\otimes \overline{\fM_{\ell_2',m_2'}}).
    \end{align*}
\end{itemize}

\begin{expl}\label{eg:double_ising}
    When $N=2$, $\fZ_1(\mathrm{PF}_2)=\fZ_1(\mathbf{Is})=\mathbf{Is}\boxtimes\overline{\mathbf{Is}}$, where $\fZ_1(\mathbf{Is})$ is the double Ising category that describes the double Ising TO.
    $\fZ_1(\mathbf{Is})$ has $9$ simple objects and consists of $x\boxtimes y$ for all $x\in\mathbf{Is}$ and $y\in\overline{\mathbf{Is}}$.
    For instance, $\one \boxtimes \one$ and $\psi\boxtimes \psi$  are both simple objects of $\fZ_1(\mathbf{Is})$.
\end{expl}

\subsection{Anyon condensation theory}\label{sec:any_cond}
For topological order, a phase transition phenomenon called anyon condensation may occur \cite{BSS2002,BSS2003,BS2009,BW2010}.
After anyon condensation, the Hilbert space of the condense phase are the energy-favorable subspace of the original Hilbert space.
Hence, there are fewer topological defects and the system is still gapped. Also, there should be a domain wall between the original phase and the condensed phase.
Some defects can move across the domain wall and others will be confined to it. The anyon condensation has a systematical mathematic description developed in \cite{Kong2014}.
This theory tells us how to condense bosonic anyons in a known TO and obtain a new TO.
In addition, we have an explicit description of the domain wall between these two TOs.
Domain walls between TOs play an important role in the description of gapless boundaries, which we will see in the following section.

\begin{figure}
    \begin{center}
        \begin{tikzpicture}[scale=1.1]
            \draw [fill=gray!30, draw=white] (-3,0) rectangle (0,3);
            \draw [fill=gray!15, draw=white] (0,0) rectangle (3,3);
            \draw[line width=0.5mm] (0,0)--(0,3);
            \node  at (-1.5,1.5){$\CC$};
            \node  at (1.5,1.5){$\CC_A^{loc}$};
            \node [below]at(0,0){$\CC_A$};
            \draw [fill] (-1,2) circle[radius=0.04] node[above] {$a$};
            \draw [->] (-1,2)--(0.9,2);
            \draw [fill] (1,2) circle [radius=0.04] node[above] {$x$};
            \draw [fill] (-1,1) circle[radius=0.04] node[below] {$b$};
            \draw [fill] (0,1) circle[radius=0.04] node[right] {$m$};
            \draw [->, line join=round, decorate, decoration={zigzag,segment length=4,amplitude=1.1,post=lineto,post length=2pt}]  (-0.9,1) -- (-0.1,1);
        \end{tikzpicture}
    \end{center}  
    \caption[]{An illustration of anyon condensation process. $\CC$ is the original phase, and $\CC_A^{loc}$ is the condensed phase. The domain wall between these two phases is $\CC_A$. The anyon $a$ in $\CC$ can freely across the domain wall; it becomes the anyon $x$ in $\CC_A^{loc}$ after condensation. The anyon $b$ in $\CC$ cannot move across the domain wall and is confined to the domain wall; it becomes the anyon $y$ in $\CC_A$.}
    \label{fig:any_cond}
\end{figure}
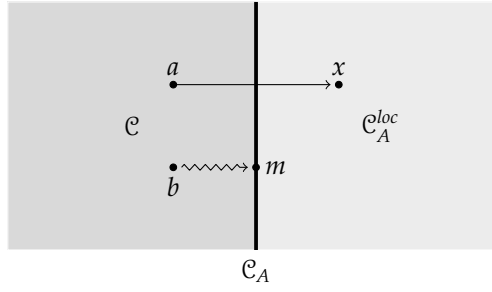

As shown in Fig.~\ref{fig:any_cond}, given a 2d anomaly-free TO described by a UMTC $\CC$\footnote{Here we omit the central charge $c$ because during the anyon condensation process, the central charge does not change}, we need to find a $condensable$ $algebra$ $A$ in $\CC$ that controls which anyons in $\CC$ can move across the domain wall freely and which anyons should be confined to the domain wall.
Anyons that can move freely across the wall  correspond to the so-called $local$ $A$-modules in $\CC$.
We denote the category of local $A$-modules in $\CC$ as $\CC_A^{loc}$. 
Mathematically, it has been proved that $\CC_A^{loc}$ is still a UMTC \cite{BEK2000,KO2002}.
Hence $\CC_A^{loc}$ describes a new 2d anomaly-free TO, which is the condensed phase.
Therefore, the local $A$-modules in $\CC$ are the topological defects in the condensed phase.
The 1d gapped domain wall between $\CC$ and $\CC_A^{loc}$ are described by a UFC $\CC_A$, which is the category of right $A$-modules.
For mathematical ingredients of the anyon condensation theory, please refer to Appendix.~\ref{sec:any_cond_maths}.
For the bootstrap analysis of condensed phase and gapped domain wall, please see \cite{Kong2014}.

\begin{expl}
In double Ising TO $(\fZ_1(\mathbf{Is}),0)$, we can find three different condensable algebras.
\bnu
    \item The trivial condensable algebra is $B_0:=\one\boxtimes\one$. 
    The condensed phase obtained by condensing $B_0$ is still the double Ising TO.
    
    \item We can find a condensable algebra $B=(\mathbb{1}\boxtimes \mathbb{1})\oplus (\psi\boxtimes \psi)=(\fM_{0,0}\boxtimes \fM_{0,0})\oplus (\fM_{2,0}\boxtimes \fM_{2,0})$  such that $\fZ_1(\mathbf{Is})_B^{loc}\simeq \fZ_1(\rep(\bZ_2))$ \cite{CJKYZ2020}.
    
    \item The maximal condensable algebra is $B_{max}:=(1\boxtimes 1)\oplus(\sigma\boxtimes \sigma)\oplus (\psi\boxtimes \psi)$.
    The condensed phase $\fZ_1(\mathbf{Is})_{B_{max}}^{loc}\simeq \mathbf{H}$ is the vacuum, so $B_{max}$ is a lagrangian algebra in $\fZ_1(\mathbf{Is})$.
    The domain wall is $\mathbf{Is}$.
\enu
\end{expl}





\section{A Gappable Gapless Boundary of 2d \texorpdfstring{$\bZ_N$}{Z N} Topological Order}\label{sec:gappable_boundary}
In this section, we construct a gappable gapless boundary of the $\bZ_N$ topological order step by step:
In Section \ref{sec:maths_boundary}, we begin with a review of a mathematical description of gapless boundaries \cite{KZ2018,KZ2020,KZ2021};
in Section \ref{sec:wall}, we construct a gapped domain wall between the double $\bZ_N$ parafermion topological order and the $\bZ_N$ TO;
in Section \ref{sec:boundary}, we obtain a gappable gapless boundary of the $\Zb_N$ topological order by ``dualing'' the domain wall \cite{KZ2021}.
Finally, to verify that this gappable non-chiral gapless boundary actually describes the critical point of the boundary phase transition, we compute the physical observables in Section \ref{sec:partition_function} and list some examples.

\subsection{A mathematical theory of gapless boundaries} \label{sec:maths_boundary}
Imagine a 2d TO living on a disk propagating vertically along the temporal dimension.
This 2d TO has a 1d phase living on the boundary of the disk whose trajectory in spacetime is a 1+1D worldsheet.
This boundary can be either gapped or gapless.
The gapped boundaries of 2d TO has been systematically studied \cite{bravyi1998,beigi2011,KK2012,HWW2017,HLPWW2018}.
More interestingly, a mathematical theory of 1+1D gapless boundaries has been developed recently \cite{KZ2020,KZ2021}.

At the IR fixed point of the renormalization group flow, the observables of a gapless boundary have scale invariance.
It is widely accepted that in 1+1D, scale invariance with unitarity implies conformal symmetry \cite{Polchinski1988,Nakayama2015}.
The gapless boundary modes on the 1+1D worldsheet should be described by a 1+1D conformal field theory (CFT) \cite{BPZ1984,MS1989}.
For a chiral 1+1D CFT, it can be described by a vertex operator algebra (VOA) $V$ \cite{Huang2005,Schweigert2003,FRS2002,FFRS2007}; for a non-chiral 1+1D CFT, it can be described by a so-called full field algebra (FFA) $W$ \cite{HK2007,Kong2007}. 
Since such VOA/FFA encodes the information of local observables, we call it the \textit{local quantum symmetry}.
\begin{rmk}
    When the VOAs or FFAs are rational, it has been proved that the corresponding module categories are modular tensor categories \cite{Huang2008_verlinde,Huang2008_modular}.
\end{rmk}

Besides these local observables on the worldsheet, there are also some long-wave-length-limit observables called the topological defects.
When we move a bulk anyon to the boundary, it becomes a bulk excitation, and its trajectory in spacetime becomes a topological defect line.
These topological defect lines together with higher codimensional topological defects can fuse and form a categorical structure.
As shown in Fig.~\ref{fig:TDL}, 


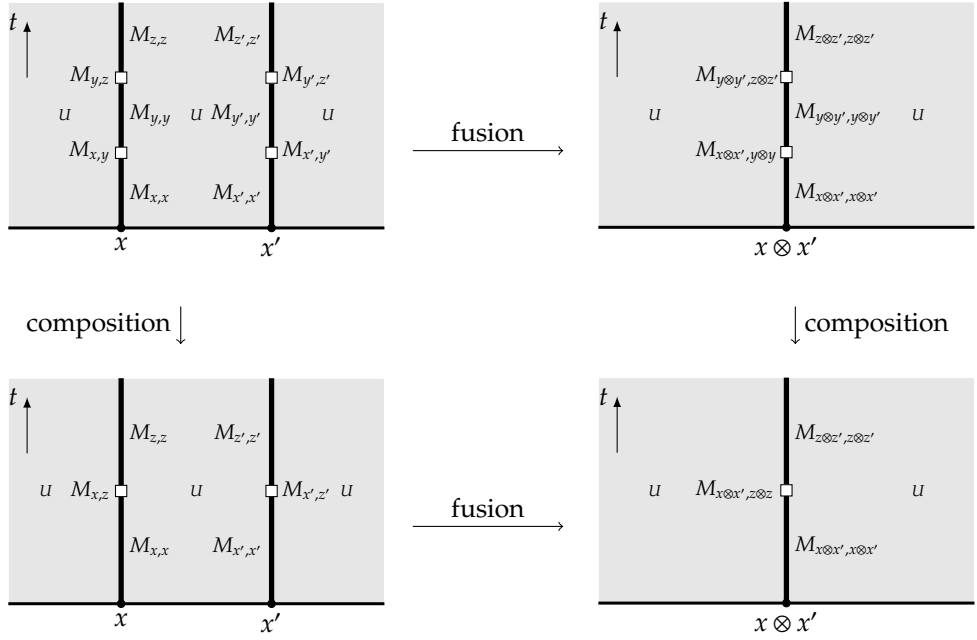
\begin{figure}
    \[
    \begin{array}{c}
        \begin{tikzpicture}[>=latex,scale=0.5]
            \draw [fill=gray!20, draw=white] (-5,0) rectangle (5,6);
            \node [scale=0.6]at (-3.5,3){$U$};
            \node [scale=0.6]at (3.5,3){$U$};
            \node [scale=0.6]at (0,3){$U$};
            \draw [line width=0.4mm] (-5,0)--(5,0);
            \draw [fill] (-2,0) circle [radius=0.10] node[below] {$x$};
            \draw [fill] (2,0) circle [radius=0.10] node[below] {$x'$};
            \draw [line width=0.7mm] (-2,0)--(-2,6) node [scale=0.8,pos=0.15,right] {$M_{x,x}$} node [scale=0.8,pos=0.5,right] {$M_{y,y}$} node [scale=0.8,pos=0.85,right] {$M_{z,z}$};
            \draw [line width=0.7mm] (2,0)--(2,6) node [scale=0.8,pos=0.15,left] {$M_{x',x'}$} node [scale=0.8,pos=0.5,left] {$M_{y',y'}$} node [scale=0.8,pos=0.85,left] {$M_{z',z'}$};
            \draw [fill=white,draw=black] (-2.15,1.85) rectangle (-1.85,2.15);
            \node [left,scale=0.8] at(-2.1,2) {$M_{x,y}$};
            \draw [fill=white,draw=black] (-2.15,3.85) rectangle (-1.85,4.15);
            \node [left,scale=0.8] at(-2.1,4) {$M_{y,z}$};
            \draw [fill=white,draw=black] (1.85,1.85) rectangle (2.15,2.15);
            \node [right,scale=0.8] at(2.1,2) {$M_{x',y'}$};
            \draw [fill=white,draw=black] (1.85,3.85) rectangle (2.15,4.15);
            \node [right,scale=0.8] at(2.1,4) {$M_{y',z'}$};
            \draw [->] (-4.5,4)--(-4.5,5.5) node[left] {$t$};
        \end{tikzpicture}
    \end{array}
    \begin{array}{c}
        \begin{tikzpicture}
            \draw [->] (-1,0)--(1,0) node [above,pos=0.5] {fusion};
        \end{tikzpicture}
    \end{array}
    \begin{array}{c}
        \begin{tikzpicture}[>=latex,scale=0.5]
            \draw [fill=gray!20, draw=white] (-5,0) rectangle (5,6);
            \node [scale=0.6]at (-3.5,3){$U$};
            \node [scale=0.6]at (3.5,3){$U$};
            \draw [line width=0.4mm] (-5,0)--(5,0);
            \draw [fill] (0,0) circle [radius=0.10] node[below] {$x\otimes x'$};
            \draw [line width=0.7mm] (0,0)--(0,6) node [scale=0.8,pos=0.15,right] {$M_{x\otimes x',x\otimes x'}$} node [scale=0.8,pos=0.5,right] {$M_{y\otimes y',y\otimes y'}$} node [scale=0.8,pos=0.85,right] {$M_{z\otimes z',z\otimes z'}$};

            \draw [fill=white,draw=black] (-0.15,1.85) rectangle (0.15,2.15);
            \node [right,scale=0.8] at(-2.6,2) {$M_{x\otimes x',y\otimes y}$};
            \draw [fill=white,draw=black] (-0.15,3.85) rectangle (0.15,4.15);
            \node [right,scale=0.8] at(-2.6,4) {$M_{y\otimes y',z\otimes z'}$};
            \draw [->] (-4.5,4)--(-4.5,5.5) node[left] {$t$};
        \end{tikzpicture}
    \end{array}
    \]
    \[
    \begin{array}{c}
        \begin{tikzpicture}[scale=0.5]
            \draw [->] (0,0.5)--(0,-0.5) node [left,pos=0.5] {composition};
            \node at(5,0) {};
            \node at (-5,0) {};
        \end{tikzpicture}
    \end{array}
    \begin{array}{c}
        \begin{tikzpicture}
            \node at(-1,0) {};
            \node at (1,0) {};
        \end{tikzpicture}
    \end{array}
    \begin{array}{c}
        \begin{tikzpicture}[scale=0.5]
            \draw [->] (0,0.5)--(0,-0.5) node [right,pos=0.5] {composition};
            \node at(5,0) {};
            \node at (-5,0) {};
        \end{tikzpicture}
    \end{array}
    \]
    \[
        \begin{array}{c}
            \begin{tikzpicture}[>=latex,scale=0.5]
                \draw [fill=gray!20, draw=white] (-5,0) rectangle (5,6);
                \node [scale=0.6]at (-4,3){$U$};
                \node [scale=0.6]at (4,3){$U$};
                \node [scale=0.6]at (0,3){$U$};
                \draw [line width=0.4mm] (-5,0)--(5,0);
                \draw [fill] (-2,0) circle [radius=0.10] node[below] {$x$};
                \draw [fill] (2,0) circle [radius=0.10] node[below] {$x'$};
                \draw [line width=0.7mm] (-2,0)--(-2,6) node [scale=0.8,pos=0.25,right] {$M_{x,x}$}  node [scale=0.8,pos=0.75,right] {$M_{z,z}$};
                \draw [line width=0.7mm] (2,0)--(2,6) node [scale=0.8,pos=0.25,left] {$M_{x',x'}$}  node [scale=0.8,pos=0.75,left] {$M_{z',z'}$};
                \draw [fill=white,draw=black] (-2.15,2.85) rectangle (-1.85,3.15);
                \node [left,scale=0.8] at(-2.1,3) {$M_{x,z}$};
                \draw [fill=white,draw=black] (1.85,2.85) rectangle (2.15,3.15);
                \node [right,scale=0.8] at(2.1,3) {$M_{x',z'}$};

                \draw [->] (-4.5,4)--(-4.5,5.5) node[left] {$t$};
            \end{tikzpicture}
        \end{array}
        \begin{array}{c}
            \begin{tikzpicture}
                \draw [->] (-1,0)--(1,0) node [above,pos=0.5] {fusion};
            \end{tikzpicture}
        \end{array}
        \begin{array}{c}
            \begin{tikzpicture}[>=latex,scale=0.5]
                \draw [fill=gray!20, draw=white] (-5,0) rectangle (5,6);
                \node [scale=0.6]at (-3.5,3){$U$};
                \node [scale=0.6]at (3.5,3){$U$};
                \draw [line width=0.4mm] (-5,0)--(5,0);
                \draw [fill] (0,0) circle [radius=0.10] node[below] {$x\otimes x'$};
                \draw [line width=0.7mm] (0,0)--(0,6) node [scale=0.8,pos=0.25,right] {$M_{x\otimes x',x\otimes x'}$}  node [scale=0.8,pos=0.75,right] {$M_{z\otimes z',z\otimes z'}$};
    
                \draw [fill=white,draw=black] (-0.15,2.85) rectangle (0.15,3.15);
                \node [right,scale=0.8] at(-2.6,3) {$M_{x\otimes x',z\otimes z}$};

                \draw [->] (-4.5,4)--(-4.5,5.5) node[left] {$t$};
            \end{tikzpicture}
        \end{array}
    \]
    \caption[]{Topological Defect Lines on 1+1D Worldsheet}
    \label{fig:TDL}
\end{figure}

\begin{itemize}
    \item The boundary topological excitations (or boundary conditions) are labelled by the bullets such as $x$ and $x'$ on the 1+0D boundary;

    \item The 0+1D topological defects such as $M_{x,x}$ are labelled by the segments of these topological defect lines on the worldsheet.
    
    \item The 0D topological defects are labelled by the little white squares on these topological defects lines.
    They are 0D domain walls between two 0+1D defects, so they are labeled by the label of two adjacent segments.
    For instance, the 0D domain wall between $M_{x,x}$ and $M_{y,y}$ is $M_{x,y}$.
    Physically, $M_{x,y}$ consists of boundary condition changing operators.

    \item The local quantum symmetry $U$ living on the worldsheet can transparently move across topological defect lines except for those 0D defects.
    This induces an $U$-action on each 0+1D defect, i.e. $\iota_x: U\hookrightarrow M_{x,x}$.
    Thus each defect $M_{x,y}$ is an $U$-module and hence is an object in $\mathrm{Mod}_U$.

    \item As shown in the vertical process in Fig.~\ref{fig:TDL}, 0D defects can fuse vertically, which corresponds to the composition of instantons.
    For example, consider two 0D defects, $M_{x,y}$ and $M_{y,z}$.
    Their composition is given by the following data:
    \begin{align*}
        \circ:M_{y,z}\otimes_U M_{x,y}\to M_{x,z}, 
    \end{align*}
    where $M_{y,z}\otimes_U M_{x,y}$ is the tensor product (monoidal structure) in the UMTC $\mathrm{Mod}_U$.

    \item Also, as shown in the horizontal process in Fig.~\ref{fig:TDL}, the topological defect lines can fuse horizontally, and the 0D defects on the line also fuse at the same time.
    This corresponds to the monoidal structure of enriched fusion categories.
    For instance, the horizontal fusion of $M_{x,y}$ and $M_{x',y'}$ is given by 
    \begin{align*}
        M_{x,y}\otimes_U M_{x',y'}\to M_{x\otimes x',y\otimes y'}.
    \end{align*}
\end{itemize}

In summary, these data indeed form an \textit{enriched fusion category} $^{\mathrm{Mod}_U}\CM$, which captures the information of topological defects.
Therefore, we call such an enriched fusion category the \textit{topological skeleton}.
We don't give a detailed mathematical definition of an enriched fusion category here.
For more mathematical details of enriched fusion categories, please see \cite{MP2019,KZ2020}.

\begin{rmk}
    Actually, we can regard $M_{x,x}$ as a $0D$ domain wall between $M_{x,x}$ and $M_{x,x}$ itself.
    Hence, the $1D$ defect $M_{x,x}$ can be regarded as a stacking of infinitely many $0D$ defects $M_{x,x}$.
    Based on this view, the composition $\circ:M_{x,x}\otimes_U M_{x,x}\to M_{x,x}$ induces an algebra structure on $M_{x,x}$.
    This algebra structure is highly non-trivial, indeed, it is a so-called open string vertex operator algebra that describes an open (boundary) CFT \cite{HK2004}.
\end{rmk}

In summary, a 1+1D gapless boundary of a 2+1D gapped TO can be described by a pair $(U,{}^{\mathrm{Mod}_U}\CM)$ \cite{KZ2020}, where
\begin{itemize}
    \item the local quantum symmetry $U$ is assumed to be a unitary rational VOA $V$ \cite{DL2014}, which corresponds to a $chiral$ gapless boundary, or a unitary rational FFA $W$, which corresponds to a $non$-$chiral$ gapless boundary \cite{KZ2021}, 
    \item the topological skeleton $^{\mathrm{Mod}_U}\CM$ is a (unitary) $\mathrm{Mod}_U$-enriched fusion category $\CM$ where $\mathrm{Mod}_U$ is the background category that captures the gapless boundaries modes and, $\CM$ is the underlying category formed by the topological excitations on the 1+0D boundary.
\end{itemize}
 For readers who are interested in the bootstrap analysis that an 1d gapless boundary can be mathematically described by an enriched fusion category with a local quantum symmetry, we recommend \cite{KZ2020}.
\begin{rmk}
    In general, for a VOA $V$, it is not clear if $\mathrm{Mod}_V$ is a UMTC when $V$ is unitary.
    See \cite{Gui2019i,Gui2019ii,Gui2019energy} for a discussion of the relation between the unitarity of a VOA $V$ and that of $\mathrm{Mod}_V$.
\end{rmk}

\begin{rmk}
    If we regard $\bC$ as a trivial VOA, then  a VOA $V$ can be regarded as an FFA $W:=V\otimes \bC$, hence each chiral gapless boundary can be regarded as a special case of non-chiral gapless boundaries.
\end{rmk}

\begin{rmk}
    A gapped domain wall $\CM$ can also be described by a pair $(\bC,\,^{\mathbf{H}}\CM)$ where $\bC$ can be regarded as a trivial VOA whose module category is $\mathbf{H}$.
    Therefore, gapped and gapless boundaries are unified in the language of enriched categories.
    This finishes the unified story of enriched category language on 1d quantum phases (without symmetry). 
\end{rmk}

\begin{equation}\label{pic:gappable}
    \begin{array}{c}
        \begin{tikzpicture}
            \draw [fill=gray!20, draw=white] (0,0) circle[radius=2cm];
            \draw [line width=0.5mm,color=cyan](0,2) arc (90:270: 2cm and 2cm);
            \node at (-2,0)[left] {$\CN$};
            \draw [line width=0.5mm,color=blue](0,-2) arc (270:450: 2cm and 2cm);
            \node at (2,0)[right] {$^{\mathrm{Mod}_W}\CM$};
            \node at (0,0){$\fZ_1(\rep(\bZ_N)$};
            \node at (-2.7,1.8){$\rVec$};
        \end{tikzpicture}
    \end{array}
\end{equation}

Recall that our goal is to find a gapless boundary to describe a critical point of a self-dual topological phase transition on the 1d boundary of $\bZ_N$ TO.
Note that the $\bZ_N$ quantum double TO $(\fZ_1(\rep(\bZ_N)),0)$ is non-chiral, which does not admit a chiral gapless boundary.
Hence, we need to consider a non-chiral gapless boundary
\begin{align}
    (W,{}^{\mathrm{Mod}_W}M),
\end{align}
where $W$ is a FFA and $M$ is an $\mathrm{Mod}_W$-enriched fusion category.
Moreover, the non-chiral gapless boundary describing the critical point should be $gappable$ because this gapless boundary will become one of the gapped boundaries under a small perturbation.
During this process, the bulk is unchanged, so this gapless boundary and the two gapped boundaries share the same bulk.
Mathematically, it is equivalent to say that the center of this gapless boundary and the center of the two gapped boundaries are the same.
As shown in Picture~\ref{pic:gappable}, using $\CN$ denote one of the gapped boundaries, we have $\fZ_1(^{\mathrm{Mod}_W}\CM)=\fZ_1(\rep(\bZ_N))=\fZ_1(\CN)$ by the boundary-bulk relations \cite{KWZ2017,KYZZ2021}.


\subsection{A gapped wall between \texorpdfstring{$\bZ_N$}{Z N} quantum double and \texorpdfstring{$\bZ_N$}{Z N} double parafermion} \label{sec:wall}
In Section \ref{sec:any_cond} and Appendix.~\ref{sec:any_cond_maths} we review the anyon condensation theory and its mathematical description. 
To find a gapped domain wall between the double $\bZ_N$ parafermion TO $(\fZ_1(\mathrm{PF}_N),0)$ and the $\bZ_N$ TO $(\fZ_1(\rep(\bZ_N)),0)$, we need to find a condensable algebra $B$ in $\fZ_1(\mathrm{PF}_N)$ such that the category $\fZ_1(\mathrm{PF}_N)_B^{loc}$ of local $B$-modules in $\fZ_1(\mathrm{PF}_N)$ is equivalent to $\fZ_1(\rep(\bZ_N))$ as UMTC. 
Then the category $\fZ_1(\mathrm{PF}_N)_B$ of right $B$-modules in $\fZ_1(\mathrm{PF}_N)$ describes the gapped domain wall we need.

\begin{equation*}
    \begin{array}{c}
        \begin{tikzpicture}
            
            \draw [fill=gray!30, draw=white] (-3,0) rectangle (0,3);
            \draw [fill=gray!15, draw=white] (0,0) rectangle (3,3);
            \draw [line width=0.5mm] (0,0)--(0,3);
            \node at (-1.5,1.5) {$\fZ_1(\mathrm{PF}_N)$};
            \node at (1.5,2) {$\fZ_1(\rep(\bZ_N))$};
            \node [rotate=270] at (1.5,1.5) {$\simeq$};
            \node at (1.5,1) {$\fZ_1(\mathrm{PF}_N)_B^{loc}$};
            \node [below] at(0,0) {$\fZ_1(\mathrm{PF}_N)_B$};
        \end{tikzpicture}
    \end{array}
\end{equation*}

We claim that 
\begin{align}
    B:=\bigoplus_{u=0}^{[\frac{N}{2}]}\fM_{2u,0}\boxtimes \overline{\fM_{2u,0}},
\end{align}
is a condensable algebra in $\fZ_1(\mathrm{PF}_N)$.
To see this, consider the fusion subcategory $\CP\subset \mathrm{PF}_N$ generated by simple objects $\fM_{2u,0}$, $u=0,\dots,[N/2]$. 
Let $\otimes :\CP\boxtimes \overline{\CP}\to \CP$ be the tensor functor, then there is a lagrangian algebra in $\CP\boxtimes \overline{\CP}$ given by $\otimes^R(\fM_{0,0})=\bigoplus_{u=0}^{[N/2]}\fM_{2u,0}\boxtimes \overline{\fM_{2u,0}}$ \cite{Kong2009}, where $\otimes^R$ is the right adjoint of the tensor functor.
Since $\CP\boxtimes \overline{\CP}$ is a fusion subcategory of $\fZ_1(\mathrm{PF}_N)$, thus a condensable algebra of $\CP\boxtimes \overline{\CP}$ also is a condensable algebra of $\fZ_1(\mathrm{PF}_N)$, therefore $B$ is a condensable algebra in $\fZ_1(\mathrm{PF}_N)$.

The quantum dimension of $B$ is $\frac{(N+2)}{4\sin^2(\frac{\pi}{N+2})}$. 
Recall that the quantum dimension of $\fZ_1(\mathrm{PF}_N)$ is $(\frac{N(N+2)}{4\sin^2(\frac{\pi}{N+2})})^2$.
By Theorem \ref{thm:dimension}, we can calculate the quantum dimension of $\fZ_1(\mathrm{PF}_N)_B^{loc}$, which is 
\begin{align*}
    \dim (\fZ_1(\mathrm{PF}_N)_B^{loc})=\dfrac{\dim (\fZ_1(\mathrm{PF}_N)) }{(\dim (B))^2}=N^2.
\end{align*} 
Note that quantum dimension of $\fZ_1(\rep (\bZ_N))$ is also $N^2$, which suggests that these two categories $\fZ_1(\rep (\bZ_N))$ and $\fZ_1(\mathrm{PF}_N)_B^{loc}$ might be equivalent.
What remains is to prove that the equivalence holds.

An equivalence between categories $\fZ_1(\mathrm{PF}_N)_B^{loc}$ and $\fZ_1(\rep(\bZ_N))$ provides a one-to-one correspondence between simple objects of these two categories.
In Section \ref{sec:ZrepZN}, we already know enough about the simple objects in $\fZ_1(\rep(\bZ_N))$, hence it is reasonable to discuss the simple objects of $\fZ_1(\mathrm{PF}_N)_B^{loc}$.
Recall that an object of $\fZ_1(\mathrm{PF}_N)_B^{loc}$ is a local $B$-module in $\fZ_1(\mathrm{PF}_N)$, so we begin with a discussion of $B$-modules.
A right $B$-module is a pair $(x,\rho_x)$ where $x$ is an object in $\fZ_1(\mathrm{PF}_N)$ and $\rho_x:x\otimes B\to x$ is the $B$-action on $x$ satisfying some coherence conditions.
Notice that there might be multiple inequivalent $B$-actions such as $\rho_x^1$ and $\rho_x^2$ on the same object $x$, resulting in $(x,\rho_x^1)$ and $(x,\rho_x^2)$ being two inequivalent $B$-modules.
Therefore, it is better to start considering the free $B$-modules $x\otimes B$ first, whose $B$-action is uniquely induced by the multiplication $m_B:B\otimes B\to B$ of the condensable algebra $B$.
\begin{align*}
    \rho_{x\otimes B}:x\otimes B\otimes B\xrightarrow{\id_x\otimes m_B} x\otimes B.
\end{align*}
In other words, the $B$-action on a free $B$ module $x\otimes B$ has no relation with $x$.
In addition, it has been proved that each simple $B$-module is a direct summand of some free $B$-modules \cite{DMNO2013}.
This is even more motivating for us to consider the free $B$-modules.

For each simple object $x$ in $\fZ_1(\mathrm{PF}_N)$, $x\otimes B$ is a free $B$-module.
If $x\otimes B$ is simple and local, then it is a simple object in $\fZ_1(\mathrm{PF}_N)_B^{loc}$.
To check this, we need the following adjunction between the tensor functor $-\otimes B:\fZ_1(\mathrm{PF}_N)\to \fZ_1(\mathrm{PF}_N)_B$ which sends each object $x$ in $\fZ_1(\mathrm{PF}_N)$ to the free module $x\otimes B$, and the forgetful functor $U: \fZ_1(\mathrm{PF}_N)_B\to \fZ_1(\mathrm{PF}_N)$ which forgets the $B$-action $\rho_y$ for each $B$-module $(y,\rho_y)$.

\begin{prp}
    For any $x,y$ in $\mathrm{PF}_N$, there is a (natural) isomorphism: 
    \begin{align}\label{eq:adjunction}
        \hom_{\fZ_1(\mathrm{PF}_N)_B}(x\otimes B,y\otimes B)\simeq \hom_{\fZ_1(\mathrm{PF}_N)}(x,y\otimes B).
    \end{align}
\end{prp}

We use this adjunction to check if a free $B$-module $x\otimes B$ is simple and local, and if two simple free $B$-modules $x\otimes B$ and $y\otimes B$ are equivalent to each other.

\bnu
    \item By Schur's lemma, $x\otimes B$ is a  simple $B$-module if and only if $\dim \hom_{\fZ_1(\mathrm{PF}_N)_B}(x\otimes B,x\otimes B)=1$ \cite{EGNO2016}.
    By Eq.~\ref{eq:adjunction}, it equals to $\dim \hom_{\fZ_1(\mathrm{PF})}(x,x\otimes B)$ so that we only need to count the number of $x$ in $x\otimes B$.
    
    \item  When we only consider the free $B$-modules, we can just check that if the double braiding $c_{B,x}\circ c_{x,B}$ is trivial.
    If so, then $x\otimes B$ is a local $B$-module.
    
    \item If $\dim \hom_{\fZ_1(\mathrm{PF}_N)_B}(x\otimes B,y\otimes B)=0$, then due to Schur's lemma \cite{EGNO2016}, $x\otimes B$ and $y\otimes B$ must be inequivalent.
    As above, $\dim \hom_{\fZ_1(\mathrm{PF}_N)}(x,y\otimes B)$ is equivalent to the number of $x$ in $y\otimes B$, hence we only need to calculate $y\otimes B$ and count $x$.
\enu

However, we have to discuss the simple objects of $\fZ_1(\mathrm{PF}_N)_B$ separately because the properties of $\fZ_1(\mathrm{PF}_N)_B$ are very different for $N$ as odd and even cases.

\subsubsection{Odd case}\label{sec:odd}
Fortunately, when $N$ is odd, i.e. $N=2k+1$ for some $k\in\bN$, all simple local $B$-modules in $\fZ_1(\mathrm{PF}_{2k+1})$ are free $B$-modules.
By above procedure, it is easy to check that
\begin{align}
    \cX_{a,b}:=(\fM_{0,2a}\boxtimes \overline{\fM_{0,2b}})\otimes B,\quad a,b=0,\dots, 2k,
\end{align}
are all simple and local, and are inequivalent to each other.
Therefore, $\cX_{a,b}$ is a simple object in $\fZ_1(\mathrm{PF}_{2k+1})_B^{loc}$ for each $a,b$.
Notice that we have listed $(2k+1)^2$ inequivalent simple local $B$-modules, and $\fZ_1(\rep(\bZ_{2k+1}))$ has $(2k+1)^2$ inequivalent simple objects.
It suggests us that we have exhausted all the simple local $B$-modules in $\fZ_1(\mathrm{PF}_{2k+1})$.
Indeed, we have the following theorem, whose  rigorous proof is shown in Appendix.~\ref{equivalence}.
\begin{thm}
    The category $\fZ_1(\mathrm{PF}_{2k+1})_B^{loc}$ is equivalent to the category $\fZ_1(\rep(\bZ_{2k+1}))$ as modular tensor categories.
    The correspondence of simple objects of two categories is given by
    \begin{align}
        \cO_{a+b,a-b}\sim \cX_{a,b}=\bigoplus_{u=0}^{k} \fM_{2u,2a}\boxtimes\overline{\fM_{2u,2b}}.
    \end{align}
\end{thm}

Recall the physical notation $\mathbf{e}^{\alpha}\mathbf{m}^{\beta}$ of simple objects $\mathcal{O}_{\alpha,\beta}$.
Let $\alpha=a+b \mod (2k+1)$ and $\beta=a-b \mod (2k+1)$, we can write the correspondence of simple objects between $\fZ_1(\rep(\bZ_{2k+1}))$ and $\fZ_1(\mathrm{PF}_{2k+1})_B^{loc}$ as follows,
\begin{align}
    \mathbf{e}^{\alpha}\mathbf{m}^{\beta}=\begin{cases}
    \cX_{\frac{\alpha+\beta}{2},\frac{\alpha-\beta}{2}}, & \mathrm{if}\quad \alpha+\beta\equiv 0\mod 2,\\
    \cX_{2k+1-(\alpha+\beta),2k+1-(\alpha-\beta)}, & \mathrm{if}\quad \alpha+\beta\equiv 1\mod 2.
    \end{cases}
\end{align}

The above formula is a little complicated, but we can rewrite it in a compact form.
First, we consider the case with only charges $\mathbf{e}^{\alpha}$. 
We can rewrite the expansion of $\cX_{a,b}$ such that the second index of each term is $\alpha$:
\begin{align}\label{eq:charge}
    \mathbf{e}^{\alpha}:=\bigoplus_{\ell+\alpha \equiv 0\mod 2}\fM_{\ell,\alpha} \boxtimes \overline{\fM_{\ell,\alpha}}.
\end{align}

Similarly, for the case with only flux $\mathbf{m}^{\beta}$, we can rewrite the expansion of $\cX_{a,b}$ as follows:
\begin{align}\label{eq:flux}
    \mathbf{m}^{\beta}=\bigoplus_{\ell+\beta\equiv 0\mod 2}\fM_{\ell,\beta}\boxtimes \overline{\fM_{\ell,(\beta\mod (2k+1))-2\beta}}.
\end{align}

Recall that $\mathbf{e}^{\alpha}\mathbf{m}^{\beta}=\mathbf{e}^{\alpha}\otimes \mathbf{m}^{\beta}$, therefore by Eq.~\eqref{eq:charge} and Eq.~\eqref{eq:flux} we have a general formula:
\begin{align}\label{eq:expansion}
    \mathbf{e}^{\alpha}\mathbf{m}^{\beta}=\bigoplus_{m:=\alpha+\beta\mod (2k+1)\atop \ell+m\equiv 0\mod 2} \fM_{\ell,m}\boxtimes \overline{\fM_{\ell,m-2\beta}}.
\end{align}
Here we choose $m$ to be the remainder of $\alpha+\beta$ modulo $(2k+1)$ to avoid double counting.

In summary, we prove that the condensed phase $(\fZ_1(\mathrm{PF}_{2k+1})_B^{loc},0)$ of double $\bZ_{2k+1}$ parafermion TO $(\fZ_1(\mathrm{PF}_{2k+1}),0)$ via condensing $B$ is the $\bZ_{2k+1}$ quantum double TO $(\fZ_1(\rep(\bZ_{2k+1})),0)$.
According to the anyon condensation theory, the gapped domain wall between double $\bZ_{2k+1}$ parafermion TO and $\bZ_{2k+1}$ quantum double TO is described by $\fZ_1(\mathrm{PF}_{2k+1})_B$.

\subsubsection{Even case}
When $N=2k$ for some $k\in\bN$, the task of determining all simple local $B$-modules is quite challenging.
The argument based on Eq.~\eqref{eq:adjunction} only works for free $B$-modules.
With the notation in the above subsection, $\{\cX_{a,b}\mid a,b=0,\dots,2k-1 \}$ is still a set of simple local $B$-modules.
However, the elements in this set satisfy an equivalence relation $\cX_{a,b}\simeq\cX_{a+k,b+k}$ by procedure 2, so half of the elements in this set are inequivalent as listed in $\{\cX_{a,b}\mid a=0,\dots,2k-1,\,b=0,\dots,k-1\}$.
The rest of the simple local $B$-modules are not free $B$-modules. 

Indeed, when $N=2$ \cite{CJKYZ2020}, there are 4 inequivalent simple local $B$-modules:
\begin{align*}
    \mathbb{1}:=(\fM_{0,0}\boxtimes \fM_{0,0})\otimes B,\,\, \mathbf{f}:=(\fM_{2,0}\boxtimes \fM_{0,0})\otimes B,\,\, \mathbf{e}:=\fM_{1,1}\boxtimes \fM_{1,1},\,\, \mathbf{m}:=(\fM_{1,1}\boxtimes \fM_{1,1})^{tw}.
\end{align*}
which correspond to 4 simple objects in $\fZ_1(\rep(\bZ_2))$ respectively.
Notice that $\fM_{1,1}\boxtimes \fM_{1,1}$ and $(\fM_{1,1}\boxtimes \fM_{1,1})^{tw}$ are not free $B$-module.
Here $(\fM_{1,1}\boxtimes \fM_{1,1})^{tw}$ is a $B$-module whose object is $\fM_{1,1}\boxtimes \fM_{1,1}$ with a twisted $B$-action. For a general $N=2k>2$, it is difficult to determine the simple local $B$-modules that are not free because we need to consider the $B$-actions on modules now. 
One possible way is to guess a $B$-module and its $B$-action, check if it is a local $B$-module, and then check if it is simple. 
However, for two $B$-modules, they may be isomorphic to each other, which raises the difficulty of finding simple local B-modules. 

We do not give all the simple local $B$-modules in the even cases, and also no proof for the equivalence between UMTCs $\fZ_1(\mathrm{PF}_{2k})_B^{loc}$ and $\fZ_1(\rep(\bZ_{2k}))$.
But we believe that the equivalence: $\fZ_1(\mathrm{PF}_{2k})_B^{loc}\simeq\fZ_1(\rep(\bZ_{2k}) $ still holds, since we confirm that half of the simple objects ($\{\mathbf{e}^{\alpha}\mathbf{m}^{\beta}\mid \alpha+\beta\equiv 0\mod 2\}$) in $\fZ_1(\rep(\bZ_{2k}))$ indeed correspond to the simple objects $\{\cX_{a,b}\}$ in $\fZ_1(\mathrm{PF}_{2k})_B^{loc}$. 



\subsection{A gappable non-chiral gapless boundary of \texorpdfstring{$\bZ_N$}{Z N} topological order} \label{sec:boundary}
In the previous subsection, we have found a gapped domain wall between $\fZ_1(\mathrm{PF}_{N})_B^{loc}$ and $\fZ_1(\rep(\bZ_{N}))$.
There is a holographic duality between gapped domain walls of two 2+1D TOs and the anomalous gapless boundaries of one of the TOs that is set up \cite{KZ2021}. By a magical trick called the $topological$ $Wick$ $rotation$ \cite{KZ2020}, we can ``dual" this gapped domain wall to a gappable non-chiral gapless boundary of $\bZ_N$ TO. In the following, let's discuss the details of this duality.

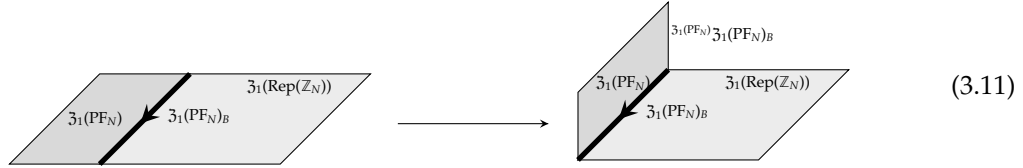
\begin{equation}\label{fig:TWR_1}
    \begin{array}{c}
        \begin{tikzpicture}[scale=0.6]
            \tikzset{->-/.style={decoration={markings,mark=at position #1 with {\arrow{stealth}}},postaction={decorate}}}
            \draw [fill=gray!30] (-3,0)--(-1,2)--(1,2)--(-1,0)--cycle;
            \draw [fill=gray!15] (-1,0)--(1,2)--(5,2)--(3,0)--cycle;
            \draw [->-=0.55, line width=0.7mm] (1,2)--(-1,0);
            \node [scale=0.6] at (1.2,1.06) {$\fZ_1(\mathrm{PF}_N)_B$};
            \node [scale=0.6]at (3.2,1.7) {$\fZ_1(\rep(\bZ_N))$};
            \node [scale=0.6]at (-1.1,1) {$\fZ_1(\mathrm{PF}_N)$};
            \node at (0,3.5) {};
        \end{tikzpicture}
    \end{array}
    \begin{array}{c}
        \begin{tikzpicture}[scale=0.8]
            \tikzset{-->/.style={decoration={markings,mark=at position #1 with {\arrow{stealth}}},postaction={decorate}}}
            \draw[-->=1] (-0.5,0)--(2,0);
            \node at (1,1.25) {};
        \end{tikzpicture}
    \end{array}
    \begin{array}{c}
        \begin{tikzpicture}[scale=0.6]   
            \tikzset{->-/.style={decoration={markings,mark=at position #1 with {\arrow{stealth}}},postaction={decorate}}}
            \draw [fill=gray!15] (-1,0)--(3,0)--(5,2)--(1,2)--cycle;
            \draw [fill=gray!30] (1,2)--(1,3.5)--(-1,1.5)--(-1,0)--cycle;
            \draw [->-=0.55, line width=0.7mm] (1,2)--(-1,0);
            \node [scale=0.6] at (1.2,1.06) {$\fZ_1(\mathrm{PF}_N)_B$};
            \node [scale=0.6] at (3.2,1.7) {$\fZ_1(\rep(\bZ_N))$};
            \node [scale=0.6] at (2.2,2.8) {${}^{\fZ_1(\mathrm{PF}_N)}\fZ_1(\mathrm{PF}_N)_B$};
            \node [scale=0.6] at (0,1.7) {$\fZ_1(\mathrm{PF}_N)$};
        \end{tikzpicture}
    \end{array}   
\end{equation}

One can imagine that flipping $\fZ_1(\mathrm{PF}_N)$ from the spatial dimension to the temporal dimension.
The vertical part becomes a background of the $\bZ_N$ quantum double TO.
By the classification theorem of non-chiral gapless edge \cite{KZ2021}, we obtain an enriched fusion category $^{\fZ_1(\mathrm{PF}_N)} \fZ_1(\mathrm{PF}_N)_B$.
Recall that $\mathrm{PF}_N=\mathrm{Mod}_{V_{\mathrm{PF}_N}}$, thus we have $\fZ_1(\mathrm{PF}_N)\simeq \mathrm{PF}_N\boxtimes\overline{\mathrm{PF}_N}\simeq \mathrm{Mod}_{V_{\mathrm{PF}_N}\otimes_{\bC}\overline{V_{\mathrm{PF}_N}}}$ which is the module category over the FFA $V_{\mathrm{PF}_N}\otimes_{\bC}\overline{V_{\mathrm{PF}_N}}$.
The enriched fusion category $^{\fZ_1(\mathrm{PF}_N)} \fZ_1(\mathrm{PF}_N)_B$ together with the FFA $W:=V_{\mathrm{PF}_N}\otimes_{\bC}\overline{V_{\mathrm{PF}_N}}$ form a non-chiral gapless boundary of 2+1D $\bZ_N$ quantum double TO $(\fZ_1(\rep(\bZ_N)),0)$,
\begin{align}
    (V_{\mathrm{PF}_N}\otimes_{\bC}\overline{V_{\mathrm{PF}_N}},\,^{\fZ_1(\mathrm{PF}_N)}\fZ_1(\mathrm{PF}_N)_B).
\end{align}

\begin{rmk}
    Physically, such folding process is not possible.
    Consider that we violently chop away the left half of a 2d topological order.
    The system near the boundary is then away from the current RG fixed point. 
    The boundary will eventually experience self-healing by flowing to a new RG fixed point. 
    One of the possible RG fixed points is the non-chiral gapless
    boundary $(W,{}^{\fZ_1(\mathrm{PF}_N)}\fZ_1(\mathrm{PF}_N)_B)$. This fictional folding trick is named as ``$topological$ $Wick$ $rotation$" in \cite{KZ2018,KZ2020}
    The $topological$ $Wick$ $rotation$ tells us that the information lost through the brutal force cutting in the spatial dimensions can be completely restored in the temporal dimension. That essentially establishes a holographic duality.
\end{rmk}


\begin{rmk}
    When we regard the gapped domain wall $\fZ_1(\mathrm{PF}_N)_B$ as an enriched fusion category ${}^\mathbf{H}\fZ_1(\mathrm{PF}_N)_B$, there is another interpretation of the gapless boundary $(W,{}^{\fZ_1(\mathrm{PF}_N)}\fZ_1(\mathrm{PF}_N)_B)$.
    It can be regarded as a fusion of the canonical gapless boundary $(W,{}^{\fZ_1(\mathrm{PF}_N)}\fZ_1(\mathrm{PF}_N))$ with the gapped domain wall $(\bC,{}^\mathbf{H}\fZ_1(\mathrm{PF}_N)_B)$, as shown in the following picture.
\end{rmk}

\begin{equation}
    \begin{array}{c}
        \begin{tikzpicture}[scale=0.6]   
            \tikzset{->-/.style={decoration={markings,mark=at position #1 with {\arrow{stealth}}},postaction={decorate}}}
            \draw [->-=0.55, line width=0.7mm] (1,2)--(-1,0) ;
            \draw (-1,0)--(5,0)--(7,2)--(1,2)--(1,3.5)--(-1,1.5)--cycle;
            \draw [->-=0.55, line width=0.7mm] (3,2)--(1,0) node[below right,pos=0.5,scale=0.6] {$\fZ_1(\mathrm{PF}_N)_B$};
            \draw [dashed] (1,0)--(1,1.5)--(3,3.5)--(3,2);
            \node [scale=0.6]at (0.25,0.3) {$\fZ_1(\mathrm{PF}_N)$};
            \node [scale=0.6]at (5.5,1.6) {$\fZ_1(\rep (\bZ_N))$};
            \node [scale=0.6]at (-2.45,1) {$(W,{}^{\fZ_1(\mathrm{PF}_N)}\fZ_1(\mathrm{PF}_N))$};
            \node [scale=0.6]at (0,1.7) {$\fZ_1(\mathrm{PF}_N)$};
            \node [scale=0.6]at (4.15,2.8) {$(\bC,{}^{\mathbf{H}}\fZ_1(\mathrm{PF}_N)_B)$};

        \end{tikzpicture}
    \end{array}
    \begin{array}{c}
        \begin{tikzpicture}[scale=0.8]
            \tikzset{-->/.style={decoration={markings,mark=at position #1 with {\arrow{stealth}}},postaction={decorate}}}
            \draw[-->=1] (-0.5,0)--(2,0);
            \node at (1,1.25) {};
        \end{tikzpicture}
    \end{array}
    \begin{array}{c}
        \begin{tikzpicture}[scale=0.6]   
            \tikzset{->-/.style={decoration={markings,mark=at position #1 with {\arrow{stealth}}},postaction={decorate}}}
            \draw [->-=0.55, line width=0.7mm] (1,2)--(-1,0) node[scale=0.6,pos=0.5,below right]{$\fZ_1(\mathrm{PF}_N)_B$};
            \draw (-1,0)--(3,0)--(5,2)--(1,2)--(1,3.5)--(-1,1.5)--cycle;
            \node [scale=0.6]at (3.5,1.6) {$\fZ_1(\rep (\bZ_N))$};
            \node [scale=0.6]at (2.55,2.8) {$(W,{}^{\fZ_1(\mathrm{PF}_N)}\fZ_1(\mathrm{PF}_N)_B)$};
            \node [scale=0.6]at (0,1.7) {$\fZ_1(\mathrm{PF}_N)$};
        \end{tikzpicture}
    \end{array}
\end{equation}

The enriched fusion category ${}^{\fZ_1(\mathrm{PF}_N)}\fZ_1(\mathrm{PF}_N)_B$ can be described in more detail when $N=2k+1$.
\begin{itemize}
    \item Objects (topological boundary excitations) are right $B$-modules in $\fZ_1(\mathrm{PF}_N)$, including the simple local $B$-modules $\cX_{a,b}$.
    They can fuse horizontally by the fusion rule of $\fZ_1(\mathrm{PF}_N)_B$.
    \item For two objects $x$ and $y$ in $\fZ_1(\mathrm{PF}_N)_B$, the 0D defects $M_{x,y}$ between them  are given by $(x\otimes_B y^*)^*$ \cite{Ostrik2003}, which should be viewed as an object in $\fZ_1(\mathrm{PF}_N)$.
\end{itemize}
Here $y^*$ is the dual of $y$, since $\fZ_1(\mathrm{PF}_{2k+1})_B$ is a pointed fusion category, the dual of $\cX_{a,b}$ is its inverse $\cX_{2k+1-a,2k+1-b}$.
In more detail, we can calculate the 0D defects for any simple local $B$-modules, and therefore any local $B$-modules.
For two local $B$-modules $\cX_{a,b}$ and $\cX_{c,d}$, we have 
\begin{align*}
    (\cX_{a,b}\otimes_B \cX_{c,d}^*)^*=(\cX_{a,b}\otimes_B\cX_{2k+1-c,2k+1-d})^*=(\cX_{2k+1-(c-a),2k+1-(d-b)})^*=\cX_{c-a,d-b}.
\end{align*}

This non-chiral gapless boundary is gappable.
To see this, we need to prove that the center of ${}^{\fZ_1(\mathrm{PF}_N)}\fZ_1(\mathrm{PF}_N)_B$ is $\fZ_1(\rep(\bZ_N))$.
By the definition in \cite{KYZZ2021}, it is given by  the centralizer of $\fZ_1(\mathrm{PF}_N)$ in $\fZ_1(\fZ_1(\mathrm{PF}_N)_B)$, i.e.
\begin{align}
    \fZ_2(\fZ_1(\mathrm{PF}_N),\fZ_1(\fZ_1(\mathrm{PF}_N)_B)).
\end{align}
Recall by anyon condensation theory \cite{KO2002,Kong2014}, there are equivalences between UMTCs,
\begin{align}
    \fZ_1(\fZ_1(\mathrm{PF}_N)_B)\simeq \fZ_1(\mathrm{PF}_N)\boxtimes \fZ_1(\mathrm{PF}_N)_B^{loc}\simeq \fZ_1(\mathrm{PF}_N)\boxtimes \fZ_1(\rep(\bZ_{N})).
\end{align}
Finally, by the centralizer Theorem \ref{thm:centralizer} proved in \cite{Muger2003}, we can prove that 
\begin{align}
    \fZ_1({}^{\fZ_1(\mathrm{PF}_N)}\fZ_1(\mathrm{PF}_N)_B)\simeq\fZ_2(\fZ_1(\mathrm{PF}_N),\fZ_1(\fZ_1(\mathrm{PF}_N)_B))\simeq\fZ_1(\rep(\bZ_{N})),
\end{align}
which ensures that the non-chiral gapless boundary $(W,{}^{\fZ_1(\mathrm{PF}_N)}\fZ_1(\mathrm{PF}_N)_B)$ has the same bulk as the gapped boundary $(\bC,{}^{\mathbf{H}}\rep(\bZ_N))$.

\subsection{The partition functions of \texorpdfstring{$M_{x,y}$}{Mxy}}\label{sec:partition_function}
Now we have constructed a gappable non-chiral gapless boundary $(V_{\mathrm{PF}_N}\otimes_{\bC}\overline{V_{\mathrm{PF}_N}},\,^{\fZ_1(\mathrm{PF}_N)}\fZ_1(\mathrm{PF}_N)_B)$ of $\bZ_N$ quantum double TO.
It is natural to ask if it describes the critical point of the topological phase transition on the 1d gapped boundary of $\bZ_N$ TO.
To verify this, it is enough to check the partition functions of $M_{x,y}$.
In physics, the partition function of $M_{x,y}$ can be computed by the path integral over an annulus obtained by compactifying the time axis, and fixing the two boundary conditions to be $x$ and $y$.
By the fact that $M_{x,y} = M_{1,(x\otimes_B y^*)^*}$ as objects in $\fZ_1(\mathrm{PF}_N)$, it is easy to see that the partition function of $M_{x,y}$ is the same as that of $M_{1,(x\otimes_B y^*)^*}$.
Thus, calculating partition functions of $M_{1,x}$ is enough for us to get all partition functions of $M_{x,y}$. We denote the partition function of $M_{1,x} = x$ by $\mathsf{Z}(x)$ in the following.

Note that $x=M_{1,x}$ can be viewed as the simple objects in $\fZ_1(\mathrm{PF}_N)$.
Therefore, the partition function $\mathsf{Z}(x)$ can be expressed in terms of the characters $\chi_{\{\ell;m\}}$ of simple objects $\fM_{\ell,m}$ in the parafermion CFT. Recall the expansion of simple local $B$-modules Eq.~\eqref{eq:expansion}, we have 
\begin{align}\label{eq:partition_function_m}
    \mathsf{Z}(\mathbf{e}^{\alpha}\mathbf{m}^{\beta})=\sum_{m:=\alpha+\beta\mod (2k+1)\atop \ell+m\equiv 0\mod 2} \chi_{\{\ell;m\}} \overline{\chi}_{\{\ell;m-2\beta\}}.
\end{align}



However, when $N=2k$, since we only figure out half of the simple objects in $\fZ_1(\mathrm{PF}_{2k})_B^{loc}$, the other half of the partition functions given by Eq.~\eqref{eq:partition_function_m} cannot be verified as correct. Fortunately, all partition functions can be recovered from the lattice model construction in Section \ref{sec:lattice} . 

\subsection{Some examples}
In this subsection, we list some examples of $\bZ_N$.
Recall that the condensable algebra in $\fZ_1(\mathrm{PF}_N)$ is $B:=\bigoplus_{u=0}^{[N/2]}\fM_{2u,0}\boxtimes \overline{\fM_{2u,0}}$.
\begin{notation}
    For the convention of the following discussion, we use $\{\ell;m\}$ denote the simple objects $\fM_{\ell,m}$ in $\mathrm{PF}_N$.
    And use $\{\ell;m\mid \ell';m'\}$ denote the simple objects $\fM_{\ell,m}\boxtimes \overline{\fM_{\ell',m'}}$ in $\fZ_1(\mathrm{PF}_N)$.
\end{notation}

\subsubsection{Example \texorpdfstring{$\mathbb{Z}_2$}{Z 2}}
Let $N=2$, as Example \ref{eg:double_ising} shows, there are $9$ simple objects in double Ising TO.
The condensable algebra is $B=\{0;0\mid 0;0\}\oplus \{2;0\mid 2;0\}$.
There are four local $B$-modules:
\begin{align*}
    \mathbb{1}:=(\fM_{0,0}\boxtimes \overline{\fM_{0,0}})\otimes B,\,\,  \mathbf{e}:=\fM_{1,1}\boxtimes \overline{\fM_{1,1}},\,\, \mathbf{m}:=(\fM_{1,1}\boxtimes \overline{\fM_{1,1}})^{tw},\,\, \mathbf{f}:=(\fM_{2,0}\boxtimes \overline{\fM_{0,0}})\otimes B.
\end{align*}
The corresponding partition function is given by 
\begin{align*}
    &\mathsf{Z}(\one)=\chi_{\{0;0\}}\overline{\chi}_{\{0;0\}}+\chi_{\{2;0\}}\overline{\chi}_{\{2;0\}},\\
    &\mathsf{Z}(\mathbf{e})=\chi_{\{1;1\}}\overline{\chi}_{\{1;1\}},\\
    &\mathsf{Z}(\mathbf{m})=\chi_{\{1;1\}}\overline{\chi}_{\{1;1\}},\\
    &\mathsf{Z}(\mathbf{f})=\chi_{\{2;0\}}\overline{\chi}_{\{0;0\}}+\chi_{\{0;0\}}\overline{\chi}_{\{2;0\}}.
\end{align*}
This recovers the result in \cite{CJKYZ2020}.

\subsubsection{Example \texorpdfstring{$\mathbb{Z}_3$}{Z 3}}
Let $N=3$, for simple objects $\{\ell;m\mid \ell';m'\}$, the index $\ell,\ell'$ has the range $\{0,1,2,3\}$, $m,m'$ has the range $\{0,1,2,3,4,5\}$.
There are $6$ inequivalent simple objects $\mathrm{Irr}(\mathrm{PF}_3):=\{\{0;0\},\{0;2\},\{0;4\},\{2;0\},\{2;2\},\{2;4\}\}$ in $\mathrm{PF}_3$.
Thus there are $36$ inequivalent simple objects in double parafermion $\fZ_1(\mathrm{PF}_3)$ given by $\mathrm{Irr}(\mathrm{PF}_3) \times \mathrm{Irr}(\mathrm{PF}_3)$.

The index $t$ of direct summand of the condensable algebra $B$ has range $\{0,1\}$, then the condensable algebra in $\fZ_1(\mathrm{PF}_3)$ is $B=\{0;0\mid 0;0\}\oplus \{2;0\mid 2;0\}$.
Recall that in odd case, the set of free $B$-modules $\{\cX_{a,b}\mid 0\leq a,b\leq 2\}$ gives all simple local $B$-modules.
Under the equivalence (see Appendix.~\ref{equivalence}) between two modular tensor categories $\fZ_1(\mathrm{PF}_{2k+1})_B^{loc}$ and $\fZ_1(\rep(\bZ_{2k+1}))$, we have $\alpha=a+b \mod 3$ and $\beta=a-b \mod 3$.
Then we can write all the simple local $B$-modules as follows, 
\begin{align*}
    \mathbb{1}:=\{0;0\mid 0;0\}\otimes B&=\{0;0\mid 0;0\}+\{2;0\mid 2;0\},\\
    \mathbf{em}^2:=\{0;0\mid 0;2\}\otimes B&=\{0;0\mid 0;2\}+\{2;0\mid 2;2\},\\
    \mathbf{e}^2\mathbf{m}:=\{0;0\mid 0;4\}\otimes B&=\{0;0\mid 0;4\}+\{2;0\mid 2;4\},\\
    \mathbf{em}:=\{0;2\mid 0;0\}\otimes B&=\{0;2\mid 0;0\}+\{2;2\mid 2;0\},\\
    \mathbf{e}^2:=\{0;2\mid 0;2\}\otimes B&=\{0;2\mid 0;2\}+\{2;2\mid 2;2\},\\
    \mathbf{m}^2:=\{0;2\mid 0;4\}\otimes B&=\{0;2\mid 0;4\}+\{2;2\mid 2;4\},\\
    \mathbf{e}^2 \mathbf{m}^2:=\{0;4\mid 0;0\}\otimes B&=\{0;4\mid 0;0\}+\{2;4\mid 2;0\},\\
    \mathbf{m}:=\{0;4\mid 0;2\}\otimes B&=\{0;4\mid 0;2\}+\{2;4\mid 2;2\},\\
    \mathbf{e}:=\{0;4\mid 0;4\}\otimes B&=\{0;4\mid 0;4\}+\{2;4\mid 2;4\}.\\
\end{align*}
This gives the partition function of simple objects in $\fZ_1(\rep(\mathbb{Z}_3))$.
\begin{align*}
    &\mathsf{Z}(\one)=\chi_{\{0;0\}}\overline{\chi}_{\{0;0\}}+\chi_{\{2;0\}}\overline{\chi}_{\{2;0\}},\\
    &\mathsf{Z}(\mathbf{e})=\chi_{\{0;4\}}\overline{\chi}_{\{0;4\}}+\chi_{\{2;4\}}\overline{\chi}_{\{2;4\}},\\
    &\mathsf{Z}(\mathbf{e}^2)=\chi_{\{0;2\}}\overline{\chi}_{\{0;2\}}+\chi_{\{2;2\}}\overline{\chi}_{\{2;2\}},\\
    &\mathsf{Z}(\mathbf{m})=\chi_{\{0;4\}}\overline{\chi}_{\{0;2\}}+\chi_{\{2;4\}}\overline{\chi}_{\{2;2\}},\\
    &\mathsf{Z}(\mathbf{em})=\chi_{\{0;2\}}\overline{\chi}_{\{0;0\}}+\chi_{\{2;2\}}\overline{\chi}_{\{2;0\}},\\
    &\mathsf{Z}(\mathbf{e}^2\mathbf{m})=\chi_{\{0;0\}}\overline{\chi}_{\{0;4\}}+\chi_{\{2;0\}}\overline{\chi}_{\{2;4\}},\\
    &\mathsf{Z}(\mathbf{m}^2)=\chi_{\{0;2\}}\overline{\chi}_{\{0;4\}}+\chi_{\{2;2\}}\overline{\chi}_{\{2;4\}},\\
    &\mathsf{Z}(\mathbf{em}^2)=\chi_{\{0;0\}}\overline{\chi}_{\{0;2\}}+\chi_{\{2;0\}}\overline{\chi}_{\{2;2\}},\\
    &\mathsf{Z}(\mathbf{e}^2\mathbf{m}^2)=\chi_{\{0;4\}}\overline{\chi}_{\{0;0\}}+\chi_{\{2;4\}}\overline{\chi}_{\{2;0\}}.
\end{align*}

\subsubsection{Example \texorpdfstring{$\mathbb{Z}_4$}{Z 4}}
When $N=4$, $\bZ_4$ is the second primary even case.
In this case the index $\ell,\ell'$ has the range $\{0,1,2,3,4\}$, $m,m'$ has the range $\{0,1,2,3,4,5,6,7\}$.
There are $10$ inequivalent simple objects $\mathrm{Irr}(\mathrm{PF}_4):=\{\{0;0\},\{2;0\},\{4;0\},\{0;2\},\{2;2\},\{4;2\},\{1;1\},\{3;1\},\{1;3\},\{3;3\}\}$ in $\mathrm{PF}_4$.
Thus there are $100$ inequivalent simple objects in double parafermion $\fZ_1(\mathrm{PF}_4)$ given by $\mathrm{Irr}(\mathrm{PF}_4) \times \mathrm{Irr}(\mathrm{PF}_4)$.

The index $t$ of direct summand of the condensable algebra $B$ has range $\{0,1,2\}$, then the condensable algebra in $\fZ_1(\mathrm{PF}_4)$ is $B=\{0;0\mid 0;0\}\oplus \{2;0\mid 2;0\}\oplus \{4;0\mid 4;0\}$.
In principle, there should be $16$ simple objects, or simple local $B$-modules.
The set of free $B$-modules $\{\cX_{a,b}\}$ gives 8 inequivalent simple local $B$-modules.
We also have the relation $\alpha=a+b \mod 4$ and $\beta=a-b \mod 4$.
\begin{align*}
    \mathbb{1}:=\{0;0\mid 0;0\}\otimes B&=\{2;0\mid 0;0\}\oplus \{2;0\mid 0;0\}\oplus \{4;0\mid 4;0\},\\
    \mathbf{em}:=\{0;2\mid 0;0\}\otimes B&=\{0;2\mid 0;0\}\oplus \{2;2\mid 2;0\}\oplus \{4;2\mid 4;0\},\\
    \mathbf{e}\mathbf{m}^3:=\{0;0\mid 0;2\}\otimes B&=\{0;0\mid 0;2\}\oplus \{2;0\mid 2;2\}\oplus \{4;0\mid 4;2\},\\
    \mathbf{e}^2:=\{0;2\mid 0;2\}\otimes B&=\{0;2\mid 0;2\}\oplus \{2;2\mid 2;2\}\oplus \{4;2\mid 4;2\},\\
    \mathbf{e}^3\mathbf{m}:=\{4;0\mid 0;2\}\otimes B&=\{0;0\mid 4;2\}\oplus \{2;0\mid 2;2\}\oplus \{4;0\mid 0;2\},\\
    \mathbf{e}^3\mathbf{m}^3:=\{0;2\mid 4;0\}\otimes B&=\{0;2\mid 4;0\}\oplus \{2;2\mid 2;0\}\oplus \{4;2\mid 0;0\},\\
    \mathbf{m}^2:=\{0;2\mid 4;2\}\otimes B&=\{0;2\mid 4;2\}\oplus \{2;2\mid 2;2\}\oplus \{4;2\mid 0;2\},\\
    \mathbf{e}^2\mathbf{m}^2:=\{0;0\mid 4;0\}\otimes B&=\{0;0\mid 4;0\}\oplus \{2;0\mid 2;0\}\oplus \{4;0\mid 0;0\}.
\end{align*}

The remaining $8$ simple local modules cannot be determined by the above algorithm. In fact, they are not free $B$-modules but should be the direct summand of some free $B$-modules. We can now write the partition functions as we did in the $\bZ_3$ case, which we omit here.

\section{A Lattice Model Realization} \label{sec:lattice}
In this section, we give a lattice model realization of the self-dual critical point of the pure boundary phase transition between $\bfe$-condensed and $\bfm$-condensed gapped boundaries of the 2d $\bZ_N$ topological order, and recover all the ingredients of the gappable non-chiral gapless boundary constructed in Section\,\ref{sec:boundary}. 
We choose the string-net model to be the lattice realization of the 2d $\bZ_N$ topological order, because we hope that the construction in this paper can be further extended to a more general case. 
For the convenience of readers in the mathematical background, we review the string-net model in detail. 

\subsection{\texorpdfstring{$\bZ_N$}{Z N} Toric Code Model } \label{sec:wen-model}
It is well known that an arbitrary 2d non-chiral topological order can be realized by Levin-Wen model \cite{LW2005}, which is a machine that imports a fusion category $\CC$ and outputs the corresponding Drinfeld center $\fZ_1(\CC)$. 
When the input category $\CC=\rep(\bZ_N)$, we can obtain the $\bZ_N$ TO as the $\bZ_N$ generalized version of the toric code \cite{hung2012,schulz2012}.

The Levin-Wen model is usually defined on a 2d honeycomb lattice. 
The edges of the lattice are labeled by the oriented simple-type strings, which are the simple objects of the input fusion category $\CC$. 
And a vertex $v$ connecting three edges labelled by simple objects $i,j,k$ is associated to the hom space $V_{i j}^{k}=\operatorname{Hom}_{\CC}(i \otimes j, k)$ of $\CC$. Therefore, the Hilbert space of Levin-Wen model can be organized as $\CH=\bigotimes_{v} \CH_{v}$, where $\CH_{v}=\bigoplus_{i, j, k} V_{i j}^{k}$.
The Hamiltonian is supposed to be exactly solvable and is written as,

\begin{equation}
\begin{aligned}
H_{\text{Levin-Wen}}=-\sum_{v}A_v-\sum_{p}B_p
\end{aligned}
\end{equation}
where $A_v\Bigl\vert\bmm\scalebox{0.6}{\Ygraph{i}{j}{k}}\emm\Bigr\rangle =\delta_{ijk}\Bigl\vert\bmm\scalebox{0.6}{\Ygraph{i}{j}{k}}\emm\Bigr\rangle$ is the vertex operator for imposing the particular string-net branching rule, such as the $\bZ_N$ branching rules: $i+j+k=0$ (mod $N$).
And $B_p=\sum^{N}_{s=0}a_sB^s_p$ is the ``magnetic--flux" operator defined for each hexagonal plaquette of the string--net lattice, with $a_s=d_s/D$ and $D=\sum_{i=0}^N d^2_i$. 
In general, a plaquette operator is constituted by a product of $6j$ symbols $F$,

\be
\begin{aligned}
B^s_p &\BLvert\Psix[1]{l}{e}\Brangle
= \sum_{e'_{A}}
F^{l_1 e^*_1 e_6}_{s^* e'_6 e'^*_1}F^{l_2 e^*_2 e_1}_{s^* e'_1 e'^*_2}F^{l_3 e^*_3 e_2}_{s^* e'_2 e'^*_3} F^{l_4 e^*_4 e_3}_{s^* e'_3 e'^*_4}F^{l_5 e^*_5 e_4}_{s^* e'_4 e'^*_5}F^{l_6 e^*_6 e_5}_{s^* e'_5 e'^*_6} \BLvert\Psix[1]{l}{e'}\Brangle
\end{aligned},
\ee
which is obviously very complicated.

\begin{figure}
    \begin{eqnarray*}
    \begin{tikzpicture}[x=12mm,y=6.94392mm]
      \tikzset{
        box/.style={line width=0.1cm,
          gray!50,
          regular polygon,
          regular polygon sides=6,
          minimum size=13.856mm,
          rotate=90,
         draw
        }
      }
    
      \tikzset{
        pointb/.style={circle,inner sep=0pt,minimum size=0.7mm,color=black,fill=black,line width=0.7mm,draw
        }
      }
    
      \tikzset{
        pointw/.style={circle,inner sep=0pt,minimum size=0.7mm,color=white,fill=white,line width=0.7mm,draw
        }
      }

      \tikzset{->-/.style={decoration={markings,mark=at position #1 with {\arrow{stealth}}},postaction={decorate}}}

    \draw[help lines,fill=cyan!10,thick] (-1.2,-1.2)--(8.7,-1.2)--(8.7,8.7)--(-1.2,8.7)--cycle;
    \foreach \i in {0,...,8} 
        \foreach \j in {0,...,2} {
                \node[box] at (1*\i,3*\j) {}; 
                \node[box] at (\i-0.5,3*\j+1.5) {};}           
    \draw[cyan!50,line width=0.14cm] plot[smooth,tension=0.9] coordinates {(1,6) (0.55,4.5) (1,3.05) (2,2.95) (2.5,1.5)};
    \node[line width=0.1cm,cyan!70,regular polygon, draw, regular polygon sides = 6, minimum size=11.35mm,rotate=90] at (2.5,1.5) {};
    \node[line width=0.1cm,cyan!70,regular polygon, draw, regular polygon sides = 6, minimum size=11.35mm,rotate=90] at (1,6) {};
    \node[below, cyan, scale=1.3] at (2.5,1.9) {$m$}; 
    \node[above, cyan, scale=1.3] at (1,5.6) {$m$}; 
    
    \draw[purple!50,line width=0.14cm] plot[smooth,tension=0.9] coordinates {(6.5,5.5) (6,5) (6,4) (6.5,3.5) (6.5,2.5) (6,2) (6,1) (5.5,0.5)};
    \draw[purple!70,line width=0.1cm] (6,5)--(6.5,5.5)--(6.5,6.5) (6.5,5.5)--(7,5);
    \draw[purple!70,line width=0.1cm] (5,1)--(5.5,0.5)--(6,1) (5.5,0.5)--(5.5,-0.5);
    \node[below, purple, scale=1.5] at (6.5,5.5) {$e$}; 
    \node[above, purple, scale=1.5] at (5.5,0.5) {$e$}; 
    \node[pointw] at (4.5,0.5) {};

    \draw [->-=0.7,black!70,line width=0.7mm] (-0.5,-0.5) -- (0,-1) node [pos=0.2, below] {};
    \draw [->-=0.7,black!70,line width=0.7mm] (0.5,-0.5) -- (0.5,0.5) node [pos=0.2, below] {};
    \draw [->-=0.7,black!70,line width=0.7mm] (0,1) -- (-0.5,0.5) node [pos=0.2, below] {};

    \draw [->-=0.7,black!70,line width=0.7mm] (0.5,-0.5) -- (0,-1) node [pos=0.2, below] {};
    \draw [->-=0.7,black!70,line width=0.7mm] (-0.5,-0.5) -- (-0.5,0.5) node [pos=0.2, below] {};
    \draw [->-=0.7,black!70,line width=0.7mm] (0,1) -- (0.5,0.5) node [pos=0.2, below] {};

    \foreach \i in {0,...,9} 
        \foreach \j in {0,...,2} {
                  \node[pointw] at (1*\i-0.5,3*\j+0.5) {};
                  \node[pointb] at (1*\i-1,3*\j+1) {};}
    
    \foreach \i in {1,...,8} 
        \foreach \j in {0,...,3} {
                  \node[pointw] at (1*\i-1,3*\j-1) {};
                  \node[pointb] at (1*\i-0.5,3*\j-0.5) {};};
                  
    \foreach \j in {0,...,3}{
                  \node[pointw] at (8,3*\j-1) {};
                  \node[pointb] at (-0.5,3*\j-0.5) {};
    }
    \foreach \j in {1,...,3}{
                  \node[pointw] at (-1,3*\j-1) {};
    }
    \foreach \j in {0,...,2}{
                  \node[pointb] at (8.5,3*\j-0.5) {};
    }

    \end{tikzpicture}
    \end{eqnarray*}
    \caption{\label{bulk} $\mathbb{Z}_N$ toric code model on the honeycomb lattice. The arrow diagram at the bottom-left corner defines the orientation for all edges. For a specific plaquette, the edges with the clockwise arrow belongs to $p_+$; the opposite ones belongs to $p_-$.}
    \end{figure}
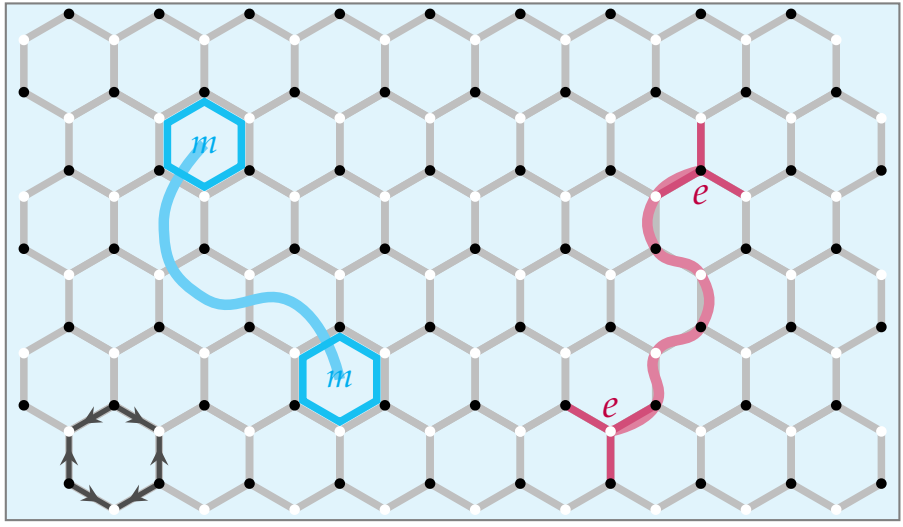

However, for the $\bZ_N$ topological order, there are $N$ different string types, labeled by $0,...,N-1$. 
In addition, the string is oriented. 
A natural way to fix the orientation is to use the sub-lattice of the honeycomb, $L_B$ and $L_W$, whose sites are colored black and white respectively as shown in Fig.~\ref{bulk}. 
We require that the arrow always points from a black vertex toward a white vertex to uniquely fix the orientation of each edge. 
After fully fixing the direction of each edge, the degrees of freedom of the original string can be viewed as the $\bZ_N$ spin degrees of freedom lying at the edge \cite{hung2012,schulz2012}. 
In the following, we use the $N$ orthonormal basis $\vert \mathsf{n}\rangle_{i}$ to label the associated spin degrees of freedom on the edge $i$, where $\mathsf{n}\in \mathbb{Z}_{N}$. 
Obviously, there is $a_s=1/N$ since the quantum dimension $d_i$ of the abelian topological orders is 1. 
And $F$ symbols satisfying the $\bZ_N$ branching rule: $s+e_i=e'_i$ (mod $N$) are 1, the others are 0. Thus the action of $B^s_p$ on the plaquette is simply shifting the string/spin type $e_i$ of every edge of the plaquette by $s$ (mod $N$ ).
This suggests that $B^s_p$ can be more naturally expressed as a product of spin shifting operators acting on the $\bZ_N$ spin of plaquette's edges 
\be
B^s_p 
= \prod_{i\in p_+}(X^{\dagger}_{i})^s\prod_{j\in p_-}(X_{j})^{s}.
\ee
Here we introduce the spin shifting operators $X_{i}$, which generalize the Pauli operator $\sigma^x$ to $N$-dimensional space.
For the edges in opposite orientation, the shifting actions need to be conjugated. 
For convenience, we also introduce operators $Z_{i}$ to measure the value of spin at the edge $i$, which generalize the Pauli operator $\sigma^z$ to $N$-dimensional space.
They are defined by,
$$
Z_{i}|\mathsf{n}\rangle_{i}=\omega^{\mathsf{n}}|\mathsf{n}\rangle_{i} \quad \text { and } \quad X_{i}|\mathsf{n}\rangle_{i}=|\mathsf{n}-1\rangle_{i}, \quad \text{for}~\omega=\mathrm{e}^{2\pi\mathrm{i}/N}.
$$
Instead of anti-commutation relations of Pauli operators, the operators $X_i$ and $Z_i$ on the same site satisfy the Weyl algebra's relations \cite{weyl1950},
\be
\begin{aligned}
Z_iX_i &=\omega X_iZ_i,\\
Z^N_i=X^N_i=1,\quad Z^{\dagger}_I &= Z^{N-1}_i,\quad X^{\dagger}_i=X^{N-1}_i.
\end{aligned}
\ee
while operators at different sites commute. 
Also, the operator $A_v$ can be expressed as a product of $Z_i$: $A_v=\frac{1}{N}\sum_{s=0}^{N-1}\prod_{i\in v}Z^s_{i}=\frac{1}{N}\sum_{s=0}^{N-1}A^s_v$, to make the groundstates satisfy branching rules.
Finally, we obtain the $\bZ_N$ Levin-Wen model written by spin operators, 
\be
    \begin{aligned}
    H^{\bZ_N}_{\text{bulk}}&=-\frac{1}{N}\sum_{v}\sum_{s=0}^{N-1}\prod_{i\in v}Z^s_{i}-\frac{1}{N}\sum_{p}\sum_{s=0}^{N-1}\prod_{i\in p_+}(X^{\dagger}_{i})^s\prod_{j\in p_-}(X_{j})^{s}\\
    &=-\frac{1}{N}\sum_{v}\sum_{s=0}^{N-1}A^s_v-\frac{1}{N}\sum_{p}\sum_{s=0}^{N-1}B^s_p.
    \end{aligned}
\ee

It is easy to check that this model is indeed exactly solvable by using the relations between $Z_i$ and $X_i$.
Since all the terms in the Hamiltonian $H^{\bZ_N}_{\text{bulk}}$ mutually commute, the groundstate subspace can be stabilized by simultaneously satisfying: $A^s_v |\text{GS}\rangle=1$ \& $B^s_p|\text{GS} \rangle =1$, $\forall$ $v$, $p$ and $s$. When some $A^s_v$ or $B^s_p$'s eigenvalues violate the above condition, these eigenstates with higher energy at local vertices or plaquettes can be regarded as some particle-like topological excitations. 
According to the different eigenvalues \{$\mathrm{e}^{\frac{2\pi\mathrm{i}}{N}\alpha}$, $\mathrm{e}^{\frac{2\pi\mathrm{i}}{N}\beta}$\} of $A^1_{v}$ and $B^1_{p}$ carried by these particles, they are usually known as $\bZ_N$ charge $\mathbf{e}^{\alpha}$ and $\bZ_N$ flux $\mathbf{m}^\beta$, where $\alpha$ and $\beta$ can be chosen from $0,\dots,(N-1)$. 
Obviously, these charges and fluxes naturally correspond to the objects $\cO_{\alpha,0}$ and $\cO_{0,\beta}$ in UMTC $\fZ_1(\rep(\bZ_N))$. 
They are not created by local perturbations but can be generated in pairs by string operators. 
String operators $S^{\bfe}$ and $S^{\bfm}$ can be defined by collecting the corresponding spin operators along the string, 
but note that the positive and negative orientation of the string $S^{\bfe}$ correspond to $X^{\dagger}$ and $X$, respectively.

\subsection{Boundary Lattice Model} \label{sec:duality}
In this subsection, we consider the case that the model is equipped with a boundary. 
It is well known that the $\bZ_N$ topological order can admit many gapped 
boundaries. 
The lattice realizations of many different kinds of gapped boundaries have been given in \cite{beigi2011,cong2017}. 
Here, we try to construct a general boundary Hamiltonian to recover many of the gapped boundaries of $H^{\bZ_N}_{\text{Bulk}}$ by tuning the parameters. 
In addition, it is predictable that the phase transitions between these different gapped boundaries naturally correspond to the gappable gapless boundaries. 
That's our final goal.

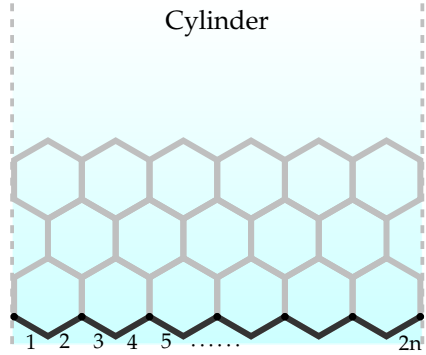
\begin{figure}[h]
    \begin{eqnarray*}
    \begin{array}{c}
        \begin{tikzpicture}[x=9mm,y=5.20785mm]
              \tikzset{
                box/.style={line width=0.08cm,
                  gray!50,
                  regular polygon,
                  regular polygon sides=6,
                  minimum size=10.3995mm,
                  rotate=90,
                 draw
                }
              }
            
              \tikzset{
                pointb/.style={circle,inner sep=0pt,minimum size=0.5mm,color=black,fill=black,line width=0.5mm,draw
                }
              }
            
              \tikzset{
                pointw/.style={circle,inner sep=0pt,minimum size=0.5mm,color=white,fill=white,line width=0.5mm,draw
                }
              }
          
              \tikzset{->-/.style={decoration={markings,mark=at position #1 with {\arrow{stealth}}},postaction={decorate}}}
          
          \fill[top color=white,bottom color=cyan!20,draw=white] (-0.53,-1.2) rectangle   (5.53,7.7);
          \draw[line width=0.05cm,gray!50,dashed] (-0.53,-1.2)--(-0.53,7.7) (5.53,-1.2)--(5.53,7.7);

          \foreach \i in {0,...,5} 
          \foreach \j in {0,...,1} {
                        \node[box] at (1*\i,3*\j) {}; 
          }         
          \foreach \i in {1,...,5} 
          \foreach \j in {0} {
                        \node[box] at (\i-0.5,3*\j+1.5) {};
          }

          \foreach \i in {0,1,...,5} {
            \draw[black!80,line width=0.08cm] (-0.5+\i,-0.5)--(0+\i,-1)--(0.5+\i,-0.5);
            };

           \node[scale=0.8] at (-0.25,-0.75) [below] {1};
           \node[scale=0.8]  at (0.25,-0.75) [below] {2};
           \node[scale=0.8]  at (0.75,-0.75) [below] {3};
           \node[scale=0.8]  at (1.25,-0.75) [below] {4};
           \node[scale=0.8]  at (1.75,-0.75) [below] {5};
           \node[scale=0.8]  at (2.5,-1) [below]{$\cdots\cdots$};
           \node[scale=0.8]  at (5.35,-0.8) [below] {2n};

           \foreach \i in {1,...,6} 
           \foreach \j in {0} {
                     \node[pointb] at (1*\i-0.5,3*\j-0.5) {};};
           \node[pointb] at (-0.5,-0.5) {};
            
            \node at (2.5,7) {Cylinder};

            \end{tikzpicture}
    \end{array}
    \end{eqnarray*}
\caption{\label{cylingder} $\bZ_N$ toric code with a boundary on a cylinder. Degrees of freedom at the boundary are in black and labeled by numbers. The $(2n+1)-th$ edge is identified with the $1st$ edge at the boundary. }
\end{figure}

Now let's begin with considering the model with a boundary $H^{\bZ_N}_{\text{bulk}}$+ $H_{\text{bdy}}$ on a cylinder, as shown in Fig.~\ref{cylingder}. 
We only deal with the boundary at the bottom; the upper part is the bulk of infinite length.
Moreover, the spins of edges in black are artificially and rigidly defined as boundary degrees of freedom.
Hence, our boundary Hamiltonian $H_{\text{bdy}}$ is assumed to only contain the local operators acting on these black edges. 
Finally, these operators must commute with $A^s_v$ and $B^s_p$ in order to preserve the exact solvability of $H^{\bZ_N}_{\text{bulk}}$.
If the boundary Hamiltonian $H_{\text{bdy}}=0$, the whole system will have a very large degeneracy that depends on the number $n$ of edges at the boundary. 
In principle, we want to remove most of the degeneracy from the boundary by including the fewest local operators into $H_{\text{bdy}}$. 
The remaining difficulty is how to find and represent all of the local operators that meet the aforementioned criteria.

\begin{figure}
    \begin{eqnarray*}
\begin{array}{c}
    \begin{tikzpicture}[x=9mm,y=5.20785mm]

          \tikzset{
            pointb/.style={circle,inner sep=0pt,minimum size=0.5mm,color=black,fill=black,line width=0.5mm,draw
            }
          }
        
          \tikzset{
            pointw/.style={circle,inner sep=0pt,minimum size=0.5mm,color=white,fill=white,line width=0.5mm,draw
            }
          }
      
          \tikzset{->-/.style={decoration={markings,mark=at position #1 with {\arrow{stealth}}},postaction={decorate}}}
      
      \fill[top color=white,bottom color=cyan!20,draw=white] (-0.53,-1.2) rectangle   (5.53,7.7);
      \draw[line width=0.05cm,gray!50,dashed] (-0.53,-1.2)--(-0.53,7.7) (5.53,-1.2)--(5.53,7.7);

        \draw[black,line width=0.05cm,dashed] (-0.5,-0.75)--(5.5,-0.75);


        \node at (2.5,7) {Cylinder};
        \draw [line width=0.5mm,color=blue,dashed](4.15,-0.72) arc (0:180: 1.5cm and 1.5cm);
        \draw[blue,line width=0.5mm] plot[smooth,tension=0.9] coordinates {(0.75,-1) (1,-0.5) (1.5,-1) (2,-0.5)(2.5,-1)(3,-0.5)(3.5,-1)(4,-0.5)(4.25,-1)};
        \node at (0.55,-1) [below] {$a$};
        \node at (4.55,-1) [below] {$b$}; 
        \node at (2.5,3) {$ \Downarrow$ push $S_{a,b}$} ;
        \end{tikzpicture}
\end{array}
\end{eqnarray*}
\caption{\label{bdy string} Local boundary operators can be represented by the string operators $S_{\text{bdy}}(a,b)$ labeled by blue string, where $S_{\text{bdy}}(a,b)$ can be obtained by pushing $S(a,b)$ labeled by dashed blue string to the boundary.}
\end{figure}
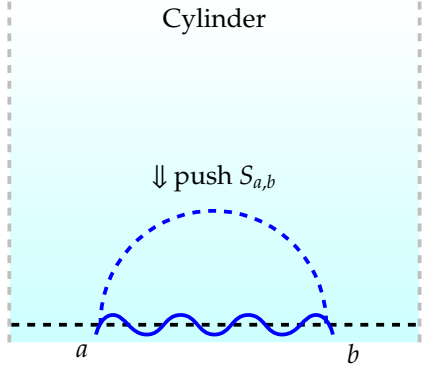

Usually, any operator that commutes with $\{A^s_v,B^s_p\}$ consists of closed-string operators in the bulk. 
But in the case with a boundary, open strings $S(a,b)$ with two endpoints $a$,$b$ outside the bulk are also available. 
Therefore, all we need to do is push all allowed strings onto the boundary so that we can get all local operators on the boundary that commute with $\{A^s_v,B^s_p\}$, as shown in Fig.~\ref{bdy string}. 
In the current model, only pushing open strings $S(a,b)$ can yield non-trivial $S_{\text{bdy}}(a,b)$. 
As the boundary is merely 1d, closed-string operators on the boundary must be identity, except for two non-contractible loop operators $L^{\bfe}$ and $L^{\bfm}$, as shown in Fig.~\ref{Loop}. However, the boundary operators generated by $L^{\bfe}$ and $L^{\bfm}$ are not local. There are two classes of string operators $S^{\bfe}$ and $S^{\bfm}$ in the $\bZ_N$ TO. When we push them onto the boundary, there are only two types of allowed configurations to ensure that the endpoints are outside the bulk and only act on the boundary spin, as shown in Fig.~\ref{allowed string}. 
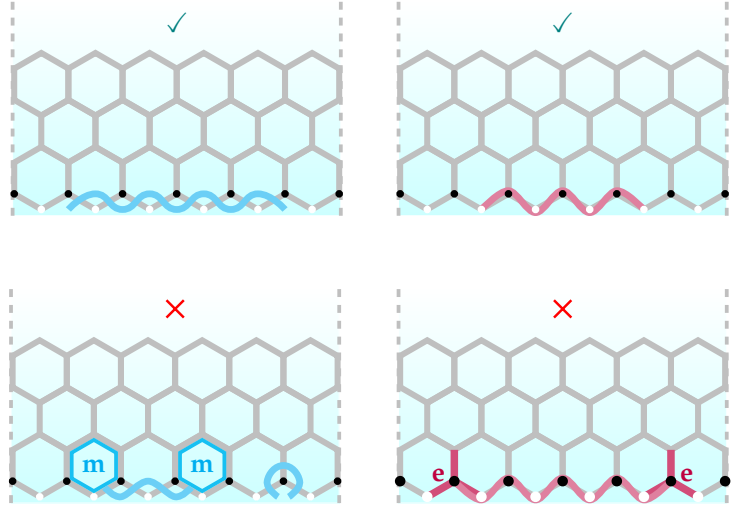
\begin{figure}
  \begin{eqnarray*}
  \begin{array}{c}
      \begin{tikzpicture}[x=7.2mm,y=4.16628mm]
            \tikzset{
              box/.style={line width=0.08cm,
                gray!50,
                regular polygon,
                regular polygon sides=6,
                minimum size=8.3196mm,
                rotate=90,
               draw
              }
            }
          
            \tikzset{
              pointb/.style={circle,inner sep=0pt,minimum size=0.5mm,color=black,fill=black,line width=0.5mm,draw
              }
            }
          
            \tikzset{
              pointw/.style={circle,inner sep=0pt,minimum size=0.5mm,color=white,fill=white,line width=0.5mm,draw
              }
            }
        
            \tikzset{->-/.style={decoration={markings,mark=at position #1 with {\arrow{stealth}}},postaction={decorate}}}
        
        \fill[top color=white,bottom color=cyan!20,draw=white] (-0.53,-1.2) rectangle   (5.53,5.7);
        \draw[line width=0.05cm,gray!50,dashed] (-0.53,-1.2)--(-0.53,5.7) (5.53,-1.2)--(5.53,5.7);

        \foreach \i in {0,...,5} 
        \foreach \j in {0,...,1} {
                      \node[box] at (1*\i,3*\j) {}; 
        }         
        \foreach \i in {1,...,5} 
        \foreach \j in {0} {
                      \node[box] at (\i-0.5,3*\j+1.5) {};
        }


         \foreach \i in {1,...,6} 
         \foreach \j in {0} {
                   \node[pointw] at (1*\i-1,3*\j-1) {};
                   \node[pointb] at (1*\i-0.5,3*\j-0.5) {};};
          \node[pointb] at (-0.5,-0.5) {};
          
          \draw[cyan!50,line width=0.09cm] plot[smooth,tension=0.9] coordinates {(0.5,-1) (1,-0.5) (1.5,-1) (2,-0.5)(2.5,-1)(3,-0.5)(3.5,-1)(4,-0.5)(4.5,-1)};
          \node at (2.5,5) {\Checkmark};

          \end{tikzpicture}
  \end{array}
  \begin{array}{c}
      \begin{tikzpicture}[x=7.2mm,y=4.16628mm]
            \tikzset{
              box/.style={line width=0.08cm,
                gray!50,
                regular polygon,
                regular polygon sides=6,
                minimum size=8.3196mm,
                rotate=90,
               draw
              }
            }
          
            \tikzset{
              pointb/.style={circle,inner sep=0pt,minimum size=0.5mm,color=black,fill=black,line width=0.5mm,draw
              }
            }
          
            \tikzset{
              pointw/.style={circle,inner sep=0pt,minimum size=0.5mm,color=white,fill=white,line width=0.5mm,draw
              }
            }
        
            \tikzset{->-/.style={decoration={markings,mark=at position #1 with {\arrow{stealth}}},postaction={decorate}}}
        
        \fill[top color=white,bottom color=cyan!20,draw=white] (-0.53,-1.2) rectangle   (5.53,5.7);
        \draw[line width=0.05cm,gray!50,dashed] (-0.53,-1.2)--(-0.53,5.7) (5.53,-1.2)--(5.53,5.7);

        \foreach \i in {0,...,5} 
        \foreach \j in {0,...,1} {
                      \node[box] at (1*\i,3*\j) {}; 
        }         
        \foreach \i in {1,...,5} 
        \foreach \j in {0} {
                      \node[box] at (\i-0.5,3*\j+1.5) {};
        }
        \node at (2.5,5) {\Checkmark};


          \draw[purple!50,line width=0.09cm] plot[smooth,tension=0.6] coordinates {(1,-1)(1.5,-0.4) (2,-1.1) (2.5,-0.4) (3,-1.1)(3.5,-0.4)(4,-1)};

          \foreach \i in {0,1,...,6} 
         \foreach \j in {0} {
                   \node[pointw] at (1*\i-1,3*\j-1) {};
                   \node[pointb] at (1*\i-0.5,3*\j-0.5) {};};
      \end{tikzpicture}
  \end{array}
  \end{eqnarray*}

  \begin{eqnarray*}
      \begin{array}{c}
          \begin{tikzpicture}[x=7.2mm,y=4.16628mm]
                \tikzset{
                  box/.style={line width=0.08cm,
                    gray!50,
                    regular polygon,
                    regular polygon sides=6,
                    minimum size=8.3196mm,
                    rotate=90,
                   draw
                  }
                }
              
                \tikzset{
                  pointb/.style={circle,inner sep=0pt,minimum size=0.5mm,color=black,fill=black,line width=0.5mm,draw
                  }
                }
              
                \tikzset{
                  pointw/.style={circle,inner sep=0pt,minimum size=0.5mm,color=white,fill=white,line width=0.5mm,draw
                  }
                }
            
                \tikzset{->-/.style={decoration={markings,mark=at position #1 with {\arrow{stealth}}},postaction={decorate}}}
            
            \fill[top color=white,bottom color=cyan!20,draw=white] (-0.53,-1.2) rectangle   (5.53,5.7);
            \draw[line width=0.05cm,gray!50,dashed] (-0.53,-1.2)--(-0.53,5.7) (5.53,-1.2)--(5.53,5.7);

            \foreach \i in {0,...,5} 
            \foreach \j in {0,...,1} {
                          \node[box] at (1*\i,3*\j) {}; 
            }         
            \foreach \i in {1,...,5} 
            \foreach \j in {0} {
                          \node[box] at (\i-0.5,3*\j+1.5) {};
            }


             \foreach \i in {1,...,6} 
             \foreach \j in {0} {
                       \node[pointw] at (1*\i-1,3*\j-1) {};
                       \node[pointb] at (1*\i-0.5,3*\j-0.5) {};};
              \node[pointb] at (-0.5,-0.5) {};
          
              \node[line width=0.05cm,cyan!70,regular polygon, draw, regular polygon sides = 6, minimum size=6.81mm,rotate=90] at (1,0) {};
              \node [cyan] at (1,0) {$\mathbf{m}$};
              \node[line width=0.05cm,cyan!70,regular polygon, draw, regular polygon sides = 6, minimum size=6.81mm,rotate=90] at (3,0) {};
              \node [cyan] at (3,0) {$\mathbf{m}$};
              \draw[cyan!50,line width=0.09cm] plot[smooth,tension=0.6] coordinates {(1.2,-0.7)(1.5,-1) (2,-0.5) (2.5,-1) (2.8,-0.7)};
              \draw[cyan!50,line width=0.09cm] plot[smooth,tension=0.9] coordinates {(4.4,-1) (4.2,-0.5)(4.5,0)(4.8,-0.5)(4.6,-1)};
              \node at (2.5,5) {\Cross};

              \end{tikzpicture}
      \end{array}
      \begin{array}{c}
          \begin{tikzpicture}[x=7.2mm,y=4.16628mm]
          \tikzset{
              box/.style={line width=0.08cm,
                  gray!50,
                  regular polygon,
                  regular polygon sides=6,
                  minimum size=8.3196mm,
                  rotate=90,
                  draw
                  }
          }
              
          \tikzset{
                  pointb/.style={circle,inner sep=0pt,minimum size=0.7mm,color=black,fill=black,line width=0.7mm,draw
                  }
          }
              
          \tikzset{
                  pointw/.style={circle,inner sep=0pt,minimum size=0.7mm,color=white,fill=white,line width=0.7mm,draw
                  }
          }
            
                \tikzset{->-/.style={decoration={markings,mark=at position #1 with {\arrow{stealth}}},postaction={decorate}}}
            
            \fill[top color=white,bottom color=cyan!20,draw=white] (-0.53,-1.2) rectangle   (5.53,5.7);
            \draw[line width=0.05cm,gray!50,dashed] (-0.53,-1.2)--(-0.53,5.7) (5.53,-1.2)--(5.53,5.7);

            \foreach \i in {0,...,5} 
            \foreach \j in {0,...,1} {
                          \node[box] at (1*\i,3*\j) {}; 
            }         
            \foreach \i in {1,...,5} 
            \foreach \j in {0} {
                          \node[box] at (\i-0.5,3*\j+1.5) {};
            }


          \draw[purple!70,line width=0.1cm] (0.5,-0.5)--(1,-1) (0.5,-0.5)--(0,-1) (0.5,0.5)--(0.5,-0.5);
          \node[left, purple, scale=1] at (0.5,-0.3) {$\mathbf{e}$};
          \draw[purple!70,line width=0.1cm] (4.5,-0.5)--(5,-1) (4.5,-0.5)--(4,-1) (4.5,0.5)--(4.5,-0.5) ;
          \node[right, purple, scale=1] at (4.5,-0.3) {$\mathbf{e}$};
          \draw[purple!50,line width=0.09cm] plot[smooth,tension=0.6] coordinates {(0.5,-0.5)(1,-1.1)(1.5,-0.4) (2,-1.1) (2.5,-0.4) (3,-1.1)(3.5,-0.4)(4,-1.1)(4.5,-0.5)};
              
          \node at (2.5,5) {\Cross};

          \foreach \i in {0,1,...,6} 
          \foreach \j in {0} {
                       \node[pointw] at (1*\i-1,3*\j-1) {};
                       \node[pointb] at (1*\i-0.5,3*\j-0.5) {};
                       };  
      \end{tikzpicture}
      \end{array}
  \end{eqnarray*}
  \caption{\label{allowed string} Configurations of boundary strings $S_{\text{bdy}}^{\bfe}$ and $S_{\text{bdy}}^{\bfm}$ that are in purple and cyan, respectively. Upper part: two types of allowed configurations, marked with green check mark. Bottom part: invalid configurations are marked by red cross, because the end points are inside the bulk or they have action on the internal degrees of freedom.}
\end{figure}
All valid configurations can be achieved by the product of a number of the two shortest boundary strings \Se{} and \Sm{}. 
So far, we have obtained all the local operators on the boundary that satisfy the criteria.
And they can be represented by \Se{} and \Sm{}. 
In other words, any legal $H_{\text{bdy}}$ can be viewed as a combination and product of \Se{} and \Sm{}. 
A typical example is $H^{\bfm}_{\text{bdy}}=-\sum_{j\in \text{odd link}}(Z_jZ_{j+1}+Z^\dagger_jZ^{\dagger}_{j+1})$ that realizes the $\bfm$-condensed boundary \cite{beigi2011,KK2012}, as shown in Fig.~\ref{gapped bdy}. $H^{\bfm}_{\text{bdy}}$ can be viewed as the sum of \Sm{} that describes the fluctuation of $\bfm$ along the boundary. 
Indeed, these terms give rise to the $\bfm$ condensation. 
As we expect, $N^n$-$fold$ degeneracy from the boundary is lifted by local operators \Sm{} so that only the topological $N$-$fold$ degeneracy from the semi-infinite open string $S^{\bfm}_{\text{INF}}$ remains. 
We will discuss this later.
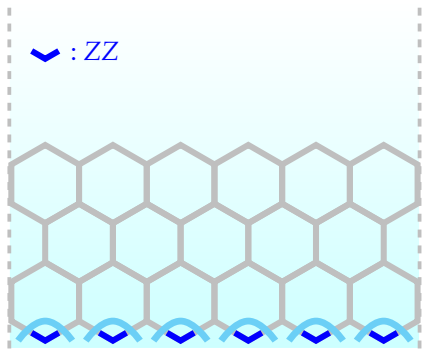
\begin{figure}
  \begin{eqnarray*}
\begin{array}{c}
  \begin{tikzpicture}[x=9mm,y=5.20785mm]
        \tikzset{
          box/.style={line width=0.08cm,
            gray!50,
            regular polygon,
            regular polygon sides=6,
            minimum size=10.3995mm,
            rotate=90,
           draw
          }
        }
      
        \tikzset{
          pointb/.style={circle,inner sep=0pt,minimum size=0.5mm,color=black,fill=black,line width=0.5mm,draw
          }
        }
      
        \tikzset{
          pointw/.style={circle,inner sep=0pt,minimum size=0.5mm,color=white,fill=white,line width=0.5mm,draw
          }
        }
    
        \tikzset{->-/.style={decoration={markings,mark=at position #1 with {\arrow{stealth}}},postaction={decorate}}}
    
    \fill[top color=white,bottom color=cyan!20,draw=white] (-0.53,-1.2) rectangle   (5.53,7.7);
    \draw[line width=0.05cm,gray!50,dashed] (-0.53,-1.2)--(-0.53,7.7) (5.53,-1.2)--(5.53,7.7);

    \foreach \i in {0,...,5} 
    \foreach \j in {0,...,1} {
                  \node[box] at (1*\i,3*\j) {}; 
    }         
    \foreach \i in {1,...,5} 
    \foreach \j in {0} {
                  \node[box] at (\i-0.5,3*\j+1.5) {};
    }

    \draw[blue,line width=0.08cm] (-0.2,6.4)--(0,6.2)--(0.2,6.4) node[right] {:~$ZZ$};
    \foreach \i in {0,1,...,5} {
      \draw[blue,line width=0.08cm] (-0.2+\i,-0.8)--(0+\i,-1)--(0.2+\i,-0.8);
      };
  \foreach \i in {-1,...,4} {
     \draw[cyan!50,line width=0.09cm] plot[smooth,tension=0.9] coordinates {(0.6+\i,-1) (1+\i,-0.5) (1.4+\i,-1)};  }


  \end{tikzpicture}
\end{array}
\end{eqnarray*}
\caption{\label{gapped bdy} $\mathbf{m}$-condensed boundary of $\bZ_N$ topological order. }
\end{figure}

It is more interesting for us to consider the model only with nearest neighbor interactions. 
Therefore, the general boundary Hamiltonian $H^{\mathbb{Z}_N}_{\text{bdy}}$ with periodic conditions along the circle can be written down,
\be
\begin{aligned}
H^{\mathbb{Z}_N}_{\text{bdy}}&=-\sum \sum^{N-1}_{s=1}a_s(\text{\Se{}})^s-\sum\sum^{N-1}_{s=1}b_s (\text{\Sm{}})^s\\
&=-\sum_{i\in \text{even link} } \sum^{N-1}_{s=1}a_s(X^{\dagger}_iX_{i+1})^s-\sum_{j\in \text{odd link}}\sum^{N-1}_{s=1}b_s (Z_jZ_{j+1})^s,~~~~~\text{with~}X_{2n+1}= X_{1},
\end{aligned}
\label{Hbdy}
\ee
where the lattice translation symmetry has been applied. 
Since $a_s$ and $b_s$ are complex number, so the conditions on the couplings $a_s=a_{N-s}^{\ast}$ and $b_s=b_{N-s}^{\ast}$ are necessary to make the Hamiltonian hermitian. 
It is clear that the model captures many of the boundaries of the $\fZ_1(\rep(\bZ_N))$ constructed before. In addition to the $\bfm$-condensed boundary($b_1=b_{n-1}=1,$ others $=0$) and the $\bfe$-condensed boundary(($a_1=a_{n-1}=1,$ others $=0$)), the gapped boundaries induced by the subgroups of $\bZ_N$ are also included \cite{beigi2011}. 
In the Landau's paradigm, they are seen as partially symmetry broken phases. Thus, it can be expected that the model \eqref{Hbdy} will have a rich phase diagram.

Before proceeding further, we make an important observation about the boundary model. Obviously, the model has two different $\bZ_N$ symmetries generated by \Se{} and \Sm{} operators respectively. 
The corresponding symmetry operators are written as $\prod_{i=1}^{n}X^{\dagger}_{2i}X_{2i+1}$ and $\prod_{i=1}^{n}Z_{2i}Z_{2i+1}$. We call them $\bZ_N$ and ${\bZ}^{\vee }_N$. Indeed, they are the loop operators $L^{\bfe}_{\text{bdy}}$ and $L^{\bfm}_{\text{bdy}}$ (as shown in Fig.~\ref{Loop}) respectively, from the viewpoint of bulk. 
It means that the two $\bZ_N$ symmetries of the 1+1D boundary theory come from the $large$ $gauge$ $transformations$ of the $\bZ_N$ gauge theory in the 2+1D bulk. 
The symmetry sectors of the boundary Hilbert space $\CH_{\text{bdy}}$ should match up with the topological sectors (anyons) of the bulk Hilbert space $\CH_{\text{bulk}}$. 
Therefore, the Hilbert space $\CH_{\text{bdy}}$ has a decomposition as $\CH_{\text{bdy}}=\bigoplus_{\alpha,\beta}\CH^{\alpha,\beta}_{\text{bdy}}$, where $\alpha$ and $\beta$ are the charge and flux number.
Let's recall the example $\bfm$-condensed boundary $H^{\bfm}_{\text{bdy}}=-\sum_{j\in \text{odd link}}(Z_jZ_{j+1}+Z^\dagger_jZ^{\dagger}_{j+1})$ with $\bZ_N$ symmetry; another $\bZ^{\vee }_N$ symmetry is identity for the periodic case. For this example, there are $N$-$fold$ degeneracy states $\vert \Psi_{\alpha}\rangle$ labeled by the eigenvalue $\alpha$ of the $\bZ_N$ symmetry operator $L^{\bfe}$. These different symmetry sector $\alpha$ of  $\bZ_N$ can be constructed by inserting $\alpha$ semi-infinite open strings $S^{\bfm}_{\text{INF}}$, as shown in Fig.~\ref{Loop}. Because $S^{\bfm}_{\text{INF}}$ commutes with $H^{\bZ_N}_{\text{bulk}}$ and $H^{\bfm}_{\text{bdy}}$, it won't change the energy of the system. 
Viewed from the boundary, the  $N$-$fold$ degeneracy can be kept until a phase transition occurs as long as the $\bZ_N$ is not broken. 
In Landau's paradigm, the degeneracy is attributed to the $spontaneous$ $symmetry$ $breaking$. 
On the other hand, the viewpoint from bulk provides a topological interpretation. 
Since $S^{\bfm}_{\text{INF}}$ only can be detected by large loop $L^{\bfe}$ rather than any local operators in the bulk, the $N$-$fold$ degeneracy is topological robust against any local perturbation in bulk. 
The two statements can be related when we push $L^{\bfe}$ onto the boundary. $L^{\bfe}$ will become the symmetry constraint on the local perturbations of the boundary. 
That shows a holographic point of view: symmetry originates from the topology in one higher dimension \cite{CJKYZ2020,JW2020,JW2020metallic,kong2020,lichtman2021,chatterjee2022}. 
If we insert another semi-infinite open string $(S^{\bfe}_{\text{INF}})^{\beta}$ into the system instead of $S^{\bfm}_{\text{INF}}$, the energy of the system will increase because $S^{\bfm}_{\text{bdy}}(S^{\bfe}_{\text{INF}})^{\beta}\vert \Psi_{\text{bdy}}\rangle=\omega^{\beta}S^{\bfe}_{\text{INF}}\vert \Psi_{\text{bdy}}\rangle$ at the cross point. But if we also replace the $S^{\bfm}_{\text{bdy}}$ of $H^{\bfm}_{\text{bdy}}$ at the cross point with $\omega^{-\beta} S^{\bfm}_{\text{bdy}}$, $S^{\bfe}_{\text{INF}}\vert \Psi_{\text{bdy}}\rangle$ will be the one with lower energy. It actually assumes a twisted boundary condition at the cross point. We'll clarify this later.


\begin{figure}
  \begin{eqnarray*}
              \begin{array}{c}
                  \begin{tikzpicture}[scale=0.3]
                     
                      \fill[left color=cyan!20,right color=cyan!20,middle color=cyan!5,shading=axis,opacity=1] (-4,10) -- (-4,0) arc (180:360:4cm and 2cm) -- (4,10) arc (360:180:4cm and 2cm);
                      \draw [dashed,line width=0.5mm, black!60](4,0) arc (0:180:4cm and 2cm);
                      \draw [line width=0.5mm](-4,0) arc (180:360:4cm and 2cm);
                      \draw [dashed,cyan!50,line width=0.8mm](4,1) arc (0:180:4cm and 2cm);
                      \draw [cyan,line width=0.8mm](-4,1) arc (180:360:4cm and 2cm) node[right] {$L^{\mathbf{m}}$};
                      \draw [dashed,purple!50,line width=0.8mm](4,2)  arc (0:180:4cm and 2cm);
                      \draw [purple,line width=0.8mm](-4,2)  arc (180:360:4cm and 2cm) node[right] {$L^{\mathbf{e}}$};
                      \draw [line width=0.5mm,gray!50,dashed] (-4,0)--(-4,10);
                      \draw [line width=0.5mm,gray!50,dashed] (4,0)--(4,10);
                  \end{tikzpicture}
              \end{array}
              \begin{array}{c}
                  \begin{tikzpicture}[scale=0.3]
          
                      \fill[left color=cyan!20,right color=cyan!20,middle color=cyan!5,shading=axis,opacity=1] (-4,10) -- (-4,0) arc (180:360:4cm and 2cm) -- (4,10) arc (360:180:4cm and 2cm);
                      \draw [dashed,line width=0.5mm, black!60](4,0) arc (0:180:4cm and 2cm);
                      \draw [line width=0.5mm](-4,0) arc (180:360:4cm and 2cm);
                      \draw [dashed,purple!50,line width=0.8mm](4,1) arc (0:180:4cm and 2cm);
                      \draw [purple,line width=0.8mm](-4,1) arc (180:360:4cm and 2cm) node[right] {$L^{\bfe}$};
                      \draw [cyan,line width=0.8mm] (0,-2)--(0,8);
                      \node [cyan] at (0,9)[right] {$S^{\bfm}_{\text{INF}}$};
                      \draw [line width=0.5mm,gray!50,dashed] (-4,0)--(-4,10);
                      \draw [line width=0.5mm,gray!50,dashed] (4,0)--(4,10);
                  \end{tikzpicture}
              \end{array}   
   \end{eqnarray*}      
   \caption{\label{Loop} Left part: two kinds of non-contractible loop operators $L^{\bfe}$ and $L^{\bfm}$ wrapped around the cylinder. Right part: the semi-infinite open string $S^{\bfm}_{\text{INF}}$ can be detected by loop $L^{\bfe}$}.
\end{figure}
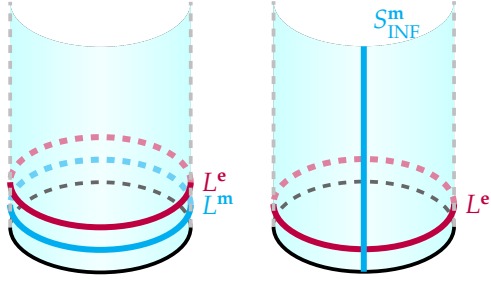

Let's return to the general case. In the $\CH^{\alpha,1}_{\text{bdy}}$, the model Eq.~\eqref{Hbdy} can be exactly mapped to the $\mathbb{Z}_N$ chiral clock model through the Kramers-Wannier transformation \cite{kogut1979,fendley2012},

\be
\begin{aligned}
  X^{\dagger}_iX_{i+1}\longrightarrow \tilde{X}_{i+\frac{1}{2}},&\quad Z_jZ_{j+1}\longrightarrow \tilde{Z}^{\dagger}_{j-\frac{1}{2}}\tilde{Z}_{j+\frac{1}{2}},\\
  \bZ_N\longrightarrow\prod \tilde{X}_{i+\frac{1}{2}},&\quad \bZ^{\vee}_N\longrightarrow \text{identity},
\end{aligned}
\ee
\be
\begin{aligned}
  \tilde{H}^{\mathbb{Z}_N}_{\text{bdy}} &=-\sum_{i\in \text{even link} } \sum^{N-1}_{s=1}a_s(\tilde{X}_{i+\frac{1}{2}})^s-\sum_{j\in \text{odd link}}\sum^{N-1}_{s=1}b_s (\tilde{Z}^{\dagger}_{j-\frac{1}{2}}\tilde{Z}_{j+\frac{1}{2}})^s\\
   &=-\sum_{i\in \bullet } \sum^{N-1}_{s=1}a_s(\tilde{X}_{i})^s-\sum_{i\in \bullet}\sum^{N-1}_{s=1}b_s (\tilde{Z}^{\dagger}_{i}\tilde{Z}_{i+1})^s,
\end{aligned}
\ee
where $\bZ^{\vee}_N=L^{\bfm}_{\text{bdy}}\rightarrow\mathbb{1}$ after transformation. 
That's why we have to work inside the $\CH^{\alpha,1}_{\text{bdy}}$. But if we first remove the identification of $\frac{1}{2}$-site and $(2n+\frac{1}{2})$-site, an extra degree of freedom can help us realize the mapping in every $\CH^{\alpha,\beta}_{\text{bdy}}$. At this point, we have
\be
\begin{aligned}
    \bZ^{\vee}_N\longrightarrow \tilde{Z}^{\dagger}_{\frac{1}{2}}\tilde{Z}_{2n+\frac{1}{2}}.
\end{aligned}
\ee
In addition to the condition of the symmetry sector: $\bZ_N|\Psi_{\text{bdy}}\rangle\rightarrow\prod\tilde{X}_{i+\frac{1}{2}}|\tilde{\Psi}_{\text{bdy}}\rangle=\alpha|\tilde{\Psi}_{\text{bdy}}\rangle$,
we also have a twisted boundary condition: $\bZ^{\vee}_N|\Psi_{\text{bdy}}\rangle\rightarrow\tilde{Z}^{\dagger}_{\frac{1}{2}}\tilde{Z}_{2n+\frac{1}{2}}|\tilde{\Psi}_{\text{bdy}}\rangle=\beta|\tilde{\Psi}_{\text{bdy}}\rangle$, for the subspace $\CH^{\alpha,\beta}_{\text{bdy}}$. Therefore, the mapping is exact only in every subspace $\CH^{\alpha,\beta}_{\text{bdy}}$. The ($\alpha$, $\beta$) of  $\CH^{\alpha,\beta}_{\text{bdy}}$ corresponds to the symmetry sector $\alpha$ and twisted boundary condition $\beta$ of the mapped model $\tilde{H}^{\bZ_N}_{\text{bdy}}$. Working in the $\CH^{\alpha,\beta}_{\text{bdy}}$ is more convenient. In addition to that, it also has real physical significance. Assuming that the cylinder has a finite length, we imagine a process in which a pair of $\bfe$ are created in the bulk and one of them is moved to the boundary. The boundary state must be in the corresponding sector $\CH^{\beta}_{\text{bdy}}$. We can also directly change the sign of $A^s_v$ and $B^s_p$ in the $H^{\bZ_N}_{\text{bulk}}$ to make the system favor any topological sector $(\alpha,\beta)$ that we want. Therefore, the symmetry of the boundary is not independent but is determined by the topological defect in the bulk \cite{CJKYZ2020,lichtman2021}. The model is essentially a ($\alpha$, $\beta$)-constrained $\bZ_N$ chiral clock model as depicted in Fig.~\ref{gapless bdy}, which is distinct from the pure 1+1D $\bZ_N$ chiral clock model. It is a sign that such a theory describing the boundary of 2+1D $\bZ_N$ topological order has a non-invertible gravitational anomaly \cite{JW2020}.

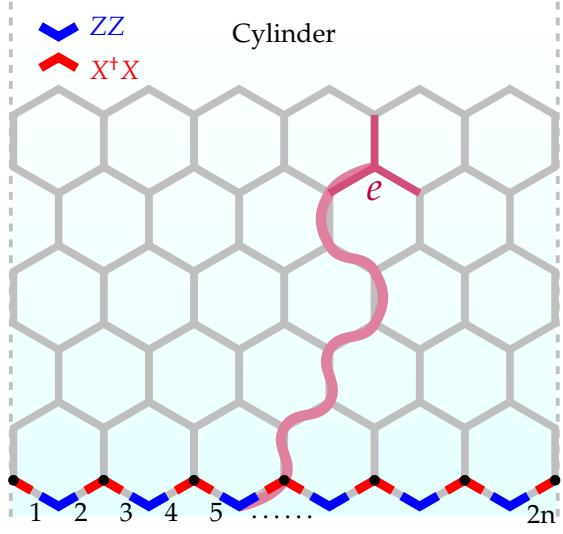
\begin{figure}
  \begin{eqnarray*}
  \begin{tikzpicture}[x=12mm,y=6.94392mm]
    \tikzset{
      box/.style={line width=0.1cm,
        gray!50,
        regular polygon,
        regular polygon sides=6,
        minimum size=13.856mm,
        rotate=90,
       draw
      }
    }
  
    \tikzset{
      pointb/.style={circle,inner sep=0pt,minimum size=0.7mm,color=black,fill=black,line width=0.7mm,draw
      }
    }
  
    \tikzset{
      pointw/.style={circle,inner sep=0pt,minimum size=0.7mm,color=white,fill=white,line width=0.7mm,draw
      }
    }

    \tikzset{->-/.style={decoration={markings,mark=at position #1 with {\arrow{stealth}}},postaction={decorate}}}

\fill[top color=cyan!0,bottom color=cyan!10,draw=white] (-0.53,-1.2) rectangle   (5.53,8.7);
\draw[line width=0.05cm,gray!50,dashed] (-0.53,-1.2)--(-0.53,8.7) (5.53,-1.2)--(5.53,8.7);

\foreach \i in {0,...,5} 
\foreach \j in {0,...,2} {
              \node[box] at (1*\i,3*\j) {}; 
}         
\foreach \i in {1,...,5} 
\foreach \j in {0,...,1} {
              \node[box] at (\i-0.5,3*\j+1.5) {};
}

  
 \draw[purple!50,line width=0.14cm] plot[smooth,tension=0.9] coordinates {(3.5,5.5) (3,5) (3,4) (3.5,3.5) (3.5,2.5) (3,2) (3,1) (2.5,0.5) (2.5,-0.5) (2,-1)};
 \draw[purple!70,line width=0.1cm] (3,5)--(3.5,5.5)--(3.5,6.5) (3.5,5.5)--(4,5);
 \node[below, purple, scale=1.5] at (3.5,5.5) {$e$}; 

      \foreach \i in {0,1,...,5} {
      \draw[blue,line width=0.12cm] (-0.2+\i,-0.8)--(0+\i,-1)--(0.2+\i,-0.8);
      };
      \draw[red,line width=0.12cm] (-0.5,-0.5)--(-0.3,-0.7);
      \foreach \i in {0,1,...,4} {
      \draw[red,line width=0.12cm] (0.3+\i,-0.7)-- (0.5+\i,-0.5)--(0.7+\i,-0.7);
      };
      \draw[red,line width=0.12cm] (5.3,-0.7)--(5.5,-0.5);
      \draw[blue,line width=0.12cm] (-0.2,8.2)--(0,8)--(0.2,8.2) node[right] {$ZZ$};
      \draw[red,line width=0.12cm] (-0.2,7.4)-- (0,7.6)--(0.2,7.4) node[right] {$X^{\dagger}X$};
      \node at (2.5,8) {Cylinder};
      \node at (-0.25,-0.75) [below] {1};
      \node at (0.25,-0.75) [below] {2};
      \node at (0.75,-0.75) [below] {3};
      \node at (1.25,-0.75) [below] {4};
      \node at (1.75,-0.75) [below] {5};
      \node at (2.5,-1) [below]{$\cdots\cdots$};
      \node at (5.35,-0.8) [below] {2n};
      
      \foreach \i in {1,...,6} 
      \foreach \j in {0} {
                \node[pointb] at (1*\i-0.5,3*\j-0.5) {};};
       \node[pointb] at (-0.5,-0.5) {};

  \end{tikzpicture}
  \end{eqnarray*}
  \caption{\label{gapless bdy} The ($\alpha$, $\beta$)-constrained boundary lattice model}
\end{figure}
By tuning the parameters $a_s$ and $b_s$, the boundary model may undergo a phase transition described by non-chiral conformal field theory. However, the $H^{\bZ_N}_{\text{bdy}}$ actually has many different critical points for large $N$, which have many different additional symmetries. We need to further specify the system's FFA  $W$, $i.e.$, $local$ $quantum$ $symmetry$. It is very difficult to find all of them. In this paper, we consider a special case,  a $\mathbb{Z}_N$ symmetric integrable parameter point: $a_s=b_s=\frac{1}{\sin{s\pi/N}}$, which is the self-dual Fateev-Zamolodchikov spin chain described by $\mathbb{Z}_N$ parafermion CFT \cite{fateev1982,ZF1985,jimbo1986}. For any $N$, it always can be a critical point between $\bZ_N$ spontaneous symmetry breaking phase($\bfe$-condensed) and $\bZ_N$ symmetric phase($\bfm$-condensed) of the boundary. In order to handle this model analytically, we only work with the Fateev-Zamolodchikov critical point in the following. Such a known result helps us determine the local quantum symmetry of the boundary as $(V_{\mathrm{PF}_N}\otimes_{\bC}\overline{V_{\mathrm{PF}_N}})$.  We expect that the observables of this critical point can be captured by the gapless boundary theory of TO.

\subsection{Low-Energy Effective Field Theory and Partition Functions}
In this section, we'll try to obtain the low-energy effective field theory of the self-dual Fateev-Zamolodchikov critical point at the boundary. As mentioned before, the Fateev-Zamolodchikov spin chain is an integrable model which can be extended to the $\bZ_N$ parafermion CFT by taking the thermodynamic limit \cite{fateev1982,ZF1985,jimbo1986}. It is more convenient to calculate the partition function within the framework of CFT. And the parafermion CFT can be constructed through the $\widehat{\mathrm{SU}}(2)_N/\widehat{\mathrm{U}}(1)_{2N}$ coset procedure \cite{goddard1986,francesco2012}, hence it is also known as the $\widehat{\mathrm{SU}}(2)_N/\widehat{\mathrm{U}}(1)_{2N}$ coset CFT with central charge,
\begin{equation}
c\left(\widehat{\mathrm{SU}}(2)_N/\widehat{\mathrm{U}}(1)_{2N}\right)=c\left(\widehat{\mathrm{SU}}(2)_N\right)-c\left(\widehat{\mathrm{U}}(1)_{2N}\right) \\
=\frac{3N}{N+2}-1.
\end{equation}
The coset formalism is useful for us to construct the partition functions on the subspace $\CH^{\alpha,\beta}_{\text{bdy}}$.
For a simple introduction to the coset construction, please see the Appendix.~\ref{coset}. Here we directly list the results of the coset construction. The representation $\ell$ of $\widehat{\mathrm{SU}}(2)_N$ can be decomposed into a direct sum of representations $m$ of $\widehat{\mathrm{U}}(1)_{2N}$. In physical, the $\widehat{\mathrm{SU}}(2)_N$  WZW model can be viewed as composed of two building blocks: a $\widehat{\mathrm{SU}}(2)_N/\widehat{\mathrm{U}}(1)_{2N}$ piece associated with the parafermions and a $\widehat{\mathrm{U}}(1)_{2N}$ factor associated with a free boson. 
\begin{equation}
\begin{aligned}
(\ell) \mapsto \bigoplus_{m=-\ell}^{\ell}(m)_{\ell} &\quad(\text{branching condition: }\ell-m=0 \bmod 2),\\
\CH_{\mathrm{SU(2)}_N \mathrm{wzw}}^{\ell}&=\CH_{\mathrm{para}}^{\ell, m} \times \CH_{\mathrm{boson}}^{m}~,\\
\text{ $0\leq\ell\leq N$},&~~~~\text{$-N+1 \leq m \leq N$}.~~~~~
\end{aligned}
\end{equation}
where $\CH_{\mathrm{\mathrm{SU(2)_N wzw}}}^{\ell}$ denotes the subspace of the Hilbert space of the $\widehat{\mathrm{SU}}(2)_N$ current algebra which contains the states in the highest weight representation with isospin $\ell$. Also, $\CH_{\mathrm{para}}^{\ell, m}$ denotes the parafermionic states obtained from the highest weight parafermionic state which have $\mathbb{Z}_N$ charge $m$ mod $N$. $\CH_{\mathrm{boson}}^{m}$ is a subspace of the bosonic theory containing the states of momenta $m$. 

In the following, we use $\ell$, $m$ and $\{\ell;~m\}$ to label the $\widehat{\mathrm{SU}}(2)_N$ representations, the $\widehat{\mathrm{U}}(1)_{2N}$ representations and the coset field (parafermion primary field), respectively. To avoid confusion with the $\widehat{\mathrm{SU}}(2)_N$ characters $\chi^{(N)}_{\ell}(q)$, the $\widehat{\mathrm{U}}(1)_{2N}$ characters is indicated by $K^{(N)}_m(q)$, where $q=\mathrm{e}^{2\pi\mathrm{i}\tau}$ and $\tau$ is modular parameter.

Following the procedure of coset construction $\hat{\mathrm{g}}_k/\hat{\mathrm{p}}_{x_e k}$ in Appendix.~\ref{coset}, the character decomposition is
\begin{equation}
\begin{aligned}
\chi_{\ell}^{(N)}(q)&=\sum_{m=-N+1}^{N} \chi_{\{\ell;~m\}}(q) K_{m}^{(N)}(q),\\
\end{aligned}
\end{equation}
with character identity,
\begin{equation}
\chi_{\{\ell;~m\}}(q)=\chi_{\{N-\ell;~m+N\}}(q)=\chi_{\{N-\ell;~m-N\}}(q).
\end{equation}

The fractional part of conformal weight of $\{\ell;~m\}$ is,
\begin{equation}
h_{\{\ell;~m\}}~\text{mod $1$} =h_{\ell}-h_{m}=\frac{\ell(\ell+2)}{4(N+2)}-\frac{m^2}{4N}.
\end{equation}

Modular transformation matrices of coset characters $ \chi_{\{\ell;~m\}}$,
\begin{equation}
\begin{aligned}
\mathcal{S}^{\widehat{\mathrm{SU}}(2)_N/\widehat{\mathrm{U}}(1)_{2N}}_{\{\ell;~m\},\left\{\ell^{\prime} ;~m^{\prime}\right\}}&=\mathcal{S}_{\ell,\ell^{\prime}}^{\widehat{\mathrm{SU}}(2)_N} \overline{\mathcal{S}}_{m,m^{\prime}}^{\widehat{\mathrm{U}}(1)_{2N}} \\
\mathcal{T}^{\widehat{\mathrm{SU}}(2)_N/\widehat{\mathrm{U}}(1)_{2N}}_{\{\ell;~m\},\left\{\ell^{\prime} ;~m^{\prime}\right\}}&=\mathcal{T}_{\ell,\ell^{\prime}}^{\widehat{\mathrm{SU}}(2)_N} \overline{\mathcal{T}}_{m,m^{\prime}}^{\widehat{\mathrm{U}}(1)_{2N}}\\
\mathcal{S}_{\ell,\ell^{\prime}}^{\widehat{\mathrm{SU}}(2)_N} &=\sqrt{\frac{2}{2+N}}\sin\frac{\pi(\ell+1)(\ell^{\prime}+1)}{N+2}\\
\overline{\mathcal{S}}_{m,m^{\prime}}^{\widehat{\mathrm{U}}(1)_{2N}}&=\frac{1}{\sqrt{2N}}\mathrm{e}^{\frac{\pi\mathrm{i}mm^{\prime}}{N}}\\
\mathcal{T}_{\ell,\ell^{\prime}}^{\widehat{\mathrm{SU}}(2)_N} &=\mathrm{e}^{2\pi\mathrm{i}(\frac{\ell(\ell+2)}{4(N+2)}-\frac{c\left(\widehat{\mathrm{SU}}(2)_N\right)}{24})}\delta_{\ell,\ell^{\prime}}\\
\overline{\mathcal{T}}_{m,m^{\prime}}^{\widehat{\mathrm{U}}(1)_{2N}}&=\mathrm{e}^{-2\pi\mathrm{i}(\frac{m^2}{4N}-\frac{c\left(\widehat{\mathrm{U}}(1)\right)}{24})}\delta_{m,m^{\prime}}.
\end{aligned}
\end{equation}
To obtain the partition functions on the subspace  $\CH^{\alpha,\beta}_{\text{bdy}}$, we need to consider the so-called twisted partition function $\mathsf{Z}_{r,s}$ ($r$ and $t$ are defined modulo $N$), which describes the CFT on the torus with $r$-twisted spatial boundary condition and $t$-twisted temporal boundary condition \cite{gepner1987}. Let's begin with a modular invariant untwisted coset partition function. For simplicity, we take the fully diagonal modular invariant solution as an example,
\begin{equation}
\mathsf{Z}_{0,0}=\frac{1}{2}\sum_{\{\ell;~m\}}\chi_{\{\ell;~m\}}\overline{\chi}_{\{\overline\ell;~\overline m\}}=\frac{1}{2}\sum_{\{\ell;~m\}}\chi_{\{\ell;~m\}}\overline{\chi}_{\{\ell;~ m\}}.
\end{equation}
The twisted partition functions $\mathsf{Z}_{0,t}$ is easier to calculate than the general case. The $t$ twist in temporal direction can be interpreted as a $\bZ_N$ symmetry line defect wrapped around the spatial circle of the torus. Therefore, we just need to insert the $\bZ_N$
symmetry operator $Q^t$ into the trace over the Hilbert space $\CH_{\mathrm{para}}^{\ell, m}$ to implement this symmetry defect line,
\begin{eqnarray}
        \begin{array}{c}
\begin{tikzpicture}[scale=0.7]
    \tikzset{->-/.style={decoration={markings,mark=at position #1 with {\arrow{stealth}}},postaction={decorate}}}
    \draw [help lines,line width=0.3mm] (0,0) -- (0,2)--(2,2)--(2,0)--cycle;
    \draw [->-=0.55,line width=0.5mm] (0,1)--(2,1)  node [pos=0.5, above] {$Q^t$};
\end{tikzpicture}
\end{array}\mathsf{Z}_{0,t}=\mathbf{Tr}_{\CH_{\mathrm{para}}^{\ell, m}}Q^t q^{L_0-\frac{c}{24}}\overline q^{\overline L_0-\frac{c}{24}}=\frac{1}{2}\sum_{\{\ell;~m\}}\mathrm{e}^{2\pi\mathrm{i}tm/N}\chi_{\{\ell;~m\}}\overline{\chi}_{\{\ell;~ m\}}.
\end{eqnarray}
As the $\bZ_N$ symmetry operator, $Q^t$ measures the $\bZ_N$ charge of the field $\{\ell,m\}$ with $\mathrm{e}^{2\pi\mathrm{i}tm/N}$ as the result. Since the $S$ transformation can exchange the temporal and spatial directions of the torus, the twisted partition functions $\mathsf{Z}_{0,t}$ will transform to  $\mathsf{Z}_{-t,0}$,
\begin{eqnarray*}
    \begin{array}{c}
\begin{tikzpicture}[scale=0.7]
\tikzset{->-/.style={decoration={markings,mark=at position #1 with {\arrow{stealth}}},postaction={decorate}}}
\draw [help lines,line width=0.3mm] (0,0) -- (0,2)--(2,2)--(2,0)--cycle;
\draw [->-=0.55,line width=0.5mm] (0,1)--(2,1)  node [pos=0.5, above]{};
\end{tikzpicture}
\end{array} \mathsf{Z}^{\mathrm{para}}_{0,t}~~~\xrightarrow{S}~~~
\begin{array}{c}
    \begin{tikzpicture}[scale=0.7]
    \tikzset{->-/.style={decoration={markings,mark=at position #1 with {\arrow{stealth}}},postaction={decorate}}}
    \draw [help lines,line width=0.3mm] (0,0) -- (0,2)--(2,2)--(2,0)--cycle;
    \draw [->-=0.55,line width=0.5mm] (1,0)--(1,2)  node [pos=0.5, right]{};
    \end{tikzpicture}
    \end{array}\mathsf{Z}^{\mathrm{para}}_{-t,0}.
\end{eqnarray*}

By the modular transformation, we can get all the twisted partition functions $\mathsf{Z}_{r,t}$ \cite{gepner1987}.
\begin{equation}
\mathsf{Z}_{r, t}=\frac{1}{2}\mathrm{e}^{-2 \pi \mathrm{i}  r t / N} \sum_{\{\ell;~m\}} \mathrm{e}^{2 \pi \mathrm{i} t m / N} \chi_{\{\ell;~m\}} \overline\chi_{\{\ell;~m-2r\}}.
\end{equation}

The gapless boundary of $\mathbb{Z}_N$ topological order is described by a set of $\mathbb{Z}_N$ parafermion partition functions on the $H^{\alpha,\beta}_{\text{bdy}}$. In physical, it comes from the fact that the boundary partition function will respond to the different topological defects in the bulk. For example, pushing the $\bfe^{\alpha}$ particle of $\mathbb{Z}_N$ topological order to the boundary can change the $\bZ_N$ symmetric boundary Hilbert space from $0$ symmetry sector to $\alpha$ sector (here, symmetry sector $\alpha$ corresponds to the charge $m$ of $\widehat{\mathrm{U}}(1)_{2N}$ CFT in the ${\{\ell;~m\}}$ notation). This is reflected in the change of the partition functions of the boundary theory. Therefore in contrast to the pure 1+1D parafermion CFT with trivial bulk that is described by only one modular invariant partition function, the boundary theory with a topological order as bulk requires a set of partition functions. Obviously, because $\CH_{\text{bdy}}=\bigoplus_{\alpha,\beta}\CH^{\alpha,\beta}_{\text{bdy}}$, these boundary partition functions can be labeled by $\Zb_N$ topological excitations: $\mathsf{Z}({\bfe^{\alpha}\bfm^\beta})$.

We know that $\bfm^\beta$ determines the spatial direction with the $\beta$ twisted condition and $\bfe^\alpha$ determines the symmetry sector($\mathbb{Z}_N$ charge) $\alpha$ mod $N$. Thus, the partition function with topological defect $\bfe^{\alpha}$ simply needs to collect all characters $\chi_{\{\ell;~m\}}$ with the charge $\alpha$,
\begin{equation}
\mathsf{Z}(\bfe^\alpha)=\sum_{\{\ell;~m\}\atop \&~m=\alpha\text{ mod } N}\chi_{\{\ell;~m\}} \overline\chi_{\{\ell;~m\}}.
\end{equation}
For the spatial $\beta$ twisted case, the $\mathbb{Z}_N$ symmetry acts on the coset field $\{\ell;~m~,\overline\ell;~m-2\beta\}_{\delta_{\ell,\overline\ell}}$ will get the charge $(m-\beta)$. Thus, 
\begin{equation}
\mathsf{Z}(\bfe^{\alpha} \bfm^\beta)=\sum_{\{\ell;~m\}\atop \&~m=(\alpha+\beta)\text{ mod } N}\chi_{\{\ell;~m\}}\overline\chi_{\{\ell;~m-2\beta\}}.
\label{eq:partition_function_p}
\end{equation}
Obviously, the partition functions in the anyon basis can also be obtained through a basis transformation from the twisted partition functions $\mathsf{Z}^{\mathrm{para}}_{r,t}$. To make this clearer, we reorganized the partition functions based on the $\mathbb{Z}_N$ charge basis $\alpha$,
\begin{equation}
\begin{aligned}
\mathsf{Z}_{0,t} &=\sum_{\alpha \in \mathbb{Z}_N}\mathrm{e}^{2\pi\mathrm{i}t\alpha/N}\mathsf{Z}(\bfe^\alpha),\\
\mathsf{Z}_{r, t}&=\sum_{\alpha,~\beta \in \mathbb{Z}_N}\mathrm{e}^{2\pi\mathrm{i}t\alpha/N}\mathsf{Z}(\bfe^\alpha \bfm^\beta)\delta_{\beta,r}.
\end{aligned}\label{Fourier}
\end{equation}   
Since Eq.\eqref{Fourier} above is the discrete Fourier series expansion, the $r,t$ twisted basis and the anyon $
\bfe^\alpha,\bfm^\beta$ basis are related by a discrete Fourier transformation,
\begin{equation}
\begin{aligned}
 \mathsf{Z}(\bfe^\alpha)&=1/N\sum_{s\in \mathbb{Z}_N}\mathrm{e}^{-2\pi\mathrm{i}t\alpha/N}\mathsf{Z}_{0,t},\\
 \mathsf{Z}(\bfe^\alpha\bfm^\beta)&=1/N\sum_{r,~s\in \mathbb{Z}_N}\mathrm{e}^{-2\pi\mathrm{i}t\alpha/N}\mathsf{Z}_{r, t}\delta_{\beta,r}.
\end{aligned}
\end{equation}

So far, we have recovered the partition functions $\mathsf{Z}(\bfe^{\alpha}\bfm^{\beta})$ directly from the lattice model of critical point. These partition functions Eq.~\eqref{eq:partition_function_p} all agree with the result Eq.~\eqref{eq:partition_function_m} from the mathematical theory, as expected. Therefore, the boundary model Eq.~\eqref{Hbdy} with Fateev-Zamolodchikov parameters indeed can be described by a gappable gapless boundary of $\bZ_N$ TO. 

In the following, we discuss the behavior of these partition functions $\mathsf{Z}(\bfe^\alpha\bfm^\beta)$ under modular transformations. By 
\begin{equation}
\begin{aligned}
\chi_{\{\ell;~m\}}=&\mathcal{S}^{\widehat{\mathrm{SU}}(2)_N/\widehat{\mathrm{U}}(1)_{2N}}_{\{\ell;~m\},\left\{\ell^{\prime} ;~m^{\prime}\right\}}\chi_{\{\ell';~m'\}}\\
\chi_{\{\ell;~m\}}=&\mathcal{T}^{\widehat{\mathrm{SU}}(2)_N/\widehat{\mathrm{U}}(1)_{2N}}_{\{\ell;~m\},\left\{\ell^{\prime} ;~m^{\prime}\right\}}\chi_{\{\ell';~m'\}},
\end{aligned}
\end{equation}
 it is not difficult to verify that $\mathsf{Z}(\bfe^\alpha\bfm^\beta)$ satisfy the following equations under the modular transformation,
\begin{equation}
\begin{aligned}
\mathsf{Z}(\bfe^\alpha\bfm^\beta)=&\mathcal{S}^{\Zb_N}_{\{\alpha;~\beta\},\{\alpha^{\prime} ;~\beta^{\prime}\}}\mathsf{Z}(\bfe^{\alpha'}\bfm^{\beta'})\\
\mathsf{Z}(\bfe^\alpha\bfm^\beta)=&\mathcal{T}^{\Zb_N}_{\{\alpha;~\beta\},\{\alpha^{\prime} ;~\beta^{\prime}\}}\mathsf{Z}(\bfe^{\alpha'}\bfm^{\beta'})\\
\mathcal{S}^{\Zb_N}_{\{\alpha;~\beta\},\{\alpha^{\prime} ;~\beta^{\prime}\}}&=\frac{1}{N}\mathrm{e}^{\frac{2\pi\mathrm{i}}{N}(\alpha\beta'+\beta\alpha')}\\
\mathcal{T}^{\Zb_N}_{\{\alpha;~\beta\},\{\alpha^{\prime} ;~\beta^{\prime}\}}&=\mathrm{e}^{-\frac{2\pi\mathrm{i}}{N}(\alpha\beta)}\delta_{\alpha,\alpha'}\delta_{\beta,\beta'},
\end{aligned}
\end{equation}
where $\mathcal{S}^{\Zb_N}$ and $\mathcal{T}^{\Zb_N}$ are the modular $S$ and $T$ matrices of the $\Zb_N$ topological order. In contrast to the modular invariant partition function of the pure 1+1D parafermion CFT, the partition functions $\mathsf{Z}(\bfe^\alpha\bfm^\beta)$ is modular covariant. In fact, for some simple cases, the modular-covariant condition provides us with another algebraic number theoretical approach: we can obtain finite non-negative integer solutions of the partition functions using the additional constraints derived from the modular-covariant condition. This approach is called modular bootstrap \cite{CLY2018,KLPM2021,LS2021}.

To confirm the validity of our result, here we give a few examples of specific $N$.
\subsubsection{\texorpdfstring{$N=2$}{N=2}, Ising model}
For $N=2$ case, it is the well-known Ising model. Thus $0\leq \ell \leq 2$, $-1 \leq m \leq 2$ and $0 \leq \alpha,\beta\leq 1$. There are six distinct parafermion coset fields,

\begin{equation}
\begin{array}{ll}
\{0;~0\}=\{2;~2\}	&  h = 0~\text{mod}~1\\
\{0;~2\}=\{2;~0\}	&  h = \frac{1}{2}~\text{mod}~1\\
\{1;-1\}=\{1;-1\} 	&  h = \frac{1}{16}~\text{mod}~1.\\
\end{array}	
\end{equation}
We identify the $\bZ_2$ parafermion coset characters $\chi_{\{\ell;~m\}}$ with Virasoro characters of Ising CFT: $\{\chi^{\text{Is}}_{{\mathbb{1}}}, \chi^{\text{Is}}_{\psi}, \chi^{\text{Is}}_{\sigma}\}$,
\begin{equation}
    \begin{array}{l}
        \chi_{\{0;~0\}}=\chi^{\text{Is}}_{\mathbb{1}}\\
        \chi_{\{0;~2\}}=\chi^{\text{Is}}_{\psi}\\
        \chi_{\{1;-1\}}=\chi^{\text{Is}}_{\sigma}.
    \end{array}
\end{equation}
And using Eq.~\eqref{eq:partition_function_p}, we can obtain the partition functions in the anyon basis $(\bfe^\alpha\bfm^\beta)$,
\begin{equation}
\begin{aligned}
    &  \mathsf{Z}(\mathbb{1})=\chi_{\{0;~0\}}\overline\chi_{\{0;~0\}}+\chi_{\{0;~2\}}\overline\chi_{\{0;~2\}}=\lvert\chi^{\text{Is}}_{\mathbb{1}}\rvert^2+\lvert\chi^{\text{Is}}_{\psi}\rvert^2\\
    &  \mathsf{Z}(\bfe)=\chi_{\{1;-1\}}\overline\chi_{\{1;-1\}}=\lvert\chi^{\text{Is}}_{\sigma}\rvert^2\\
    &  \mathsf{Z}(\bfm)=\chi_{\{1;-1\}}\overline\chi_{\{1;-1\}}=\lvert\chi^{\text{Is}}_{\sigma}\rvert^2\\
    &  \mathsf{Z}(\bfe\bfm)=\chi_{\{0;~0\}}\overline\chi_{\{2;~0\}}+\chi_{\{0;~2\}}\overline\chi_{\{0;~0\}}=\chi^{\text{Is}}_{\mathbb{1}}\overline\chi^{\text{Is}}_{\psi}+\chi^{\text{Is}}_{\psi}\overline\chi^{\text{Is}}_{\mathbb{1}}.
\end{aligned}
\end{equation}
This recovers the result in \cite{CJKYZ2020}.
\subsubsection{\texorpdfstring{$N=3$}{N=3}, Three potts model}
For $N=3$ case, $\bZ_3$ parafermion CFT can be realized by the so-called three Potts model. Thus $0\leq \ell \leq 3$, $-2 \leq m \leq 3$ and $0 \leq \alpha,\beta\leq 1$. There are six distinct parafermion coset fields,
\begin{equation}
\begin{array}{ll}
\{0;~0\}=\{3;~3\}	&  h = 0~\text{mod}~1\\
\{0;-2\}=\{3;~1\}	&  h = \frac{2}{3}~\text{mod}~1\\
\{0;~2\}=\{3;-1\}	&  h = \frac{2}{3}~\text{mod}~1\\
\{1;~3\}=\{2;~0\}	&  h = \frac{2}{5}~\text{mod}~1\\
\{1;~1\}=\{2;-2\}	&  h = \frac{1}{15}~\text{mod}~1\\
\{1;-1\}=\{2;~2\}	&  h = \frac{1}{15}~\text{mod}~1.\\
\end{array}	
\end{equation}
We identify the $\bZ_3$ parafermion coset characters $\chi_{\{\ell;~m\}}$ with Virasoro characters of the minimal model: $M(6,5)$ $\{\chi^{\text{(6,5)}}_{{\mathbb{1}}}, \chi^{\text{(6,5)}}_{{\mathbb{\psi}}}, \chi^{\text{(6,5)}}_{{\mathbb{\varepsilon}}}, \chi^{\text{(6,5)}}_{X},  \chi^{\text{(6,5)}}_{{Y}},  \chi^{\text{(6,5)}}_{{Z}}\}$,
\begin{equation}
\begin{array}{l}
    \chi_{\{0;~0\}}=\chi^{(6,5)}_{\mathbb{1}}+\chi^{(6,5)}_{Y}\\
	\chi_{\{1;~3\}}=\chi^{(6,5)}_{\varepsilon}+\chi^{(6,5)} _{X}\\
    \chi_{\{0;-2\}}=\chi_{\{0;~2\}}=\chi^{(6,5)}_{Z}\\
	\chi_{\{1;-1\}}=\chi_{\{1;~1\}}=\chi^{(6,5)}_{\sigma}.	
\end{array}
\end{equation}

And using Eq.~\eqref{eq:partition_function_p}, we can obtain the partition functions in the anyon basis $(\bfe^\alpha\bfm^\beta)$, 
\begin{align}
&\nonumber  \mathsf{Z}(\mathbb{1})=\chi_{\{0;~0\}}\overline\chi_{\{0;~0\}}+\chi_{\{1;~3\}}\overline\chi_{\{1;~3\}}=\lvert\chi^{(6,5)}_{\mathbb{1}}+\chi^{(6,5)}_{Y}\rvert^2+\lvert\chi^{(6,5)}_{\varepsilon}+\chi^{(6,5)}_{X}\rvert^2\\
&\nonumber  \mathsf{Z}(\bfe)=\chi_{\{0;-2\}}\overline\chi_{\{0;-2\}}+\chi_{\{1;~1\}}\overline\chi_{\{1;~1\}}=\lvert\chi^{(6,5)}_{Z}\rvert^2+\lvert\chi^{(6,5)}_{\sigma}\rvert^2\\
&\nonumber  \mathsf{Z}(\bfe^2)=\chi_{\{0;~2\}}\overline\chi_{\{0;~2\}}+\chi_{\{1;-1\}}\overline\chi_{\{1;-1\}}=\lvert\chi^{(6,5)}_{Z}\rvert^2+\lvert\chi^{(6,5)}_{\sigma}\rvert^2\\
&\nonumber  \mathsf{Z}(\bfm)=\chi_{\{0;-2\}}\overline\chi_{\{0;~2\}}+\chi_{\{1;~1\}}\overline\chi_{\{1;-1\}}=\lvert\chi^{(6,5)}_{Z}\rvert^2+\lvert\chi^{(6,5)}_{\sigma}\rvert^2\\
&\nonumber  \mathsf{Z}(\bfe\bfm)=\chi_{\{0;~2\}}\overline\chi_{\{0;~0\}}+\chi_{\{1;-1\}}\overline\chi_{\{1;~3\}}=\chi^{(6,5)}_{Z}\overline\chi^{(6,5)}_{\mathbb{1}}+\chi^{(6,5)}_{Z}\overline\chi^{(6,5)}_{Y}+\chi^{(6,5)}_{\sigma}\overline\chi^{(6,5)}_{\varepsilon}+\chi^{(6,5)}_{\sigma}\overline\chi^{(6,5)}_{X}\\
&\nonumber  \mathsf{Z}(\bfe^2\bfm)=\chi_{\{0;~0\}}\overline\chi_{\{0;-2\}}+\chi_{\{1;~3\}}\overline\chi_{\{1;~1\}}=\chi^{(6,5)}_{\mathbb{1}}\overline\chi^{(6,5)}_{Z}+\chi^{(6,5)}_{Y}\overline\chi^{(6,5)}_{Z}+\chi^{(6,5)}_{\varepsilon}\overline\chi^{(6,5)}_{\sigma}+\chi^{(6,5)}_{X}\overline\chi^{(6,5)}_{\sigma}\\
&\nonumber  \mathsf{Z}(\bfm^2)=\chi^{(6,5)}_{\{0;~2\}}\overline\chi^{(6,5)}_{\{0;-2\}}+\chi^{(6,5)}_{\{1;-1\}}\overline\chi^{(6,5)}_{\{1;~1\}}=\lvert\chi^{(6,5)}_{Z}\rvert^2+\lvert\chi^{(6,5)}_{\sigma}\rvert^2\\
&\nonumber  \mathsf{Z}(\bfe\bfm^2)=\chi_{\{0;~0\}}\overline\chi_{\{0;~2\}}+\chi_{\{1;~3\}}\overline\chi_{\{1;-1\}}=\chi^{(6,5)}_{\mathbb{1}}\overline\chi^{(6,5)}_{Z}+\chi^{(6,5)}_{Y}\overline\chi^{(6,5)}_{Z}+\chi^{(6,5)}_{\varepsilon}\overline\chi^{(6,5)}_{\sigma}+\chi^{(6,5)}_{X}\overline\chi^{(6,5)}_{\sigma}\\
&  \mathsf{Z}(\bfe^2\bfm^2)=\chi_{\{0;-2\}}\overline\chi_{\{0;~0\}}+\chi_{\{1;~1\}}\overline\chi_{\{1;~3\}}=\chi^{(6,5)}_{Z}\overline\chi^{(6,5)}_{\mathbb{1}}+\chi^{(6,5)}_Z\overline\chi^{(6,5)}_{Y}+\chi^{(6,5)}_{\sigma}\overline\chi^{(6,5)}_{\varepsilon}+\chi^{(6,5)}_{\sigma}\overline\chi^{(6,5)}_{X}.
\end{align}

\subsubsection{\texorpdfstring{$N=4$}{N=4}, \texorpdfstring{$c=1$}{c=1} \texorpdfstring{$\bZ_2$}{Z 2} orbifold CFT with radius \texorpdfstring{$R_{\text{orbifold}}=\sqrt{3/2}$}{R=sqrt3/2}}
$\bZ_4$ parafermion CFT is the $\bZ_2$ orbifold CFT with radius $R_{\text{orbifold}}=\sqrt{3/2}$. For $N=4$ case, $0\leq \ell \leq 4$, $-3 \leq m \leq 4$ and $0 \leq \alpha,\beta\leq 3$. There are ten distinct parafermion coset fields,
\begin{equation}
\begin{array}{ll}
\{0;~0\}=\{4;~4\}	&  h = 0~\text{mod}~1\\
\{0;-2\}=\{4;~2\}	&  h = \frac{3}{4}~\text{mod}~1\\
\{0;~2\}=\{4;-2\}	&  h = \frac{3}{4}~\text{mod}~1\\
\{0;~4\}=\{4;~0\}	&  h = 1~\text{mod}~1\\
\{1;-3\}=\{3;~1\}	&  h = \frac{9}{16}~\text{mod}~1\\
\{1;-1\}=\{3;~3\}	&  h = \frac{1}{16}~\text{mod}~1\\
\{1;~1\}=\{3;-3\}	&  h = \frac{1}{16}~\text{mod}~1\\
\{1;~3\}=\{3;-1\}	&  h = \frac{9}{16}~\text{mod}~1\\
\{2;-2\}=\{2;~2\}	&  h = \frac{1}{12}~\text{mod}~1\\
\{2;~0\}=\{2;~4\}	&  h = \frac{1}{3}~\text{mod}~1.\\
\end{array}	
\end{equation}

And using Eq.~\eqref{eq:partition_function_p}, we can obtain the partition functions in the anyon basis $(\bfe^\alpha\bfm^\beta)$, 
\begin{align}
&\nonumber  \mathsf{Z}(\mathbb{1})=\chi_{\{0;~0\}}\overline\chi_{\{0;~0\}}+\chi_{\{0;~4\}}\overline\chi_{\{0;~4\}}+\chi_{\{2;~0\}}\overline\chi_{\{2;~0\}}\\
&\nonumber  \mathsf{Z}(\bfe)=\chi_{\{1;-3\}}\overline\chi_{\{1;-3\}}+\chi_{\{1;~1\}}\overline\chi_{\{1;~1\}}\\
&\nonumber  \mathsf{Z}(\bfe^2)=\chi_{\{0;-2\}}\overline\chi_{\{0;-2\}}+\chi_{\{0;~2\}}\overline\chi_{\{0;~2\}}+\chi_{\{2;-2\}}\overline\chi_{\{2;-2\}}\\
&\nonumber  \mathsf{Z}(\bfe^3)=\chi_{\{1;-1\}}\overline\chi_{\{1;-1\}}+\chi_{\{1;~3\}}\overline\chi_{\{1;~3\}}\\
&\nonumber  \mathsf{Z}(\bfm)=\chi_{\{1;-3\}}\overline\chi_{\{1;~3\}}+\chi_{\{1;~1\}}\overline\chi_{\{1;-1\}}\\
&\nonumber  \mathsf{Z}(\bfe\bfm)=\chi_{\{0;-2\}}\overline\chi_{\{0;~4\}}+\chi_{\{0;~2\}}\overline\chi_{\{0;~0\}}+\chi_{\{2;-2\}}\overline\chi_{\{2;~0\}}\\
&\nonumber  \mathsf{Z}(\bfe^2\bfm)=\chi_{\{1;-1\}}\overline\chi_{\{1;-3\}}+\chi_{\{1;~3\}}\overline\chi_{\{1;~1\}}\\
&\nonumber  \mathsf{Z}(\bfe^3\bfm)=\chi_{\{0;~0\}}\overline\chi_{\{0;-2\}}+\chi_{\{0;~4\}}\overline\chi_{\{0;~2\}}+\chi_{\{2;~0\}}\overline\chi_{\{2;-2\}}\\
&\nonumber  \mathsf{Z}(\bfm^2)=\chi_{\{0;-2\}}\overline\chi_{\{0;~2\}}+\chi_{\{0;~2\}}\overline\chi_{\{0;-2\}}+\chi_{\{2;-2\}}\overline\chi_{\{2;-2\}}\\
&\nonumber  \mathsf{Z}(\bfe\bfm^2)=\chi_{\{1;-1\}}\overline\chi_{\{1;~3\}}+\chi_{\{1;~3\}}\overline\chi_{\{1;-1\}}\\
&\nonumber  \mathsf{Z}(\bfe^2\bfm^2)=\chi_{\{0;~0\}}\overline\chi_{\{0;~4\}}+\chi_{\{0;~4\}}\overline\chi_{\{0;~0\}}+
\chi_{\{2;~0\}}\overline\chi_{\{2;~0\}}\\
&\nonumber  \mathsf{Z}(\bfe^3\bfm^2)=\chi_{\{1;-3\}}\overline\chi_{\{1;~1\}}+\chi_{\{1;~1\}}\overline\chi_{\{1;-3\}}\\
&\nonumber  \mathsf{Z}(\bfm^3)=\chi_{\{1;-1\}}\overline\chi_{\{1;~1\}}+\chi_{\{1;~3\}}\overline\chi_{\{1;-3\}}\\
&\nonumber  \mathsf{Z}(\bfe\bfm^3)=\chi_{\{0;~0\}}\overline\chi_{\{0;~2\}}+\chi_{\{0;~4\}}\overline\chi_{\{0;-2\}}+\chi_{\{2;~0\}}\overline\chi_{\{2;-2\}}\\
&\nonumber  \mathsf{Z}(\bfe^2\bfm^3)=\chi_{\{1;-3\}}\overline\chi_{\{1;-1\}}+\chi_{\{1;~1\}}\overline\chi_{\{1;~3\}}\\
&\nonumber  \mathsf{Z}(\bfe^3\bfm^3)=\chi_{\{0;-2\}}\overline\chi_{\{0;~0\}}+\chi_{\{0;~2\}}\overline\chi_{\{0;~4\}}+\chi_{\{2;-2\}}\overline\chi_{\{2;~0\}}\\\end{align}.

\section{Conclusion and Discussions}
In summary, we study a pure boundary phase transition of two specific gapped boundaries, $\rep(\bZ_N)$ and $\rVec_{\bZ_N}$ of the $\bZ_N$ TO in this paper. 
The critical point of this boundary phase transition corresponds to a gappable non-chiral gapless boundary, which is mathematically described by the pair
\begin{align}
    (V_{\mathrm{PF}_N}\otimes_{\bC}\overline{V_{\mathrm{PF}_N}},\,^{\fZ_1(\mathrm{PF}_N)}\fZ_1(\mathrm{PF}_N)_B).
\end{align}
The $\bZ_N$ parafermion FFA $V_{\mathrm{PF}_N}\otimes_{\bC}\overline{V_{\mathrm{PF}_N}}$ encodes the local quantum symmetry and the enriched fusion category $^{\fZ_1(\mathrm{PF}_N)}\fZ_1(\mathrm{PF}_N)_B$ encodes the information of topological defects. 
All the ingredients of the enriched fusion category $^{\fZ_1(\mathrm{PF}_N)}\fZ_1(\mathrm{PF}_N)_B$, i.e., $M_{x,y}$ Eq.~\eqref{eq:expansion} match up with the low-energy effective theory Eq.~\eqref{eq:partition_function_p} based on the lattice model construction Eq.~\eqref{Hbdy}. As an example, this work shows that the mathematical theory \cite{KZ2020,KZ2021} of gapless edges of 2d topological orders is effective in studying general phase transitions.

However, for a general $N$, there are more than two kinds of gapped boundaries of $\bZ_N$ topological order, unlike the $\bZ_2$ case \cite{bravyi1998,KK2012,beigi2011}. They are described by the UFCs $\rep(\bZ_N)_{F_H\mid F_H}$ for different subgroups $H$ of $\bZ_N$, and $F_H:=\mathrm{Fun}(\bZ_N/H)$ is a function algebra in $\rep(\bZ_N)$ \cite{beigi2011}. In Laudau's paradigm, these gapped boundaries can be viewed as the partially symmetry broken phase. Partial-symmetry-breaking phase transitions may occur at the boundary. Indeed, the concrete lattice model Eq.~\eqref{Hbdy} that realizes these boundaries also has a rich phase diagram and includes many different critical points. We believe that all these phase transitions can also be described by the gappable non-chiral gapless boundaries. The methods and results in this paper can be further generalized to the other boundary phase transitions of the $\bZ_N$ TO. 

After this paper appeared on arXiv, we are aware of a few recent related works \cite{chatterjee2022emergent,pace2023}.



\bigskip
\noindent {\bf Acknowledgement}: 
We would like to thank Zhi-Hao Zhang, Liang Kong and Wei-Qiang Chen for helpful discussions.
YL is supported by NFSC (No.~12174164, No.~91963201). 
HY is supported by NSFC under Grant No. 11971219 and Guangdong Basic and Applied Basic Research Foundation under Grant No. 2020B1515120100 and Guangdong Provincial Key Laboratory (Grant No.2019B1212e03002).

\appendix
\appendixpage

\section{Derivation of category \texorpdfstring{$\mathrm{PF}_N$}{PF N}}\label{sec:math_PFN}
In this section, we first review the parafermion VOA, then, we give a detailed description of $\mathrm{PF}_N$, the category of modules over the parafermion VOA.

\subsection{Parafermion vertex operator algebras}
It is known that a rational CFT can be described by a rational VOA and its modules \cite{MS1989}.
Let $V$ be a rational VOA, the category of $V$-modules is a MTC \cite{Huang2008_modular}.
The parafermion CFT \cite{ZF1985} is described by the so called the parafermion VOA \cite{DL2012,DW2011}.

The parafermion CFT can be given by the following coset construction \cite{goddard1986,francesco2012}
\begin{align*}
    \dfrac{\widehat{\mathrm{SU}}(2)_N}{\widehat{\mathrm{U}}(1)_{2N}},
\end{align*}
where $\widehat{\mathrm{SU}}(2)_N$ is the affine Lie algebra $\widehat{\mathfrak{sl}}(2)$ at level $N$.
For readers in physical background, please refer to Appendix.~\ref{coset}.

It is known that for $\widehat{\mathfrak{sl}}(2)$, its Weyl module $V_{\widehat{\mathfrak{sl}}(2)_N}$ has a VOA structure with central charge $c_{\widehat{\mathfrak{sl}}(2)_N}=\frac{3N}{N+2}$ \cite{LL2004}.
It is worth mentioning that $V_{\widehat{\mathfrak{sl}}(2)_N}$ has a sub-VOA $V_{H}$ with a different central charge, called the Heisenberg VOA.
The central charge of $V_{H}$ is $1$.
There is a theorem \cite{FZ1992} tells us that the commutant $C_{V_{\widehat{\mathfrak{sl}}(2)_N}}(V_H)$ also has a VOA structure with the central charge
\begin{align*}
    c_{PF}=\dfrac{3N}{N+2}-1=\dfrac{2(N-1)}{N+2}.
\end{align*}

We are close enough to get the parafermion VOA.
Technically, after quotient the commutant $C_{V_{\widehat{\mathfrak{sl}}(2)}}(V_H)$ by the maximal ideal, we obtain a new VOA $V_{\mathrm{PF}_N}$, called the $parafermion$ $VOA$.
Unfortunately, the category of modules of the Heisenberg VOA is not finite \cite{LL2004}. However, there is another lattice VOA $V_L$ whose module category is finite and $C_{V_{\widehat{\mathfrak{sl}}(2)_N}}(V_L)=V_{\mathrm{PF}_N}$ \cite{ALY2019}.
Therefore, we can also obtain the parafermion VOA by replacing $V_L$ with $V_H$.


\subsection{UMTC \texorpdfstring{$\mathrm{PF}_N$}{PF N}}
It was proved that $V_{\mathrm{PF}_N}$ is a rational VOA \cite{DL2017}, thus the category $\mathrm{Mod}_{V_{\mathrm{PF}_N}}$ is a MTC.
Since the affine VOA with integral level and Heisenberg VOA is unitary, and the centralizer of a unitary sub VOA in a unitary VOA is still unitary \cite{DL2014}, the parafermion VOA is unitary. 

We use $\rep_N(\widehat{\mathfrak{sl}}(2))$ to denote the representation category of the affine Lie algebra $\widehat{\mathfrak{sl}}(2)$ at level $N$.
The coset construction can be given by the category of local right $A$-modules $(\rep_N(\widehat{\mathfrak{sl}}(2))\boxtimes \overline{\mathrm{Mod}_{V_L}})_A^{loc}$ of the category $\rep_N(\hat{\mathfrak{sl}}(2))\boxtimes \overline{\mathrm{Mod}_{V_L}}$ for some condensable algebra $A$ \cite{FFRS2004}.
Here $\mathrm{Mod}_{V_L}$ is the category of modules over the lattice VOA $V_L$, $\boxtimes$ is the Deligne tensor product whose physical meaning is stacking two topological orders together.
In order to write the category $\mathrm{PF}_N$ more explicitly,
we begin with the category $\rep_N(\hat{\mathfrak{sl}}(2))$.

From \cite{DNO2013}, the category $\rep_N(\hat{\mathfrak{sl}}(2))$ has the following data: 
\begin{itemize}
    \item simple objects: $\mathfrak{U}_\ell$, $0\leq \ell\leq N$,
    \item the fusion rules of simple objects:
    \begin{align*}
        \fU_\ell\otimes \fU_{\ell'}=\bigoplus^{\min(\ell,\ell')}_{r=\max(\ell+\ell'-N,0)}\fU_{\ell+\ell'-2r},
    \end{align*}
    \item the double braiding of $\fU_{\ell}$ and $\fU_{\ell'}$ is:
    \begin{align*}
        c_{\ell,\ell'}\circ c_{\ell',\ell}=\bigoplus_{s=\max(\ell'+\ell-N,0)}^{\min(\ell',\ell)}\mathrm{e}^{2\pi \mathrm{i}(h_{\ell'+\ell-2s}-h_{\ell}-h_{\ell'})},
    \end{align*}
    where $h_t=\frac{t(t+2)}{4(N+2)}$.
    \item the twist of simple objects:
    \begin{align}
        \theta_{\ell}=\mathrm{e}^{\ell(\ell+2)\frac{\pi\mathrm{i}}{2(N+2)}},
    \end{align}
\end{itemize}
The quantum dimension of simple objects are
    \begin{align*}
        \mathrm{dim}(\fU_{\ell})=\dfrac{q^{\ell+1}-q^{-\ell-1}}{q-q^{-1}},
    \end{align*}
    where $q=\mathrm{e}^{\frac{\pi\mathrm{i}}{N+2}}$.

The category $\mathrm{Mod}_{V_L}$ has the following data
\begin{itemize}
    \item simple objects: $W_m$, $0\leq m<2N$;
    \item the fusion rule of simple objects:
    \begin{align*}
        W_m\otimes W_{m'}\simeq W_{m+m'};
    \end{align*}
    \item the double braiding of simple objects $W_m$ and $W_{m'}$:
    \begin{align}
        c_{m',m}\circ c_{m,m'}=\mathrm{e}^{\frac{\pi\mathrm{i}}{N}mm'};
    \end{align}
    \item the twist of simple objects:
    \begin{align*}
        \theta_m=\mathrm{e}^{m^2\frac{\pi\mathrm{i}}{2N}}.
    \end{align*}
\end{itemize}

The condensable algebras in $\rep_N(\hat{\mathfrak{sl}}(2))\boxtimes\overline{\mathrm{Mod}_{V_L}}$ are classified by the quintuple $(A_1,A_2,\CC_1,\CD_1,\phi)$ \cite{DNO2013}, where $A_1$ is a condensable algebra in $\rep_N(\hat{\mathfrak{sl}}(2))$, $A_2$ is a condensable algebra in $\mathrm{Mod}_{V_L}$, $\CC_1$ is a fusion subcategory of $(\rep_N(\hat{\mathfrak{sl}}(2)))_{A_1}^{loc}$, $\CD_1$ is a fusion subcategory of $(\mathrm{Mod}_{V_L})_{A_2}^{loc}$ and $\phi$ is a braided equivalence between $\CD_1$ and $\overline{\CC_1}$.
If we find such a quintuple, we can write the corresponding condensable algebra explicitly:
\begin{align}
    A=\otimes^R(1)=\bigoplus_{x\in\mathrm{Irr}(\CC_1)}x^*\boxtimes x,
\end{align}
where $\otimes^R$ is the right adjoint of the tensor functor $\otimes:\CC_1\boxtimes \overline{\CC_1}\to \CC_1$.

Let $\CC_1$ be the fusion subcategory of $\rep_N(\hat{\mathfrak{sl}}(2))$ generated by the simple objects $\fU_0$ and $\fU_N$ and let $\CD_1$ be the fusion subcategory of $\overline{\mathrm{Mod}_{V_L}}$ generated by the simple objects $\overline{W_0}$ and $\overline{W_N}$.
Notice that the braiding of $\overline{W_N}$ is $\mathrm{e}^{-\frac{N\pi\mathrm{i}}{2}}$ and the braiding of $\fU_N$ is $\mathrm{e}^{\frac{-N\pi\mathrm{i}}{2}}$, then there is a natural braided equivalence $\phi: \CC_1\to\overline{\CD_1}$ which gives a condensable algebra
\begin{align}
    A=(\fU_0\boxtimes \overline{W_0})\oplus (\fU_N\boxtimes \overline{W_N}).
\end{align}

Now we need to find all local right $A$-modules of $\rep_N(\hat{\mathfrak{sl}}(2))\boxtimes\mathrm{Mod}_{V_L}$.
First we examine all so-called free $A$-modules $x\otimes A$.
Consider the adjunction
  \begin{align*}
    \hom_{(\rep_N(\hat{\mathfrak{sl}}(2))\boxtimes\mathrm{Mod}_{V_L})_A}(x\otimes A,y)\simeq \hom_{\rep_N(\hat{\mathfrak{sl}}(2))\boxtimes \mathrm{Mod}_{V_L}}(x, G(y)),
  \end{align*}
where $G$ is the forgetful functor.
It is not hard to see that all free modules are simple, and two free modules $(\fU_{\ell}\boxtimes \overline{W_m})\otimes A$ and $(\fU_{\ell'}\boxtimes \overline{W_{m'}})\otimes A$ are equivalent iff $\ell=\ell'$, $m=m'$ or $\ell=N-\ell'$, $m=N+m'$.
So there are $N(N+1)$ inequivalent simple objects in $(\rep_N(\hat{\mathfrak{sl}}(2))\otimes \overline{\mathrm{Mod}_{V_L}})_A$.

Let $\fM_{\ell,m}=(\fU_k\boxtimes \overline{W_m})\otimes A$ denote the simple objects in $(\rep_N(\hat{\mathfrak{sl}}(2))\boxtimes\overline{\mathrm{Mod}_{V_L}})_A$. 
We know that for two free modules $x\otimes A$ and $y\otimes A$, their relative tensor product $(x\otimes A)\otimes_A(y\otimes A)$ is isomorphic to $(x\otimes y)\otimes A$. 
Thus the fusion rule of $(\rep_N(\hat{\mathfrak{sl}}(2))\boxtimes\overline{\mathrm{Mod}_{V_L}})_A$ is actually the fusion rule of $\rep_N(\hat{\mathfrak{sl}}(2))\boxtimes\overline{\mathrm{Mod}_{V_L}}$, i.e.
\begin{align}
    \fM_{\ell,m}\otimes \fM_{\ell',m'}=\bigoplus\limits^{\min(\ell,\ell')}_{r=\max(\ell+\ell'-N,0)}\fM_{\ell+\ell'-2r,m+m'}.
\end{align}
Since simple objects of $\mathrm{PF}_N$ are the simple objects of $ (\rep_N(\hat{\mathfrak{sl}}(2))\boxtimes\overline{\mathrm{Mod}_{V_L}})_A$ that are local, we can say we have known the fusion rule of $\mathrm{PF}_N$, we just need to figure out all local free $A$-modules.
A free module $x\otimes A$ is local iff the double braiding $c_{A,x}\circ c_{x,A}$ is trivial.
The double braiding of $\rep_N(\hat{\mathfrak{sl}}(2))\boxtimes \overline{\mathrm{Mod}_{V_L}}$ is given by \cite{DNO2013}:
\begin{align}
    c_{\ell'm',\ell m}\circ c_{\ell m;\ell'm'}=\bigoplus_{s=\max(\ell+\ell'-N,0)}^{\min(\ell,\ell')}\mathrm{e}^{2\pi \mathrm{i}(h_{\ell+\ell'-2s}-h_{\ell}-h_{\ell'})}\mathrm{e}^{-\frac{\pi\mathrm{i}}{N}mm'}\id_{\ell+\ell'-2s,m+m'},
\end{align}
where $h_t:=\frac{t(t+2)}{4(N+2)}$.
Notice that the double braiding of $\fU_{\ell}\boxtimes \overline{W_m}$ with $A$ is trivial iff $\ell+m$ is even.
Then we find all simple objects $\fM_{\ell,m}$ of $\mathrm{PF}_N$ labeled by $0\leq \ell\leq N$, $0\leq m\leq 2N-1$ and $\ell+m\equiv 0 \mod 2$.
There are $N(N+1)/2$ inequivalent simple objects in $\mathrm{PF}_N$.

\section{Anyon condensation theory}\label{sec:any_cond_maths}
We give a mathematical description of anyon condensation theory \cite{Kong2014,KO2002}.
\begin{dfn}
    Let $\CC$ be a UMTC, an $algebra$ $A$ in $\CC$ is an object equipped with two morphisms $m:A\otimes A\to A$ and $\iota: 1_{\CC}\to A$ satisfying
    \begin{align*}
        m\circ(m\otimes\id_A)=m\circ(\id_A\otimes m),\\
        m\circ(\iota\otimes\id_A)=\id_A=m\circ(\id_A\otimes\iota).
    \end{align*}
    The algebra is called $commutative$ if $m=m\circ c_{A,A}$.
\end{dfn}

\begin{dfn}
    An algebra $A$ is called 
    \begin{itemize}
        \item $separable$ if $m:A\otimes A\to A$ splits as a $A$-$A$-bimodule homomorphism.
        \item $connected$ if $\dim \hom_{\CC}(1_{\CC},A)=1$;
        \item $condensable$ if $A$ is commutative connected separable;
        \item $lagrangian$ if $A$ is condensable and $\dim(A)^2=\dim(\CC)$.
    \end{itemize}
\end{dfn}

\begin{rmk}
    The connectedness means the vacuum is unique.
\end{rmk}

\begin{dfn}
    Let $A$ be an algebra in $\CC$.
    A right $A$-$module$ in $\CC$ is a pair $(M,\mu_M)$, where $M$ is a object in $\CC$ and $\mu_M:M\otimes A\to M$ is a morphism in $\CC$ such that 
    \begin{align*}
        \mu_M\circ(\mu_M\otimes\id_A)=\mu_M\circ(\id_M\otimes m),\\
        \id_A=\mu_M\circ(\id_M\otimes\iota).
    \end{align*}
\end{dfn}

\begin{dfn}
    A right $A$-module $M$ in $\CC$ is called a $local$ $A$-module if $\mu_M\circ c_{A,M}\circ c_{M,A}=\mu_M$.
\end{dfn}

Let $A$ be a condensable algebra in a UMTC $\CC$, we denote the category of right $A$-modules in $\CC$ by $\CC_A$ and denote the category of local right $A$-modules in $\CC$ by $\CC_A^{loc}$.
Then we have the following theorem \cite{KO2002}\cite{BEK2000}:
\begin{thm}
    Let $\CC$ be a UMTC, $A$ be a condensable algebra in $\CC$.
    Then $\CC_A$ is a UFC and $\CC_A^{loc}$ is a UMTC. 
\end{thm}

\begin{thm}\label{thm:dimension}
    Let $\CC$ be a UMTC, $A$ be a condensable algebra in $\CC$.
    \begin{align*}
        \dim(\CC_A)=\dfrac{\dim(\CC)}{\dim(A)},\\
        \dim(\CC_A^{loc})=\dfrac{\dim(\CC)}{\dim(A)^2}.
    \end{align*}
\end{thm}

\begin{crl}
    Let $\CC$ be a UMTC, $A$ be a lagrangian algebra in $\CC$.
    Then $\CC_A^{loc}\simeq \mathrm{Hilb}$ where $\dim(\mathrm{Hilb})=1$.
\end{crl}

\section{Proof of equivalence for odd case}\label{equivalence}
\subsection{Pre-metric group and pointed braided fusion categories}\label{Sec:PBFC}
We first review the notion of a pre-metric group \cite{DGNO2010}.

\begin{dfn}
    Let $G$ be an abelian group.
    A \textit{quadratic form} on $G$ is a map $q:G\to\bC^{\times}$ such that $q(g)=q(g^{-1})$ and $b(g,h):=\dfrac{q(gh)}{q(g)q(h)}$ is a bicharacter.
    A quadratic form $q$ is \textit{non-degenerate} if the associated bicharacter is non-degenerate.
\end{dfn}

\begin{dfn}
    A $pre$-$metric$ $group$ is a pair $(G,q)$ where $G$ is a finite abelian group and $q:G\to\bC^{\times}$ is a quadratic form.
    A $metric$ $group$ is a pre-metric group such that $q$ is non-degenerate.
\end{dfn}

\begin{dfn}
    An $orthogonal$ $homomorphism$ between pre-metric groups $(G,q)$ and $(G',q')$ is a group homomorphism $f: G\to G'$ such that $q'\circ f=q$.
\end{dfn}

\begin{prp}
    Pre-metric groups and orthogonal homomorphisms form a category.
\end{prp}

Pre-metric groups can be used to describe $pointed$ $braided$ $fusion$ $categories$. 
\begin{dfn}
    A tensor category $\CC$ is $pointed$ if every simple object of 
    $\CC$ is invertible.
\end{dfn}
Let $\CC$ be a pointed braided fusion category.
The isomorphism class of simple objects of $\CC$ forms a finite abelian group $G$.
For $x\in G$, we define $q_{\CC}(x)=c_{x,x}$ where $c_{x,x}$ is the braiding of $x$ with itself.
  
\begin{lem}
      The function $q_{\CC}:G\to\bC^{\times}$ is a quadratic form.
\end{lem}
\begin{proof}
    See \cite{EGNO2016}.
\end{proof}
This lemma shows that each pointed braided fusion category gives a pre-metric group.
Conversely, we have
\begin{lem}
    For any pre-metric group $(G, q)$ there exists a unique up to a braided equivalence pointed braided fusion category $\CC(G, q)$ such that the group of isomorphism classes of simple objects is $G$ and the associated quadratic form is $q$
\end{lem}
And we can eventually prove that 
\begin{thm}[\cite{JS1993}]
    The category of pointed braided monoidal categories is equivalent to the category of pre-metric groups.
\end{thm}

\begin{expl}
    $\rep(\bZ_N)$ is a pointed symmetric fusion category.
    So there is a pre-metric group $(\bZ_N,q_0)$ correspond to $\rep(\bZ_N)$ where $q_0(a)=1$ for each $a\in\bZ_N$ since $c_{x,x}=\id_x$ for each $x\in\rep(\bZ_N)$.
    So this pre-metric group is degenerate.
\end{expl}

Moreover if the pointed braided fusion category $\CC$ has a twist such that it becomes a pre-modular category, the corresponding pre-metric group $(G,q)$ should also be equipped with a character $\chi: G\to\bk^{\times}$ to capture the information of twist \cite{DGNO2010}.
The twist $\theta$ and character is related by the following formula
\begin{align}
    \theta(x)=q(x)\chi(x).
\end{align}
We call the triple $(G,q,\chi)$ the $twisted$ $pre$-$metric$ $group$ and denote the corresponding pointed pre-modular category as $\CC(G,q,\chi)$.
There is also a bijection between isomorphism classes of twisted pre-metric group and pointed pre-modular categories up to pre-modular equivalence.

When a pointed pre-modular category is non-degenerate, that is, a pointed modular tensor category, the corresponding twisted pre-metric group is a twisted metric group.
We can prove that if there is an isomorphism between two twisted metric groups, then the corresponding pointed modular tensor categories are equivalent.

\begin{expl}
    It is clear $\fZ_1(\rep(\bZ_N))$ is pointed, i.e. $\cO_{\alpha,\beta}$ has inverse $\cO_{N-\alpha,N-\beta}$, and all the simple objects and fusion rules of $\fZ_1(\rep(\bZ_N))$ is described by the abelian group $\bZ_{N}\times \bZ_N$.
    As the consequence, the pointed modular tensor category $\fZ_1(\rep(\bZ_N))$ is described by a twisted metric group $(\bZ_N\times \bZ_N,q,\tau)$ where $q(\alpha,\beta)=\mathrm{e}^{\frac{2\pi\mathrm{i}}{N}\alpha\beta}$ is a non-degenerate quadratic form on $\bZ_N\times \bZ_N$ and $\tau(\alpha,\beta)=1$ is a character.
    The correspondence is listed as follows:
    \begin{itemize}
        \item $\cO_{\alpha,\beta}\sim (\alpha,\beta)\in \bZ_N\times \bZ_N$;
        \item fusion rule of $\fZ_1(\rep(\bZ_N))$ is the group multiplication of $\bZ_N\times \bZ_N$;
        \item the braiding $c$ corresponds to the quadratic form $q$.
        \item the twist $\theta$ corresponds to the character $\tau$.
    \end{itemize} 
\end{expl}

\subsection{The main theorem}
In this section, we prove that there is an equivalence of MTCs, $\fZ_1(\mathrm{PF}_N)_B^{loc}\simeq \fZ_1(\rep(\mathbb{Z}_N))$ for $N=2k+1$.

\begin{dfn}
    Let $\CC$ be a braided monoidal category, $\CD$ be its subcategory, the $centralizer$ $\fZ_2(\CD,\CC)$ of $\CD$ in ${\CC}$ consists of all objects of $\CC$ that commutes with each object in $\CD$, i.e. $\fZ_2(\CD,\CC):=\{y\in\CC\mid c_{x,y}\circ c_{y,x}=\id_{y\otimes x},\,\, \forall x\in \CD\}$.
\end{dfn}

\begin{dfn}
    Let $\CC$ be a braided monoidal category, then the \textit{M{\"u}ger} $center$ $\fZ_2(\CC)$ of $\CC$ is the centralizer $\fZ_2(\CC,\CC)$.
\end{dfn}

\begin{dfn}
    A braided fusion category $\CC$ is called $non$-$degenerate$ if its M{\"u}ger center is trivial, i.e. $\fZ_2(\CC)\simeq \rVec$.
\end{dfn}

Let $\CX$ denote the braided fusion category generated by $\fU_0, \fU_N\in\rep_N(\hat{\mathfrak{sl}}(2))$.
\begin{prp}
    For $N=2k+1$, $\CX$ is non-degenerate.
\end{prp}
\begin{proof}
    The only possibly non-trivial double braiding in $\CX$ is 
    \begin{align*}
        c_{N,N}\circ c_{N,N}=e^{2\pi\mathrm{i}\frac{N}{2}}\id_{0}.
    \end{align*}
    For $N=2k+1$, this double braiding equals to $-1$.
    Therefore the M{\"u}ger center of $\CX$ is equivalent to $\rVec$.
\end{proof}

\begin{prp}
    The centralizer of $\CX$ in $\rep_{N}(\hat{\mathfrak{sl}}(2))$ is the full subcategory  $\CR^{even}$ generated by simple objects $\fU_{2s}$, $s=0,\dots, k$. 
\end{prp}
\begin{proof}
    The double braiding of $\fU_{\ell}$ and $\fU_N$ in $\rep_N(\hat{\mathfrak{sl}}(2))$ is 
    \begin{align*}
        c_{\ell,N}\circ c_{N,\ell}=e^{2\pi (h_{N-\ell}-h_N-h_{\ell})}\id_{N-\ell}\\
        =(-1)^{\ell}\id_{N-\ell}.
    \end{align*}
    When $\ell=2s$, the double braiding is trivial.
\end{proof}

We need the following Centralizer Theorem.
\begin{thm}[\cite{Muger2003}]\label{thm:centralizer}
    Let $\CC$ be a MTC and $\CD$ be a modular tensor subcategory of $\CC$.
    Then there is an equivalence of ribbon categories: $\CC\simeq \CD\boxtimes \fZ_2(\CD,\CC)$. 
\end{thm}

\begin{crl}
    There is an equivalence of ribbon categories:
    \begin{align*}
        \mathrm{PF}_{2k+1}:=(\rep_{2k+1}(\hat{\mathfrak{sl}}(2))\boxtimes\overline{\mathrm{Mod}_{V_L}})_A^{loc}\simeq (\CR^{even}\boxtimes\CX\boxtimes\overline{\mathrm{Mod}_{V_L}})_A^{loc}.
    \end{align*}
\end{crl}

Under the equivalence $\rep_{2k+1}(\hat{\mathfrak{sl}}(2))\simeq \CR^{even}\boxtimes \CX$, we have $\fU_0\mapsto \fU_0\boxtimes \fU_0$ and $\fU_{2k+1}\mapsto \fU_0\boxtimes \fU_{2k+1}$.
Respectively, the condensable algebra $A=(\fU_0\boxtimes \overline{W_0})\oplus (\fU_{2k+1}\boxtimes \overline{W_{2k+1}})$ in $\rep_{2k+1}(\hat{\mathfrak{sl}}(2))\boxtimes\overline{\mathrm{Mod}_{V_L}}$ becomes
\begin{align*}
    A'&=(\fU_0\boxtimes \fU_0\boxtimes \overline{W_0})\oplus (\fU_0\boxtimes \fU_{2k+1}\boxtimes \overline{W_{2k+1}})\\
    &=\fU_0\boxtimes((\fU_0\boxtimes \overline{W_0})\oplus (\fU_{2k+1}\boxtimes \overline{W_{2k+1}}))\\
    &=\fU_0\boxtimes A_0,
\end{align*}
where $A_0:=(\fU_0\boxtimes \overline{W_0})\oplus (\fU_{2k+1}\boxtimes \overline{W_{2k+1}})$ is a condensable algebra in $\CX\boxtimes\overline{\mathrm{Mod}_{V_L}}$.
So we can rewrite $\mathrm{PF}_{2k+1}$ as
\begin{align*}
    \mathrm{PF}_{2k+1}\simeq (\CR^{even}\boxtimes(\CX\boxtimes\overline{\mathrm{Mod}_{V_L}}))_{\fU_0\boxtimes A_0}^{loc}\simeq (\CR^{even})_{\fU_0}^{loc}\boxtimes (\CX\boxtimes \overline{\mathrm{Mod}_{V_L}})_{A_0}^{loc}.
\end{align*}
It is clear that $(\CR^{even})_{\fU_0}^{loc}\simeq \CR^{even}$ because $\fU_0$ is the trivial condensable algebra (tensor unit) in $\CR^{even}$.
And the modular tensor category $(\CX\boxtimes \overline{\mathrm{Mod}_{V_L}})_{A_0}^{loc}$ is actually a pointed modular tensor category and thus can be described by a twisted metric group.
\begin{prp}
    There is an equivalence of modular tensor categories
    \begin{align*}
        (\CX\boxtimes\overline{\mathrm{Mod}_{V_L}})_{A_0}^{loc}\simeq \CC(\bZ_{2k+1},p,\xi).
    \end{align*}
    where $p(a)=\mathrm{exp}(-\frac{2\pi\mathrm{i} }{2k+1}a^2)$ and $\xi(a)=1$ for $a\in\bZ_{2k+1}$.
\end{prp}
\begin{proof}
    The simple objects in $\CX\boxtimes \mathrm{Mod}_{V_L}$ are $\fU_0\boxtimes \overline{W_m}$ and $\fU_{2k+1}\boxtimes \overline{W_{m'}}$.
    After a short calculation we can check $\{(\fU_0\boxtimes \overline{W_{2a}})\otimes A_0\mid a=0,\dots, 2k+1\}$ exhausts all simple objects of $(\CX\boxtimes \overline{\mathrm{Mod}_{V_L}})_{A_0}^{loc}$.
    Notice that simple objects of $(\CX\boxtimes \overline{\mathrm{Mod}_{V_L}})_{A_0}^{loc}$ are all free $A_0$-modules, by the discussion in Appendix \ref{sec:math_PFN}, the fusion rules of $(\CX\boxtimes \overline{\mathrm{Mod}_{V_L}})_{A_0}^{loc}$ are actually the fusion rules of $W_{2a}\in\mathrm{Mod}_{V_L}$.
    Then there is an equivalence $(\CX\boxtimes \overline{\mathrm{Mod}_{V_L}})_{A_0}^{loc}\simeq \CC(\mathbb{Z}_{2k+1},p,\xi)$.
    
\end{proof}

\begin{crl}
    There is an equivalence of modular tensor categories
    \begin{align*}
        \mathrm{PF}_{2k+1}&\simeq \CR^{even}\boxtimes \CC(\mathbb{Z}_{2k+1},p,\xi)\\
        \fM_{2\ell,2m}&\mapsto \fU_{2\ell}\boxtimes W_{2m}.
    \end{align*}
\end{crl}

Hence we can write $\fZ_1(\mathrm{PF}_{2k+1})\simeq \fZ_1(\CR^{even})\boxtimes \fZ_1(\CC(\mathbb{Z}_{2k+1},p,\xi))$.
Under this equivalence, the condensable algebra 
$B=\bigoplus_{u=0}^{k}\fM_{2u,0}\boxtimes\overline{\fM_{2u,0}}$ in $\fZ_1(\mathrm{PF}_{2k+1})$ becomes $(\bigoplus_{u=0}^{k} \fU_{2u}\boxtimes\overline{\fU_{2u}})\boxtimes (W_0\boxtimes \overline{W_0} )$.
where $B_0:=\bigoplus_{u=0}^{k} \fU_{2u}\boxtimes\overline{\fU_{2u}}$
is a condensable algebra in $\fZ_1(\CR^{even})$, and $W_0\boxtimes \overline{W_0}$ is the trivial condensable algebra in $\fZ_1(\CC(\mathbb{Z}_{2k+1},p,\xi))$.
So we have 
\begin{align}
\fZ_1(\mathrm{PF}_{2k+1})_B^{loc}\simeq \fZ_1(\CR^{even})_{B_0}^{loc}\boxtimes \fZ_1(\CC(\mathbb{Z}_{2k+1},p,\xi)).
\end{align}
By computing quantum dimensions, it is not hard to see that $B_0$ is lagrangian in $\fZ_1(\CR^{even})$, so $\fZ_1(\CR^{even})_{B_0}^{loc}\simeq \rVec$.
Note that the twisted pre-metric group $(\bZ_{2k+1},p,\xi)$ is non-degenerate, then $\fZ_1(\CC(\mathbb{Z}_{2k+1},p,\xi))\simeq \CC(\mathbb{Z}_{2k+1}\times \mathbb{Z}_{2k+1},p\times \overline{p},\xi\times\overline{\xi} )$,
\begin{prp}
    Then there is an equivalence of MTCs
\begin{align*}
    \fZ_1(\mathrm{PF}_{2k+1})_B^{loc}\simeq \CC(\mathbb{Z}_{2k+1}\times \mathbb{Z}_{2k+1},p\times \overline{p},\xi\times\overline{\xi} ).
\end{align*}
\end{prp}

Recall that in Appendix.~\ref{Sec:PBFC}, we have shown that $\fZ_1(\rep(\bZ_{2k+1}))$ is a pointed modular tensor category $\CC(\bZ_{2k+1}\times \bZ_{2k+1},q,\tau)$.
Thus if there exists an isomorphism between two twisted metric groups, the corresponding pointed modular tensor categories should be equivalent to each other.
\begin{thm}

    There is an equivalence of modular tensor categories
    \begin{align}
        \fZ_1(\mathrm{PF}_{2k+1})_B^{loc}\simeq\fZ_1(\rep(\mathbb{Z}_{2k+1})).
    \end{align}
\end{thm}
\begin{proof}
    For $(a,b)\in\mathbb{Z}_{2k+1}\times \mathbb{Z}_{2k+1}$, we have $(p\times \overline{p})(a,b)=\mathrm{exp}(\frac{2\pi\mathrm{i}}{2k+1}(a^2-b^2))$.
    We know that $\fZ_1(\rep(\mathbb{Z}_{2k+1}))\simeq \CC(\mathbb{Z}_{2k+1}\times \mathbb{Z}_{2k+1},q,\tau)$, where $q(a,b)=\mathrm{exp}(\frac{2\pi\mathrm{i}}{2k+1} ab)$.

    It is not hard to find there is a group automorphism $f:(a,b)\mapsto (a+b,a-b)$ on $\mathbb{Z}_{2k+1}\times \mathbb{Z}_{2k+1}$ such that $f$ is an isomorphism of pre-metric groups $(\mathbb{Z}_{2k+1}\times\mathbb{Z}_{2k+1},p\times \overline{p})$ and $(\mathbb{Z}_{2k+1}\times\mathbb{Z}_{2k+1},q)$.
    It is also clear that $\xi\times\overline{\xi}(a,b)=\tau(a+b,a-b)=\tau\circ f(a,b)$.

    Then by the equivalence theorem of category of pointed modular tensor categories and category of twisted metric groups, this isomorphism induces an equivalence between $\fZ_1(\mathrm{PF}_{2k+1})_B^{loc}$ with $\fZ_1(\rep(\mathbb{Z}_{2k+1}))$.
\end{proof}

\section{Preliminaries of \texorpdfstring{$\hat{\mathrm{g}}_k/\hat{\mathrm{p}}_{x_e k}$}{gkpxek} coset construction}\label{coset}
Coset construction of Goddard–Kent–Olive \cite{goddard1986} is a well-known method to construct rational
conformal field theories with known chiral data, such as conformal weights, fusion rules, braiding and fusing matrices. Given an affine Lie algebra $\hat{\mathrm{g}}$ together with a choice $k$ of levels, i.e. a positive integer for each simple ideal of $\hat{\mathrm{g}}$, and subalgebra $\hat{\mathrm{p}}$  of $\hat{\mathrm{g}}$  with  level $x_e k$, that determined by the embedding index $x_e$, we can show that $ L_{m}^{(\mathrm{g} / \mathrm{p})} =L_{m}^{\mathrm{g}}-L_{m}^{\mathrm{p}}$ also satisfies the Virasoro algebra and further construct the other chiral data, where $L_{m}^{\mathrm{g}}$ and $L_{m}^{\mathrm{p}}$ are the Virasoro modes of $\hat{\mathrm{g}}$ and $\hat{\mathrm{p}}$, respectively. In the following, we directly list the basic ingredients of $\hat{\mathrm{g}}_k/\hat{\mathrm{p}}_{x_e k}$ coset construction.For more details, please refer to. Coset Virasoro algebra:
 
 \begin{equation}
 \begin{aligned}
 L_{m}^{(\mathrm{g} / \mathrm{p})}&=L_{m}^{\mathrm{g}}-L_{m}^{\mathrm{p}} \\
 \left[L_{m}^{(\mathrm{g} / \mathrm{p})}, L_{n}^{(\mathrm{g}/ \mathrm{p})}\right] &=\left[L_{m}^{\mathrm{g}}, L_{n}^{\mathrm{g}}\right]-\left[L_{m}^{\mathrm{p}}, L_{n}^{\mathrm{p}}\right] \\
 &=(m-n) L_{m+n}^{\mathrm{g}/ \mathrm{p}}+\left(c\left(\hat{\mathrm{g}}_{k}\right)-c\left(\hat{\mathrm{p}}_{x_{e} k}\right)\right) \frac{\left(m^{3}-m\right)}{12} \delta_{m+n, 0}
 \end{aligned}
 \end{equation}
 
 coset central charge:
 \begin{equation}
 \begin{aligned}
 c\left(\hat{\mathrm{g}}_{k} / \hat{\mathrm{p}}_{x_{e} k}\right)&=c\left(\hat{\mathrm{g}}_{k}\right)-c\left(\hat{\mathrm{p}}_{x_{e} k}\right) \\
 &=\frac{k \operatorname{dim} \mathrm{g}}{k+g}-\frac{x_{e} k \operatorname{dim} \mathrm{p}}{x_{e} k+p}
 \end{aligned}
 \end{equation}
 where $g/p$ is the dual Coxeter number of group g/p.
 
 Because the various representation $\hat{\lambda}$ of $\hat{\mathrm{g}}$ can decompose into a direct sum of representations $\hat{\mu}$ of $\mathrm{\hat{p}}$
 \begin{equation}
 \hat{\lambda} \mapsto \bigoplus_{\hat{\mu}} b_{\hat{\lambda} \hat{\mu}} \hat{\mu},
 \end{equation}
 the corresponding coset characters can be obtained from this decomposition,
 \begin{equation}
 \chi_{\mathcal{P} \hat{\lambda}}=\sum_{\hat{\mu}} \chi_{\{\hat{\lambda};~\hat{\mu}\}} \chi_{\hat{\mu}}.
 \end{equation}
 
 Conformal weight of coset primary field,
 \begin{equation}
 h_{\{\hat{\lambda};~\hat{\mu}\}}=h_{\hat{\lambda}}-h_{\hat{\mu}}+n,
 \end{equation}
 where $n$ is the grade in the representation of $\mathrm{g}$ which the tip of $\hat{\mu}$ representation lies at.
 
 Modular transformation matrices of coset characters $ \chi_{\{\hat{\lambda};~\hat{\mu}\}}$,
 \begin{equation}
 \begin{array}{l}
 \mathcal{S}_{\{\hat{\lambda} ; ~\hat{\mu}\},\left\{\hat{\lambda}^{\prime} ; ~\hat{\mu}^{\prime}\right\}}=\mathcal{S}_{\hat{\lambda} \hat{\lambda}^{\prime}}^{(k)} \overline{\mathcal{S}}_{\hat{\mu} \hat{\mu}^{\prime}}^{\left(k x_{e}\right)} \\
 \mathcal{T}_{\{\hat{\lambda} ; ~\hat{\mu}\},\left(\hat{\lambda}^{\prime} ; ~\hat{\mu}^{\prime}\right\}}=\mathcal{T}_{\hat{\lambda} \hat{\lambda}^{\prime}}^{(k)} \overline{\mathcal{T}}_{\hat{\mu} \hat{\mu}^{\prime}}^{\left(k x_{e}\right)}.
 \end{array}
 \end{equation}
 
 Modular invariant partition function in the coset theory is simply to take, for the coset mass matrix $\mathcal{M}$, the product
 
 \begin{equation}
 \mathcal{M}=\mathcal{M}^{(k)}\mathcal{M}^{(x_ek)}.
 \end{equation}
 However, the branching conditions will impose constraints on the summations of partition function. Furthermore, field identifications must be considered to divide the partition function by $N$ (If all orbits of field identifications have length $N$). The partition function is then
 \begin{equation}
 \mathsf{Z}=\frac{1}{N} \sum_{\hat{\lambda}, \hat{\lambda}^{\prime} \in P_{+}^{(k)} \hat{\mu}, \hat{\mu}^{\prime} \in P_{+}^{(kx_e)} \atop \mathcal{P} \lambda-\mu=\mathcal{P}\lambda^{\prime}-\mu^{\prime}=0~\text{mod}~\mathcal{Q}} \chi_{\{ \hat{\lambda} ;~\hat{\mu}\}}(\tau) \mathcal{M}_{\hat{\lambda}, \hat{\lambda}^{\prime}}^{(k)} \mathcal{M}_{\hat{\mu}, \hat{\mu}^{\prime}}^{\left(k x_{e}\right)} \overline{\chi}_{\{\hat{\lambda}^{\prime} ;~\hat{\mu}^{\prime}\}}(\bar{\tau}),
 \end{equation}
 
 where $\mathcal{P}\lambda-\mu=\mathcal{P}\lambda^{\prime}-\mu^{\prime}=0~\text{mod}~\mathcal{Q}$ is the branching condition, $\mathcal{Q}$ is the root lattice of g.

\bibliographystyle{alpha}
\bibliography{Zngapless}

\end{document}